\documentclass[a4paper,onecolumn,superscriptaddress,accepted=2023-03-30,10pt]{quantumarticle}

\pdfoutput=1

\usepackage[T1]{fontenc} 
\usepackage[english]{babel} 
\usepackage[utf8]{inputenc}

\usepackage[margin=1.8cm]{geometry}

\usepackage[numbers,sort&compress]{natbib}

\providecommand{\pra}{Phys.\ Rev.\ A}

\providecommand{\prl}{Phys.\ Rev.\ Lett.}

\providecommand{\njp}{New J.\ Phys.}
\providecommand{\nat}{Nature}

\usepackage[table,usenames,dvipsnames]{xcolor}
\usepackage[
		colorlinks,
		linkcolor={black!30!blue},
		citecolor={black!30!blue},
		urlcolor={black!30!blue}
]{hyperref}
\usepackage{graphicx}
\usepackage{hyperxmp} 

\usepackage{amsmath,amssymb,mathtools,amsthm,dsfont}

\usepackage{enumitem}
\usepackage{caption}
\usepackage{subcaption}

\usepackage{optidef}
\usepackage{algorithm,algorithmicx}
\usepackage[noend]{algpseudocode}
\algrenewcommand\algorithmiccomment[1]{\hfill $\triangleright$ #1}
\usepackage{tikz}
\usetikzlibrary{quantikz,fit}
\let\braket\undefined
\let\ket\undefined
\tikzset{
  operator/.append style={ minimum width=1.8em,  }
} 

\usepackage{cleveref}
\crefname{equation}{Eq.}{Eqs.}
\Crefname{equation}{Equation}{Equations}
\crefname{section}{Section}{Sections}
\crefname{figure}{Figure}{Figures} 
\crefname{theorem}{Theorem}{Theorems}
\crefname{appendix}{Appendix}{Appendices}
\crefname{proposition}{Proposition}{Propositions}
\crefname{definition}{Definition}{Definitions}
\crefname{algorithm}{Algorithm}{Algorithms}

\usepackage[printonlyused,
						withpage,
						]{acronym}
\makeatletter
\AtBeginDocument{%
 \renewcommand*{\AC@hyperlink}[2]{%
   \begingroup
     \hypersetup{hidelinks}%
     \hyperlink{#1}{#2}%
   \endgroup
 }%
}
\makeatother

\usepackage[compat=0.3]{yquant}

\newtheorem{theorem}{Theorem}
\newtheorem{corollary}[theorem]{Corollary}
\newtheorem{lemma}[theorem]{Lemma}
\newtheorem{proposition}[theorem]{Proposition}

\newtheorem{definition}[theorem]{Definition}
\theoremstyle{definition}
\newtheorem*{example}{Example}

\newcommand{\e}{\ensuremath\mathrm{e}}
\renewcommand{\i}{\ensuremath\mathrm{i}}

\DeclareMathOperator{\LandauO}{\mathrm{O}}

\DeclareMathOperator{\Tr}{Tr}

\DeclareMathOperator{\supp}{supp}

\newcommand{\fro}{\mathrm{F}}

\DeclareMathOperator{\cone}{cone}

\DeclareMathOperator{\U}{U}

\DeclareMathOperator{\GZZ}{\mathrm{GZZ}}

\newcommand{\Gate}[1]{\operatorname{#1}}

\NewDocumentCommand\Cl{mg}{
    \ensuremath{\mathrm{Cl}_{#1}\IfNoValueTF{#2}{}{(#2)}}%
}
\NewDocumentCommand\HW{mg}{
    \ensuremath{\mathrm{HW}_{#1}\IfNoValueTF{#2}{}{(#2)}}%
}

\newcommand{\RR}{\mathbb{R}}

\newcommand{\1}{\mathds{1}}

\newcommand{\FF}{\mathbb{F}}

\newcommand{\R}{\mathbb{R}}

\newcommand{\mc}[1]{\mathcal{#1}}

\renewcommand{\vec}[1]{\boldsymbol{#1}}

\DeclarePairedDelimiterX{\abs}[1]{\lvert}{\rvert}{%
  \ifblank{#1}{\,\cdot\,}{#1}
}   

\DeclarePairedDelimiterX\norm[1]\lVert\rVert{%
  \ifblank{#1}{\,\cdot\,}{#1}
}   

\DeclarePairedDelimiterX{\iiiNorm}[1]{\lvert}{\rvert}{%
  \delimsize\lvert\delimsize\lvert#1\delimsize\rvert\delimsize\rvert%
}

\DeclarePairedDelimiterXPP\snorm[1]{}\lVert\rVert{_\infty}{\ifblank{#1}{\,\cdot\,}{#1}}   

\DeclarePairedDelimiterXPP\twonorm[1]{}\lVert\rVert{_2}{\ifblank{#1}{\,\cdot\,}{#1}}   

\DeclarePairedDelimiterXPP\trnorm[1]{}\lVert\rVert{_1}{\ifblank{#1}{\,\cdot\,}{#1}}   

\DeclarePairedDelimiterXPP\fnorm[1]{}\lVert\rVert{_{\fro}}{\ifblank{#1}{\,\cdot\,}{#1}}   

\DeclarePairedDelimiterXPP\dnorm[1]{}\lVert\rVert{_\diamond}{\ifblank{#1}{\,\cdot\,}{#1}}   

\DeclarePairedDelimiterXPP\cbnorm[1]{}\lVert\rVert{_\mathrm{cb}}{\ifblank{#1}{\,\cdot\,}{#1}}   
\DeclarePairedDelimiterXPP\onenorm[1]{}\lVert\rVert{_{1\rightarrow 1}}{\ifblank{#1}{\,\cdot\,}{#1}}   
\DeclarePairedDelimiterXPP\ddnorm[1]{}\lVert\rVert{_{\diamond\rightarrow \diamond}}{\ifblank{#1}{\,\cdot\,}{#1}}   
\DeclarePairedDelimiterXPP\ssnorm[1]{}\lVert\rVert{_{\infty\rightarrow\infty}}{\ifblank{#1}{\,\cdot\,}{#1}}   

\providecommand\given{}
\newcommand\SetSymbol[1][]{%
  \nonscript\:#1\vert
  \allowbreak
  \nonscript\:
  \mathopen{}}
\DeclarePairedDelimiterX\Set[1]\{\}{%
  \renewcommand\given{\SetSymbol[\delimsize]}
  #1
}

\DeclarePairedDelimiterX\innerp[2]{\langle}{\rangle}{%
  \ifblank{#1}{\,\cdot\,}{#1} , \ifblank{#2}{\,\cdot\,}{#2}%
}

\DeclarePairedDelimiter{\ket}{\vert}{\rangle}
\newcommand{\ketb}[1]{\ket[\big]{#1}}

\DeclarePairedDelimiterX\braket[2]{\langle}{\rangle}%
  {#1\kern0.15ex\delimsize\vert\kern0.15ex\mathopen{}#2}

\DeclarePairedDelimiterX\ketbra[2]{\vert}{\vert}%
  {#1\kern0.15ex\delimsize\rangle\delimsize\langle\kern0.15ex\mathopen{}#2}

\DeclarePairedDelimiterX\sandwich[3]{\langle}{\rangle}%
  {#1\,\delimsize\vert\kern0.15ex\mathopen{}#2\kern0.15ex\delimsize\vert\kern0.15ex\mathopen{}#3}

\DeclarePairedDelimiterX\obraket[2]{(}{)}%
  {#1\kern0.15ex\delimsize\vert\kern0.15ex\mathopen{}#2}

\DeclarePairedDelimiterX\oketbra[2]{\vert}{\vert}%
  {#1\kern0.15ex\delimsize)\delimsize(\kern0.15ex\mathopen{}#2}

\DeclarePairedDelimiterX\osandwich[3]{(}{)}%
  {#1\,\delimsize\vert\kern0.15ex\mathopen{}#2\kern0.15ex\delimsize\vert\kern0.15ex\mathopen{}#3}

\newcommand{\myleft}{\mathopen{}\mathclose\bgroup\left}
\newcommand{\myright}{\aftergroup\egroup\right}

\newcommand{\verbat}[1]{\text{\normalfont{\ttfamily{#1}}}}

\DeclareMathOperator{\HTL}{\mathrm{Sym}_0}

\renewcommand{\Gate}[1]{\operatorname{\mathsf{#1}}}
\numberwithin{theorem}{section}
\numberwithin{equation}{section}

\newcommand{\hhu}{
	Institute for Theoretical Physics,
	Heinrich-Heine-Universit{\"a}t D{\"u}sseldorf,
	Germany
}

\newcommand{\siegen}{
  Department of Physics, 
  School of Science and Technology, 
  University of Siegen, 
  Germany
}

\newcommand{\tuhh}{
	Institute for Quantum and Quantum Inspired Computing, 
	Hamburg University of Technology, 
	Germany
}

\title{Synthesis of and compilation with time-optimal multi-qubit gates}

\author{P.\ Baßler}
\email{bassler@hhu.de}
\affiliation{\hhu}
\author{M.\ Zipper}
\affiliation{\hhu}
\author{C.\ Cedzich}
\affiliation{\hhu}
\author{M.\ Heinrich}
\affiliation{\hhu}
\author{P.\ H.\ Huber}
\affiliation{\siegen}
\author{M.\ Johanning}
\affiliation{\siegen}
\author{M.\ Kliesch}
\email{science@mkliesch.eu}
\affiliation{\hhu}
\affiliation{\tuhh}

\begin{document}
\maketitle

\begin{abstract}
We develop a method to synthesize a class of entangling multi-qubit gates for a quantum computing platform with fixed Ising-type interaction with all-to-all connectivity. 
The only requirement on the flexibility of the interaction is that it can be switched on and off for individual qubits. 
Our method yields a time-optimal implementation of the multi-qubit gates. 
We numerically demonstrate that the total multi-qubit gate time scales approximately linear in the number of qubits. 
Using this gate synthesis as a subroutine, we provide compilation strategies for important use cases: 
(i) we show that any Clifford circuit on $n$ qubits can be implemented using at most $2n$ multi-qubit gates without requiring ancilla qubits,
(ii) we decompose the quantum Fourier transform in a similar fashion, 
(iii) we compile a simulation of molecular dynamics, and
(iv) we propose a method for the compilation of diagonal unitaries with time-optimal multi-qubit gates, 
as a step towards general unitaries. 
As motivation, we provide a detailed discussion on a microwave controlled ion trap architecture with \ac{MAGIC} for the generation of the Ising-type interactions. 
\end{abstract}

 \hypersetup{
	     pdftitle = {Synthesis of and compilation with time-optimal multi-qubit gates},
	     pdfauthor = {Pascal Baßler, 
	     		Matthias Zipper,
		     	Christopher Cedzich,
		     	Markus Heinrich,
	     		Patrick Heinz Huber,
	     		Michael Johanning,
	     		Martin Kliesch},
	     pdfsubject = {Quantum computing},
	     pdfkeywords = {quantum, compiling, 
	     	 global, interaction, entangling, multi-qubit, gate, synthesis, Clifford, diagonal, unitary, all-to-all, connectivity,
		     Fourier transform, QFT, 
		     microwave, ion trap, Ising, electronic, Hamiltonian, MAGIC,
		     convex, optimization, linear program, LP, MIP, mixed integer program, 
		     digital-analog, DAQC, 
		     qubit, 
		     molecular, dynamics, structure problem, 
		     Molmer-Sorensen, 
		     frames		     
	     }
	    }

\section{Introduction} 
In order to run a program on any computing platform, it is necessary to decompose its higher-level logical operations into more elementary ones and eventually translate those into the platform's native instruction set. 
This process is called ``compiling.''
Both for classical and quantum computers, this is a non-trivial task. 
The performance of the compiled program depends not only on the optimizations done by the compiler but also on the available instructions and their implementation.

Especially in the era of \ac{NISQ} devices, quantum algorithms are limited by the coherence time of the noisy qubits and the number of noisy gates needed to run them \cite{Pre18}.
Thus, it is imperative not only to improve the current quantum devices but also to design fast gates and optimized compilers that use the specific architecture's peculiarities to reduce the circuit depth.
Moreover, these endeavors help to reduce the noise-levels of physical gates and are thus also important for reducing the overhead in quantum error correction \cite{LidBru13_new}.

Quantum compilation is further complicated by the fact that the type and performance of the native instructions depend severely on the available physical interactions and the extent to which they can be controlled.
Most of the compiling literature has focused on native instructions given by single and two-qubit gates.
Two-qubit gates are arguably the simplest entangling gates that can be experimentally realized and dominate in most quantum computing architectures, such as supercomputing qubits.

In contrast, ion trap quantum computers naturally involve all-to-all interactions and thus allow
for the realization of multi-qubit gates which entangle multiple qubits simultaneously \cite{wang_multibit_2001,monz_14-qubit_2011}.
Consequently, there has been a growing interest in studying compilation with multi-qubit gates, and advantages over the use of two-qubit gates have been demonstrated \cite{Linke2017ExperimentalComparison,martinez_compiling_2016,maslov_use_2018, VanDeWetering20ConstructingQuantumCircuits, grzesiak_efficient_2022, bravyi_constant_cost_2022}.

The experimental realization of multi-qubit gates in ion trap quantum computers remains an active field of research.
Recent proof-of-principle experiments have demonstrated such gates acting on up to 10 qubits \cite{figgatt_parallel_2019,lu_global_2019,grzesiak_efficient_2020}.
These rely on precomputing and controlling rather complicated laser pulse shapes to physically implement the desired interactions.

In this work, we propose a simple method that uses \emph{some} all-to-all interaction to emulate arbitrary couplings.
We use this idea to synthesize time-optimal multi-qubit gates under minimal experimental setup and control hardware assumptions.

Concretely, we consider a quantum computing platform that satisfies the following requirements:
\begin{enumerate}[label= (\Roman*)]\itemsep=0pt
	\item single-qubit rotations can be executed in parallel, and \label{item:parallel_1qubit}
	\item it offers	Ising-type interactions with all-to-all connectivity. \label{item:Ising}
\end{enumerate}
We also develop compilation strategies with these gates, 
for which we additionally require that
\begin{enumerate}[resume,label= (\Roman*)]\itemsep=0pt
	\item there is a way to exclude specific qubits from participating in the interaction. \label{item:exclude}
\end{enumerate}
This assumption is sufficient to guarantee that unitaries can be compiled in a circuit depth depending only on the size of their support.

The requirements \ref{item:parallel_1qubit}--\ref{item:exclude} are satisfied, for example, in ion trap platforms \cite{figgatt_parallel_2019,lu_global_2019,grzesiak_efficient_2020}.
The motivation for our research originates from working with an ion trap where all gate control is based on microwave pulses and where Ising-type interactions with all-to-all connectivity are mediated through \acf{MAGIC} \cite{mintert_ion_trap_2001, wunderlich_conditional_2001,timoney_quantum_2011,ospelkaus_microwave_2011,khromova2012designer,piltz2016versatile,lekitsch_blueprint_2017,wolk_quantum_2017}, see \cref{sec:MAGIC}.

For concreteness, we assume that all Ising interactions are of $\Gate{ZZ}$ type, and we call the multi-qubit gates generated by arbitrary $\Gate{ZZ}$ couplings `\emph{$\Gate{GZZ}$ gates}'. 
Furthermore, by requirement~\ref{item:Ising}, there is an Ising Hamiltonian $H$ with fixed $\Gate{ZZ}$ interactions
We then present a synthesis method which realizes an arbitrary $\Gate{GZZ}$ gate as a sequence of time evolutions under $H$, interleaved with suitable $\Gate{X}$ layers.
The purpose of these $\Gate{X}$ layers is to temporarily flip the signs of some $\Gate{ZZ}$ coupling terms in $H$ to accumulate the desired coupling over the sequence.
We show that such a sequence can always be found and use a \acf{LP} to find a time-optimal realization of the desired $\Gate{GZZ}$ gate.
The resulting gate time scales approximately linear with the number of participating qubits $n$ and requires at most $n(n-1)/2$ $\Gate{X}$ layers.

This method may produce very short evolution times that can lead to problematically crammed single-qubit rotations in practice. 
We address this issue with a variation of our approach that extends the \ac{LP} to a \ac{MIP}.

We proceed by developing several compiling strategies with $\Gate{GZZ}$ gates.
We show that any Clifford circuit on $n$ qubits can be implemented using at most $n+1$ $\Gate{GZZ}$ gates, $n$ two-qubit gates and few single-qubit gates.
As an example for non-Clifford unitaries, we decompose the \acf{QFT} in a similar fashion into $n/2$ $\Gate{GZZ}$ gates, $n/2$ two-qubit gates and single-qubit gates.
An important application of quantum computers is the simulation of molecular dynamics. 
We present a method to tailor the approximate simulation in Ref.~\cite{Cohn21QuantumFilter} to our setup by compiling layers of Givens rotations into time-optimal $\Gate{GZZ}$ gates. 
This method significantly reduces the required number of single-qubit rotations with arbitrary small angles, which are challenging to implement in practice. 
Moreover, we propose a compilation method for diagonal unitaries as a step toward compilation strategies for general unitaries.

\subsection{Comparison to previous works}

\paragraph*{Synthesis of multi-qubit gates.}
Previous works \cite{figgatt_parallel_2019,lu_global_2019,grzesiak_efficient_2020} have mainly focused on implementing multi-qubit gates on ion trap quantum computers using the laser-controlled \ac{MS} mechanism \cite{PhysRevLett.83.2274,PhysRevA.62.022311,choi_optimal_2014}.
This setup requires segmented, amplitude-modulated laser pulses, the shape of which can be efficiently precomputed using the \ac{EASE} gate implementation \cite{grzesiak_efficient_2020}.

Here, the novelty of our work is that we only require the engineering of a single, fixed Ising Hamiltonian, which can be calibrated and fine-tuned to high accuracy.
This situation can be found in \ac{MAGIC} ion traps \cite{mintert_ion_trap_2001, wunderlich_conditional_2001,timoney_quantum_2011,ospelkaus_microwave_2011,khromova2012designer,piltz2016versatile,lekitsch_blueprint_2017,wolk_quantum_2017} but may also serve as a practical design principle for other architectures. 
With our synthesis method, multi-qubit $\Gate{GZZ}$ gates can be realized using only additional $\Gate{X}$ gates, resulting in a sequential series of simple pulses.
Arguably, this requires less fine-grained control of the pulse shapes than the \ac{EASE} gate protocol \cite{grzesiak_efficient_2020} and may thus be more implementation-friendly. 
However, we leave a detailed comparison of the approaches for future experimental work.

Furthermore, we introduce \emph{gate time}, instead of gate count, as the central metric for our synthesis of multi-qubit gates.
This metric is meaningful, especially for \ac{NISQ} devices, since the execution of circuits is limited by the coherence time of the qubits.
As we show, a side effect of our method is that it also produces rather short circuits, but not necessarily the shortest ones.
From our numerical studies, we expect that the gate time of our approach scales at most linear with the number of \emph{participating} qubits $n$.
Hence, we expect our method to produce faster multi-qubit gates than the \ac{EASE} gate protocol, which additionally scales linearly in the \emph{total} number of qubits $N$.

A conceptually related approach to our synthesis method was presented in Ref.~\cite{parra_rodriguez_digital_analog_2020} in the context of \ac{DAQC}.
There, the gate synthesis is based on solving a system of linear equations and is inherently restricted to $\Gate{X}$ layers acting on at most two qubits.
In contrast, our work optimizes for the total gate time of the sequence and, to this end, allows for layers with arbitrary support.
In this way, we avoid the problem of negative times encountered in Ref.~\cite{parra_rodriguez_digital_analog_2020} and observe a gate time scaling approximately linear in $n$, in contrast to the quadratic scaling in Ref.~\cite{parra_rodriguez_digital_analog_2020}.

\paragraph*{Compilation with $\Gate{GZZ}$ gates.}
A strategy to decompose general unitaries with multi-qubit $\Gate{GZZ}$ gates is presented in Ref.~\cite{martinez_compiling_2016}. 
It is based on maximizing the fidelity while using as few multi-qubit gates as possible.
This optimization is computationally costly, so the numerical results in Ref.~\cite{martinez_compiling_2016} cover only up to $4$ qubits.

Different compiling strategies with multi-qubit $\Gate{GZZ}$ gates have recently been investigated for Clifford unitaries. 
In Ref.~\cite{maslov_use_2018}, an implementation with $12n-18$ $\Gate{GZZ}$ gates is reported, which has been improved to $6n-8$ $\Gate{GZZ}$ gates in Ref.~\cite{VanDeWetering20ConstructingQuantumCircuits}.
Subsequently, it was shown that $6\log(n)+O(1)$ $\Gate{GZZ}$ gates are enough if $n/2$ ancillary qubits are used \cite{grzesiak_efficient_2022}.
Here, our ancilla-free approach reduces the prefactor because it requires at most $n+1$ multi-qubit $\Gate{GZZ}$ gates and $n$ two-qubit gates.
Shortly after publishing the preprint of this work, it was shown in Ref.~\cite{bravyi_constant_cost_2022} that any Clifford unitary on $n$ qubits can, in fact, be implemented with at most $26$ so-called $\Gate{GCZ}$ gates which are equivalent to $\Gate{GZZ}$ up to single-qubit $\Gate{Z}$ rotations.
In Ref.~\cite{bravyi_constant_cost_2022}, the authors also pointed out that the results in Ref.~\cite{maslov_depth_2022} can be used to obtain an ancilla-free implementation with $2\log(n)+O(1)$ $\Gate{GZZ}$ gates.
The constant-depth scheme of Ref.~\cite{bravyi_constant_cost_2022} can be readily combined with the time-optimal synthesis of $\Gate{GZZ}$ gates discussed in \Cref{sec:J_shaping} to show that any Clifford unitary can be realized in linear time on a platform satisfying the requirements \ref{item:parallel_1qubit}--\ref{item:exclude}.
For large $n$ this would further reduce the number of required $\Gate{GZZ}$ gates. 
However, for small $n\leq 13$ the compilation method presented in \Cref{sec:clifford_compiling} may still be advantageous.

Refs.~\cite{piltz2016versatile, ivanov2015simplified} propose a hand-tailored implementation of the \acl{QFT} on three qubits that uses simultaneous Ising-type interactions to achieve a speed-up. 
We use the same interactions, but our scheme can be applied to systems of arbitrary size and employs our time-optimal multi-qubit gates (cf.\ \cref{sec:QFT}.)

\subsection{Outline}
The remainder of the paper is structured as follows: 
We close this introductory section with a brief introduction to the computational primitives obtained from the requirements \ref{item:parallel_1qubit}--\ref{item:exclude} and how these are realized in microwave controlled ion traps with \ac{MAGIC}.
In \cref{sec:J_shaping}, we introduce our $\Gate{GZZ}$ gate and the time-optimal gate synthesis method.
Also, we define the \ac{MIP} to solve the problem of too short Ising-evolution times and support our claim of time-optimality with numerical results.
In \cref{sec:compiling} we present compiling strategies with our multi-qubit gate for Clifford circuits, the \ac{QFT}, molecular dynamics, and general diagonal unitaries.
Moreover, we demonstrate the performance of these compilation schemes with numerical results for the Clifford circuits and the \ac{QFT}.

\begin{figure*}[t]
\centering
\begin{subfigure}[b]{.4\textwidth}
	\centering
	\definecolor{LightGray}{gray}{0.7}
	\def\r{.12}
	\begin{tikzpicture}[
		ion/.style = {ball color=blue!30!LightGray},
		plt/.style = {thick},
		V/.style={blue!70!black},
		B/.style={red!50!black},
		]
		\foreach \i/\x in {1/-2.493, 2/-1.451, 3/-0.477, 4/0.477, 5/1.451, 6/2.493}
		\shade[ion] (\x, 1.5) circle (\r)
		coordinate (ion_\i);
		\foreach \i/\ip in {1/2, 2/3, 3/4, 4/5, 5/6}
		\draw[<->, shorten <=.2cm, shorten >=.2cm] (ion_\i) -- (ion_\ip)
		node (V_\i) [midway, above] {\small $V_\mathrm{C}$};
		\draw[->,plt] (-3,0) -- ++(6, 0)
		node (z) [anchor = north] {$z$};
		\draw[->,plt] (0, -1) -- (0, 1.1);
		\draw[domain=-3:3,plt,V] plot ({\x}, {0.04*\x*\x});
		\path (-3,.4) node[anchor = south west, inner sep = 1pt,V] {$V_\mathrm{ext}(z)$};
		\draw[domain=3:-3,plt,B] plot ({-\x}, {-0.18*\x})
		node[anchor = south east, inner sep = 1pt,B] {$B(z)$};
    \path (0,-1) ++(0,-.5);
	\end{tikzpicture}
	\caption{Sketch of a typical MAGIC ion trap}
	\label{fig:trap_sketch}
\end{subfigure}
\hfill
\begin{subfigure}[b]{.5\textwidth}
	\centering
	\begin{tikzpicture}[
		level/.style = {thick,line cap=round},
		transition/.style = {<->,shorten <=0.1cm, shorten >=0.1cm}, 
		help/.style = {midway,inner sep =0},
		plt/.style={thick},
		excited/.style={level, red!70!black, plt}, 
		ground/.style={level, blue!80!black, plt},
		transform canvas = {scale=0.9,xshift=-8em}
		]
		\def\b{3.5}
		\draw[->,plt] (0, 0) -- (\b, 0) -- ++(.1,0)
		node [right, inner sep = 2pt] {$B$};
		\draw[->,plt] (0, 0) -- (0, 4)
		node [right] {Energy};
		\draw[domain=0:\b,excited] plot ({\x}, {3.2 + 0.2*\x})
		coordinate (B1);
		\draw[domain=0:\b,excited] plot ({\x}, {2.5 + 0.7*sqrt(1 + 0.04*\x*\x)})
		coordinate (B2);
		\draw[domain=0:\b,excited] plot ({\x}, {3.2 - 0.2*\x})
		coordinate (B3);
		\draw[domain=0:\b,ground]  plot ({\x}, {1.2 - 0.7*sqrt(1 + 0.04*\x*\x)})
		coordinate (B4);
		
		\draw[excited] (-.1,3.2) -- ++(-.3,0) node[left,black] {$\ket{F=1}$};
		\draw[ground] (-.1,.5) -- ++(-.3,0) node[left,black] {$\ket{F=0}$};
		
		\draw[transition] (-.25, 0.5) -- (-.25, 3.2)	node[midway,left] 
      {\small \rotatebox{90}{$f_0 \approx 12.6 \text{GHz}$}};
		
		\draw[excited] (B1) ++ (1.2,0) --++ (1,0)
		node(nF1m-1)[help]{} 
		node(F1m-1)[right,black]{$\ket{1,+1}$};
		\draw[excited] (B2) ++ (0.8,0) --++ (1,0)
		node(nF1m0)[help]{}
		coordinate(cF1m0);
		\draw[excited] (B3) ++ (0.5,0) --++ (1,0)
		node(nF1m1)[help]{}
		coordinate(cF1m1);
		\draw[ground]  (B4) ++ (0.8,0) --++ (1,0)
		node(nF0)[help]{}
		coordinate(cF0); 
		
		\path (F1m-1.west|-cF1m0) node[anchor = west] {$\ket{1,\phantom{+}0}$}
		(F1m-1.west|-cF1m1) node[anchor = west] {$\ket{1,-1}$}
		(F1m-1.west|-cF0)   node[anchor = west] {$\ket{0,0}$};
		
		\draw[transition] (nF0.north) ++ (.2,0) -- node[right, inner sep = 2pt] {$\sigma^+$} (nF1m-1.north) ; 
		\draw[transition] (nF0.north) -- node[anchor = center, fill = white, opacity = .9] {$\pi$} (nF1m0.center);
		\draw[transition] (nF0.north) ++ (-.2,0) -- node[left, inner sep = 1pt] {$\sigma^-$} (nF1m1);
	\end{tikzpicture}
	\caption{Hyperfine levels and qubits in the $^{171}\mathrm{Yb}^+$ ground state}
	\label{fig:Yb_levels}
\end{subfigure}
\label{fig:trap_and_Yb}
\caption{
Schematics of a MAGIC ion trap. 
\textbf{Left: }six ions in linear configuration, with the plot below indicating the confining trap potential~$V_\mathrm{ext}$ in axial direction and magnetic gradient field~$B$.
There is Coulomb repulsion between any two ions (only indicated for next neighbors).
\textbf{Right: }the transition frequencies of the hyperfine sublevels of the $^{171}\mathrm{Yb}^+$ ground statedepend on the magnetic field strength
(Breit-Rabi diagram \cite{breit_rabi_1931}, Zeeman effect exaggerated).
The transitions on the right correspond to $\sigma^-$, $\pi$ and $\sigma^+$~qubit, respectively.
}
\end{figure*}
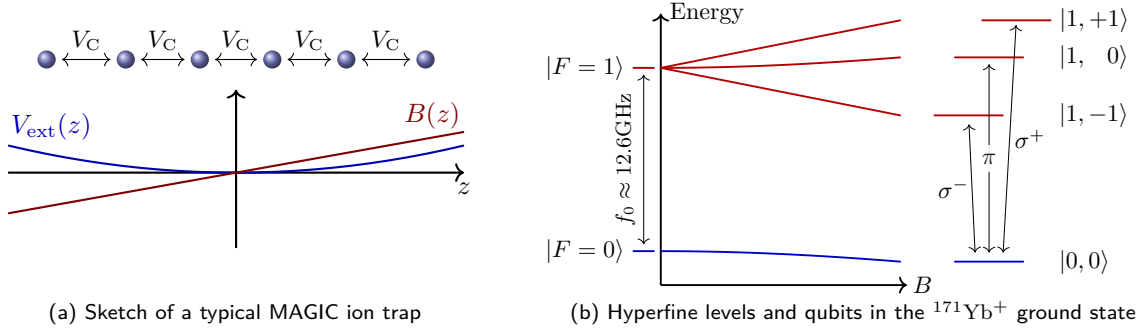

\subsection{Computational primitives}
\label{sec:comp_prim}

On the abstract quantum computing platform with $N$ qubits specified by the requirements \ref{item:parallel_1qubit}--\ref{item:exclude} above, interactions between the qubits are generated by the Ising-type Hamiltonian
\begin{equation}
\label{eq:Hising}
H_0	\coloneqq -\frac{1}{2} \vec{Z}^T J \vec{Z} = -\sum_{i<j}^N J_{ij} Z_i Z_j,
\end{equation}
where the vector $\vec{Z} = (Z_1, \dots, Z_N)^T$ collects all local Pauli $\Gate{Z}$ operators.
The \emph{coupling matrix} $J\in \RR^{N\times N}$ is a symmetric matrix which encapsulates the physical properties of the platform.
Since the diagonal entries of $J$ merely generate a global and hence unobservable phase, we henceforth assume that $J$ has vanishing diagonal.
By requirement \ref{item:exclude}, we can assume that $H_0$ acts only on the $n \leq N$ relevant qubits, and we thus assume w.l.o.g.~that $J$ is a $n\times n$ matrix.

Letting the system evolve under the Hamiltonian~\eqref{eq:Hising} for some time $t$ generates a unitary operation on the qubits.
Our gate synthesis method is based on the observation that layers of local $\Gate{X}$ gates can be used to emulate the evolution under a modified Hamiltonian:
For any binary vector $\vec{s} \in \FF_2^n$, set
\begin{equation}
\Gate{X}^{\vec s} \coloneqq \bigotimes_{i=1}^n \Gate{X}^{s_i},
\end{equation}
and define the \emph{qubit encoding} $\vec m = (-1)^{\vec{s}} \in \{-1, +1\}^n$ (to be understood entry-wise).
We then have the following modified time evolution:
\begin{equation}
\label{eq:virtual_encoding_basis}
\Gate{X}^{\vec s} \exp(-\i t H_0) \Gate{X}^{\vec s} = \exp\big( -\i t H(\vec{m})  \big).
\end{equation}
Here, $H(\vec m)	\coloneqq -\frac{1}{2} \vec{Z}^T J(\vec m) \vec{Z}$ is the modified Hamiltonian with coupling matrix $J(\vec{m}) \coloneqq J \circ \vec{m} \vec{m}^T$, and $\circ$ denotes the Hadamard (entry-wise) product.

If we apply multiple time evolution operators with different encodings in succession, we can further simplify this scheme.
For two encodings $\vec m = (-1)^{\vec s}$ and $\vec m' = (-1)^{\vec s'}$, we can combine the adjacent $\Gate{X}$~layers in \cref{eq:virtual_encoding_basis} and obtain
\begin{equation}
\label{eq:change_of_encoding}
\Gate{X}^{\vec s} \e^{ -\i t H_0} \Gate{X}^{\vec s} \Gate{X}^{\vec s'} \e^{ -\i t' H_0 } \Gate{X}^{\vec s'}
=
\Gate{X}^{\vec s} \e^{ -\i t H_0 } \Gate{X}^{\vec s \oplus \vec s'} \e^{ -\i t' H_0 } \Gate{X}^{\vec s'} .
\end{equation}
Hence, a change of encoding can be performed with a number of $\Gate X$~gates equal to the number of sign flips needed to obtain $\vec m'$ from $\vec m$.
The total number of $\Gate{X}$~layers needed to traverse a sequence of encodings is only one more than the length of the sequence.

\subsection{Experimental motivation: Ion trap quantum computing with microwaves}
\label{sec:MAGIC}
Let us give a brief overview of a physical platform on which our computational primitives can be realized. 
For details, we refer the reader to \cref{apdx:MAGIC} and Ref.~\cite{wunderlich_conditional_2001}.

The energy difference between hyperfine sublevels of some atomic state typically falls into the microwave regime of the electromagnetic spectrum, which makes pairs of such hyperfine states natural candidates for microwave-controlled qubits.
For example, the ``ground state'' of ions with nuclear spin $I = \tfrac{1}{2}$ and total electron angular momentum $J = \tfrac{1}{2}$ (e.g.\ Ytterbium-171 ions, $^{171}\mathrm{Yb}^+$) exhibits four hyperfine sublevels.
They group into a singlet with $F=0$ and a triplet with $F=1$, where $F$ is the quantum number specifying the magnitude of total angular momentum $\vec{F} = \vec{I} + \vec{J}$.
The triplet states are energetically degenerate, but can be distinguished by their value of the magnetic quantum number $m_F \in \{-1, 0, +1\}$.
The non-degenerate singlet state has $m_F = 0$.
There is an energy gap between the multiplets, which for $^{171}\mathrm{Yb}^+$ corresponds to a microwave frequency of about $12.6\text{GHz}$, see also \cref{fig:Yb_levels}.

We use one ion to implement a single qubit and choose the singlet state as the computational zero state $\ket{0} \coloneqq \ket{F=0, m_F=0}$.
We then have the freedom to encode the computational one state $\ket{1}$ into any of the triplet states, and indicate this by the magnetic quantum number $m_F \in \{-1, 0, +1\}$ of the chosen $\ket{1} \coloneqq \ket{F=1, m_F}$.

In ion traps, magnetic fields can be used to lift the degeneracy of the triplet through the Zeeman effect, see \cref{fig:Yb_levels}.
This separates the different qubit encodings in frequency space and enables single-qubit operations through microwave-driven two-level Rabi oscillations.
Certain sequences of pulses on different transitions in the multilevel system can also be used to change the qubit encoding coherently (see \cref{apdx:MAGIC}).
However, this possibility only plays a minor role in our analysis, as we will explain below.

We now extend our scope to multiple ions in the same trap.
They are stored in a linear configuration and form a ``Coulomb crystal'' due to their mutual repulsion.
In the \ac{MAGIC} scheme, the eponymous magnetic gradient along the crystal axis leads to different field strengths for the different ion positions, see \cref{fig:trap_sketch}.
The consequence are different Zeeman splittings, which make the qubits distinguishable in frequency space.
Thus, addressability is achieved, although the microwaves cannot be focused onto single ions.
Additionally, the ions experience a ``dipole force'' in the inhomogeneous field, which couples internal and external degrees of freedom (s.~\cref{apdx:MAGIC}).
This effect can be interpreted as an Ising-like interaction between the qubits, which we use in this work to generate multi-qubit gates.

To sum it up, the abstract requirements \ref{item:parallel_1qubit}--\ref{item:exclude} are realized in microwave-driven ion traps exposed to inhomogeneous magnetic fields as follows: \ref{item:parallel_1qubit} single qubit rotations are realized by microwave-driven Rabi oscillations which can be executed in parallel through digitally generated microwave signals \cite{Kasprowicz:20}. \ref{item:Ising} the Ising-type interaction is the natural interaction in this setup. \ref{item:exclude} selected ions can be taken out of the interaction by encoding them into the magnetic insensitive state with $m_F=0$.

\section{Synthesizing multi-qubit gates with Ising-type interactions}
\label{sec:J_shaping}
In this section, we investigate the set of gates which is generated by all possible time evolution operators of the Hamiltonians $H(\vec m)$ defined in \cref{eq:Hising}.
Given time steps $\lambda_{\vec m} \geq 0$ during which the encoding $\vec m$ is used, we thus consider unitaries of the form
\begin{equation}\label{eq:multi-qubitGate}
\prod_{\vec m} \e^{-\i \lambda_{\vec m} H(\vec m)}
 = \e^{-\i \sum_{\vec m} \lambda_{\vec m} H(\vec m)},
\end{equation}
where we used that the diagonal Hamiltonians $H(\vec m)$ mutually commute. 
For all possible encodings $\vec m\in \{-1,+1\}^n$ we collect the time steps $\lambda_{\vec m}$ in a vector $\vec \lambda \in \R^{2^n}$ and interpret $t= \sum_{\vec m}\lambda_{\vec m}$ as the total time of the unitary $\e^{-\i H}$.

We interpret the generated unitary as the time evolution operator under the \emph{total Hamiltonian}
\begin{equation}
\label{eq:targetHamiltonian}
  H \coloneqq -\frac 12 \vec Z^T A \vec Z\, ,
\end{equation}
where we defined the \emph{total coupling matrix}
\begin{equation}
\label{eq:Jdecomp}
A \coloneqq \sum_{\vec m} \lambda_{\vec m} J(\vec m)=J \circ \sum_{\vec m} \lambda_{\vec m} \vec m \vec m^T
\end{equation}
and used the linearity of the Hadamard product.

Since the $\vec m \vec m^T$ are symmetric, $A$ inherits the symmetry and the vanishing of the diagonal of $J$, see \cref{sec:comp_prim} and \cref{apdx:MAGIC}.
We wish to make our description of the coupling matrix independent of the platform dependent  details given by $J$. 
Therefore, we define the Hadamard quotient $M$ with entries
\begin{equation}\label{eq:MasQuotient}
	M_{ij} \coloneqq \begin{cases}
		A_{ij} / J_{ij}, & i \neq j\, , \\
		0, & i = j\, .
	\end{cases}
\end{equation}
The implicit assumption that $J$ has no vanishing non-diagonal entries is commonly met by experiments.
Our objective is to minimize the total gate time and the amount of $\vec m$'s needed to express the matrix $M$.
To this end we formulate the following \acf{LP}:
\begin{mini}
{}{\vec 1^T \vec \lambda}{}{}
\label{eq:LP1}
\addConstraint{M} {=\sum_{\vec m} \lambda_{\vec m}  \vec m  \vec m^T}{}{}
\addConstraint{ \vec \lambda}{ \in \R^{2^{n-1}}_{\geq 0}}{}{}
\addConstraint{ \vec m}{ \in \{ -1,+1 \}^n . }{}{}
\end{mini}
As above, $\vec 1=(1,1,\dots,1)$ is the all-ones vector such that $\vec{1}^T \vec \lambda = \sum_{\vec m}\lambda_{\vec m}$.
Moreover, we use the symmetry $(-\vec m)(-\vec m)^T = \vec m \vec m^T$ to reduce the degree of freedom in $\vec\lambda$ from $2^n$ to $2^{n-1}$.

This \ac{LP} has the form of a $\ell_1$-norm minimization over the non-negative vector $\vec\lambda$.
As such, it is a convex relaxation of minimizing the number of non-zero entries of $\vec \lambda$, sometimes called the $\ell_0$-``norm''.
Heuristically, it is thus expected that the \ac{LP}~\eqref{eq:LP1} favors sparse solutions.
As we shall see shortly in \cref{thrm:frame1}, the \ac{LP}~\eqref{eq:LP1} always has a \emph{feasible} solution (i.e.~there are variables $\vec \lambda$ such that all constraints are satisfied) for any symmetric matrix $M$ with vanishing diagonal.
The theory of linear programming then guarantees the existence of an optimal solution with at most $n(n-1)/2$ non-zero entries, see \cref{prop:optSparseSol1}.

For any symmetric $n\times n$ matrix $A$ with vanishing diagonal, we define an associated multi-qubit gate $\Gate{GZZ}(A)$, where $\Gate{GZZ}$ stands for ``global $ZZ$ interactions'',
\begin{equation}
\label{eq:GZZ}
\begin{aligned} 
 \Gate{GZZ}(A) &\coloneqq \e^{\i \frac{1}{2} \vec Z^T A \vec Z} \, .
\end{aligned}
\end{equation}
Here, the decomposition of $A$ is found using the \ac{LP}~\eqref{eq:LP1} and involves at most $n(n-1)/2$ different encodings $\vec m$.
Recall from \cref{sec:comp_prim}, that these encodings can be emulated with suitable $\Gate X$ layers and hence $\Gate{GZZ}(A)$ can be implemented using at most $n(n-1)/2+1$ such layers.
We call the exact number of $\Gate X$ layers the \emph{encoding cost} of $\Gate{GZZ}(A)$, and $\vec{1}^T \vec{\lambda}$ the \emph{total $\Gate{GZZ}$ time}.
For this derivation, we have intentionally been agnostic of the physical details of the ion trap but note that the values of $\lambda_{\vec m}$ and therefore $t$ depend on the (physical) coupling matrix $J$.

Finally, given an optimal decomposition of $A$, it is also possible to minimize the total number of $\Gate X$ gates needed for the implementation.
Since every $\Gate{X}$ gate introduces noise, such a minimization improves the fidelity of $\Gate{GZZ}$ gates in practice.
By \cref{eq:change_of_encoding}, the number of $\Gate X$ gates needed to change the encoding from $\vec m$ to $\vec m'$ is exactly the number of signs in $\vec m$ that have to be flipped to obtain $\vec m'$.
Since the resulting gate $\e^{-\i H}$ does not depend on the order of encodings $\vec m$, one can minimize the total number of sign flips over all possible orderings.
However, finding an optimal ordering  generally corresponds to solve a traveling salesman problem on the support of $\vec \lambda$ and is thus NP-hard \cite{dumer_hardness_nodate}.
Nevertheless, there are good heuristic algorithms, e.g.\ Christofides's algorithm introduced in Ref.~\cite{christofides_worstcase_1976}.

Before demonstrating how to use the flexibility and the all-to-all connectivity of $\Gate{GZZ}$ gates for compiling, we discuss some theoretical aspects as well as limitations and extensions of the above approach.
We conclude by presenting numerical results for the synthesis of $\Gate{GZZ}(A)$ gates for randomly chosen matrices $A$. 

\subsection{Theoretical aspects}
\label{sec:theory}
First, we show the existence of a solution for the \ac{LP}~\eqref{eq:LP1} via frame theoretic arguments, then we investigate the sparsity of optimal solutions from a geometrical viewpoint.
Let us define the $n(n-1)/2$-dimensional subspace of symmetric matrices with vanishing diagonal by 
\begin{equation}\label{def:symtraceless}
	\HTL(\RR^n) \coloneqq \Set*{M \in \mathrm{Sym}(\RR^n)\given M_{ii}= 0 \ \forall i\in [n]} \, .
\end{equation}
Moreover, we denote the set of outer products generated by all possible encodings by
 \begin{equation}\label{def:outprod}
	\mc{V} \coloneqq \Set*{ \vec m \vec m^T \given \vec m \in \{-1,+1\}^n , m_n = +1} \, .
\end{equation}
Due to the symmetry $\vec m \vec m^T = (-\vec m)(-\vec m)^T$ we can uniformly fix the value of one of the entries of $\vec m$. 
We chose the convention $m_n=+1$.
\begin{definition}
 Let $V$ be a (finite-dimensional) Hilbert space.
 A set of vectors $v_1,\dots,v_N \in V$ is called a \emph{frame} if their linear span is $V$.
 A frame is said to be \emph{tight} if there exists $a>0$ such that for all $v \in V$
 \begin{equation}
  a \| v \|^2 = \sum_{i=1}^N | \langle v, v_i \rangle |^2 \, .
 \end{equation}
 Moreover, a frame is said to be \emph{balanced} if $\sum_{i=1}^N v_i = 0$.
\end{definition}
With this definition, we obtain the following:
\begin{theorem}
	\label{thrm:frame1}
	The set $\mc{V}$ is a balanced tight frame for $\HTL(\RR^n)$.
	In particular, the \ac{LP}~\eqref{eq:LP1} has a feasible solution for any $M\in\HTL(\RR^n)$.
\end{theorem}
The proof that $\mathcal V$ is a balanced tight frame can be found in \cref{apdx:frame}, along with other properties of $\mathcal V$.
Since $\mathcal V$ is a frame for $\HTL(\RR^n)$, any matrix $M\in\HTL(\RR^n)$ can be decomposed as $M=\sum_{\vec m} \lambda_{\vec m} \vec m \vec m^T$.
In other words, the \ac{LP}~\eqref{eq:LP1} has a feasible solution for any $M\in\HTL(\RR^n)$.
Then, a standard linear programming argument shows that there is always an optimal solution with sparsity at most $n(n-1)/2$.
Such a solution can be numerically found by using variants of the simplex algorithm (see e.g.~Ref.~\cite{Karloff1991} for more details).
We formulate this fact as the following proposition and defer the proof to \cref{apdx:convOpt}, where we also show some geometric properties of optimal solutions of a more general class of \ac{LP}'s.
\begin{proposition}[Sparsity of optimal solutions]
\label{prop:optSparseSol1}
 There exists an optimal solution to the \ac{LP}~\eqref{eq:LP1} with sparsity $\leq n(n-1)/2$ for every $M \in \HTL(\RR^n)$.
 The simplex algorithm is guaranteed to return such an optimal solution.
\end{proposition}

\subsection{Practical limitations}
\label{sec:limits}
In the previous section we showed how to implement a $\Gate{GZZ}$ gate through the Ising-evolution time under at most $n(n-1)/2$ different encodings.
However, in an actual ion trap, practical limitations might occur for very short evolution times.
For this work, we neglect the potential error introduced by a finite recoding time, i.e.\ the time needed to perform $\Gate{X}$ gates.
During recoding we have additional Hamiltonian terms corresponding to the $\Gate{X}$ gates simultaneously with the ``always-on'' Ising Hamiltonian~\eqref{eq:Hising}.
These cause unwanted effects and the introduced errors become non-negligible when the Ising-evolution time approaches the recoding time, see also \cref{apdx:single_qubit}.
Below we observe in our numerical results for the \ac{LP} that some $\lambda_{\vec m}$ are below that recoding time.
In this section, we address these problems and propose extensions to our approach to mitigate them.
First, we discuss the amount of error we would make in an appropriate norm by ignoring too small $\lambda_{\vec m}$, i.e.\ by defining a threshold $\varepsilon$ and setting $\lambda_{\vec m} \equiv 0$ if $\lambda_{\vec m} < \varepsilon$.
To avoid small evolution times in the first place, we then define a \acf{MIP} which can solve the problem exactly with a lower (and upper) bound on the $\lambda_{\vec m}$.

\subsubsection{Truncation error}
\label{sec:truncation_error}
Suppose the target Hamiltonian $H$ is decomposed as in \cref{eq:targetHamiltonian}.
Given a threshold $\varepsilon > 0$ for the $\lambda_{\vec m}$, we define $C \coloneqq \Set*{\vec m \given \lambda_{\vec m} \leq \varepsilon }$ and approximate $H$ by the Hamiltonian
\begin{equation}
 H' = -\sum_{\vec m \notin C} \lambda_{\vec m} \vec Z^T ( J \circ \vec m \vec m^T ) \vec Z \, .
\end{equation}
The diamond norm error made by replacing the time evolution operator of $H$ with the one of $H'$, is upper bounded by half their spectral norm distance (see e.g.\ Ref.~\cite{Kliesch2020TheoryOfQuantum}):
\begin{equation}
\label{eq:trunc_err}
	\frac{1}{2}\snorm*{\e^{-\i H} - \e^{-\i H'}}
	= \max_{x\in\FF_2^n} \abs*{ \sin\!\Big(\frac{1}{2} (H_{x,x}-H'_{x,x}) \Big) } \, .
\end{equation}
Here, we used that $H$ and $H'$ are diagonal in the computational basis with diagonal entries $H_{x,x} \coloneqq \sandwich xHx$ and, moreover, that $\abs{1-\e^{\i\varphi}} = 2\,\abs{\sin(\varphi/2)}$.
The difference of the diagonal entries is
\begin{equation}
 H_{x,x} - H'_{x,x} = - \sum_{i\neq j} J_{ij} \sum_{\vec m\in C} \lambda_{\vec m} m_i m_j (-1)^{x_i+x_j} .
\end{equation}
Since $|\sin(\theta)|\leq |\theta|$ for all $\theta\in\R$, we find that
\begin{equation}
\begin{aligned}
 \frac12\snorm*{\e^{-\i H} - \e^{-\i H'}}
 &\leq \frac14 \sum_{i\neq j} \left| J_{ij} \right| \sum_{\vec m\in C} \lambda_{\vec m} \, .
\end{aligned}
\end{equation}
Hence, the truncation error scales with the total truncated time and the interaction strength.
Both depend on the number of participating qubits $n$, however not in a straightforward way.
The scaling of the truncation error under realistic assumptions is showed in \cref{fig:trunc_err}.

\begin{figure}[t]
	\centering
    \includegraphics[width=.5\linewidth]{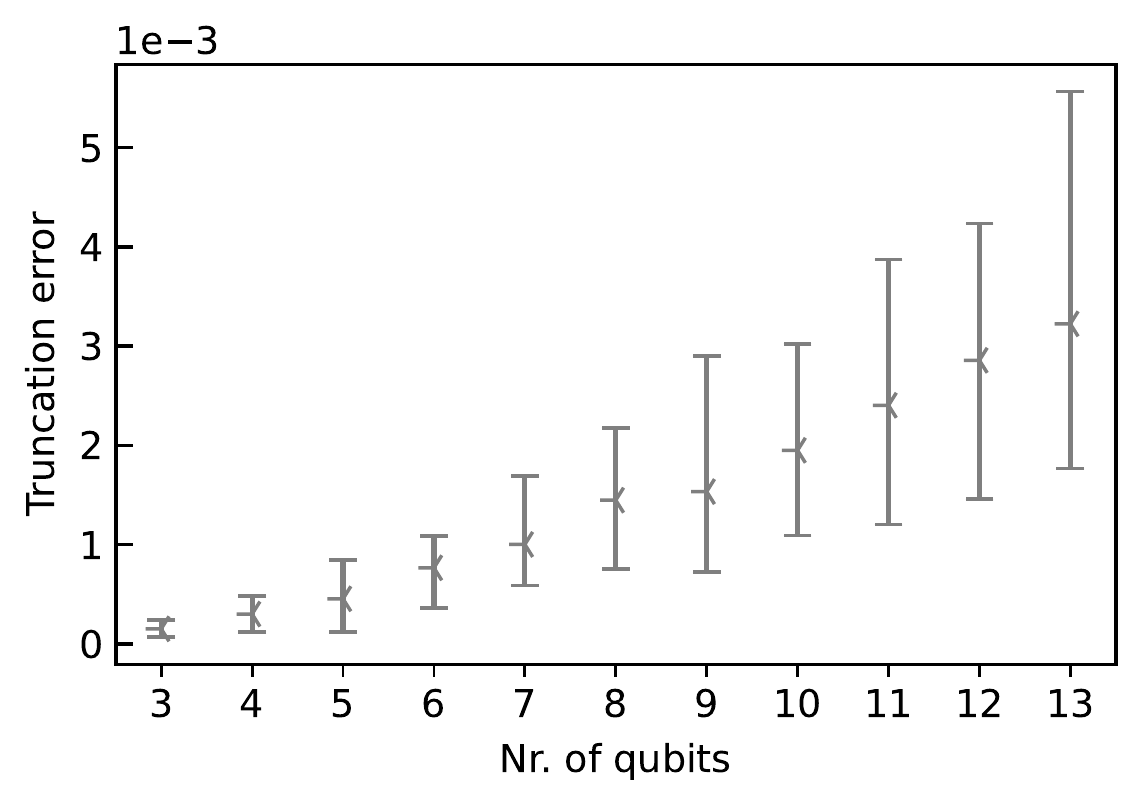}
	\caption{The error scaling, from \cref{eq:trunc_err}, introduced by truncating all $\lambda_{\vec m} < \varepsilon_l = 27 \mu \text{s}$ from the solution of the \ac{LP} is displayed.
  	The error bars show the minimum/maximum deviation from the mean over 20 randomly sampled binary $A$ matrices.
	}
	\label{fig:trunc_err}
\end{figure}

\subsubsection{The mixed integer program approach}
\label{sec:MIP}
To avoid small $\lambda_{\vec m}$ without introducing an additional error as above, we propose to add a lower bound on their values to the \ac{LP}~\eqref{eq:LP1}.
To this end, we introduce additional binary variables $\vec b$, as in Ref.~\cite{karahanoglu_mixed_2013}, which renders the optimization into the following \acf{MIP}:
\begin{mini}
{}{ \alpha \vec{1}^T \vec \lambda + (1- \alpha) \vec{1}^T  \vec b}{}{}
\label{eq:MIP}
\addConstraint{M}{=\sum_{\vec m} \lambda_{\vec m}  \vec m  \vec m^T}{}{}
\addConstraint{\varepsilon_l \vec b} {\leq \vec \lambda \leq \varepsilon_u \vec b}{}{}
\addConstraint{ \vec \lambda}{ \in \R_{\geq 0}^{2^{n-1}}}{}{}
\addConstraint{ \vec m}{ \in \{ -1,+1 \}^n }{}{}
\addConstraint{ \vec b}{ \in \FF_2^{2^{n-1}} . }{}{}
\end{mini}
Here $0 \leq \varepsilon_l < \varepsilon_u$ are bounds on the entries of $ \vec \lambda$:
The lower bound $\varepsilon_l$ can be set to the minimal Ising-evolution time which can be realized in practice.
The upper bound $\varepsilon_u$ can be chosen freely, but small values of $\varepsilon_u$ are favorable since the runtime of the \ac{MIP} is generally shorter for smaller intervals $[\varepsilon_l , \varepsilon_u]$ as explained in Ref.~\cite{karahanoglu_mixed_2013}.
Since $\max (  \vec \lambda)$ depends on $\max_{i<j} (M_{ij})$ one can select $\varepsilon_u \propto \max_{i<j} (M_{ij})$.
The parameter $\alpha\in[0,1]$ is used to assign weights to the optimization of $\vec{1}^T \vec b$ (sparsity) and $\vec{1}^T \vec \lambda$ (total $\Gate{GZZ}$ time), respectively.

\subsection{Numerical results}
\label{sec:numerics}
We investigate how well the proposed methods for the synthesis of $\Gate{GZZ}(A)$ gates performs in practice.
To this end, we solve the \ac{LP}~\eqref{eq:LP1} and the \ac{MIP}~\eqref{eq:MIP} numerically for randomly chosen matrices $A$ and compare the solutions to a naive approach.
As explicated above, our goal is to minimize the total $\Gate{GZZ}$ time $\vec{1}^T \vec{\lambda}$ subject to $A/J = M = \sum_{\vec m} \lambda_{\vec m} \vec m \vec m^T$.
The numerical results in this section show the performance of $\Gate{GZZ}$ gates on all qubits, i.e.\  $n=N$.

To demonstrate the performance of our approach in a realistic setting we consider the ion trap architecture of \cref{apdx:MAGIC} with a harmonic trap potential.
Concretely, we take Ytterbium $^{171}\mathrm{Yb}^+$ ions with Rabi frequency $\Omega = 2 \pi\, 100 \text{kHz}$, magnetic field gradient $B_1 = 100 \text{T/m}$ and axial trap frequency $\omega_z = 2 \pi\, 100 \text{kHz}$.
This determines the coupling matrix $J$ via \cref{eq:def_J}, see Ref.~\cite{johanning_quantum_2009} for more details.
An example coupling matrix $J$ for $10$ ions is shown in \cref{fig:example_coupling_matrix}.
We made the Python code for its computation available on GitHub \cite{GitHub}.

\begin{figure}[t]
	\centering
    \includegraphics[width=.5\linewidth]{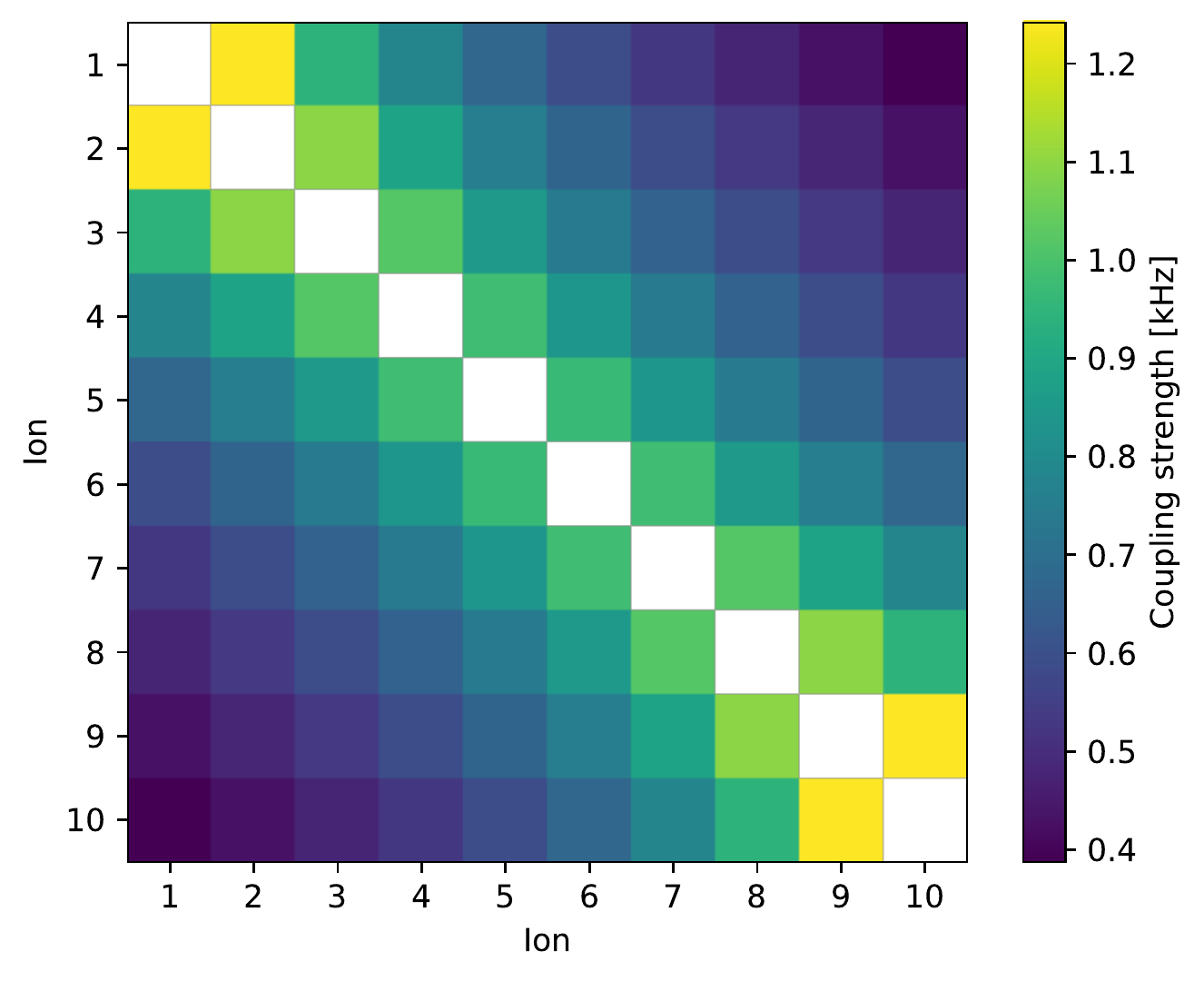}
	\caption{An example coupling matrix $J$ for an ion trap with $10$ ions and harmonic trap potential.
	$J$ is determined by the physical parameters of the trap, see \cref{sec:numerics} for the concrete values, and can be computed by \cref{eq:def_J}.
	}
	\label{fig:example_coupling_matrix}
\end{figure}

\begin{figure*}
	\begin{center}
		\begin{tabular}{lr}
			\includegraphics[width=0.45\linewidth]{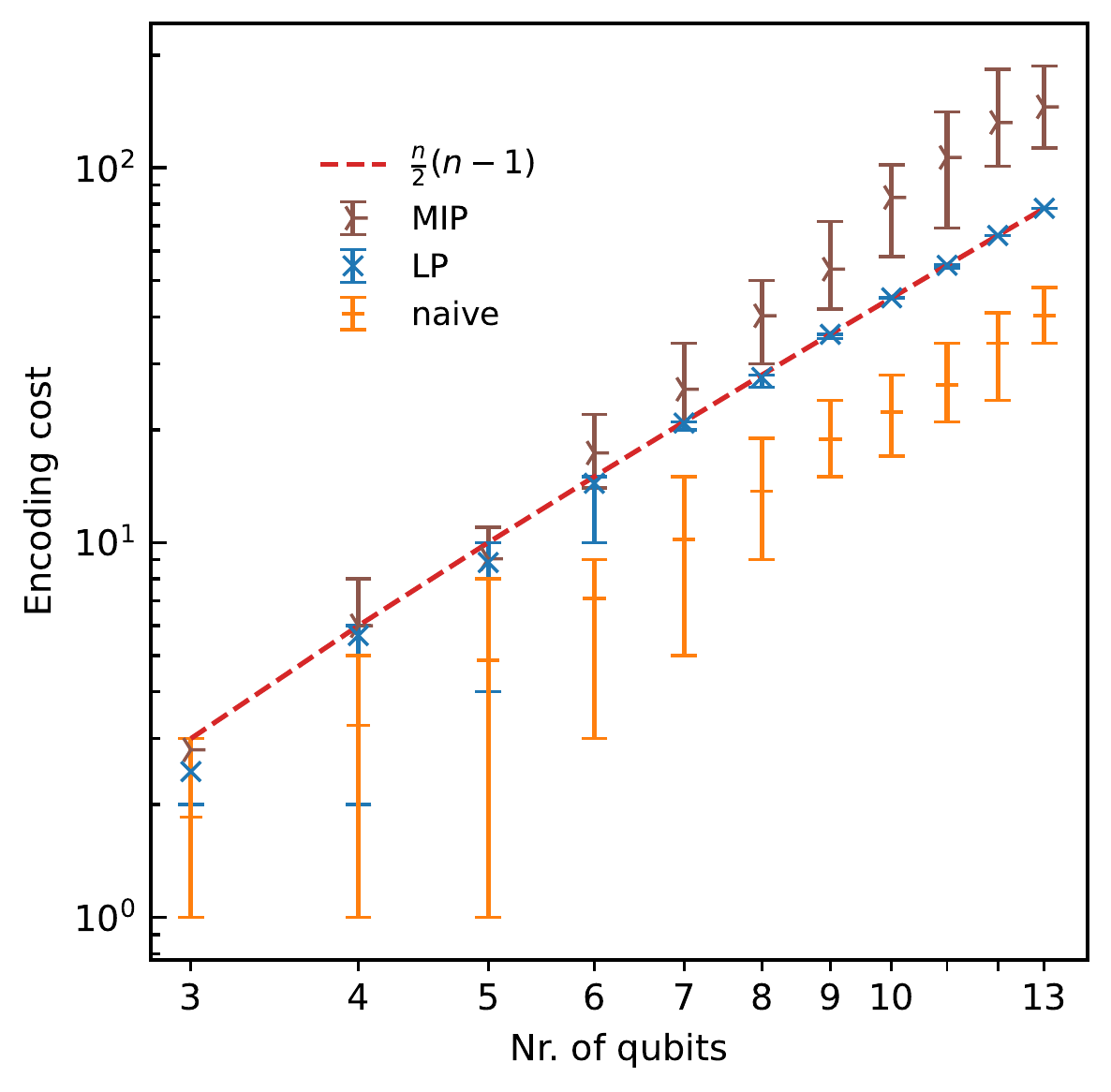}
			&
			\includegraphics[width=0.45\linewidth]{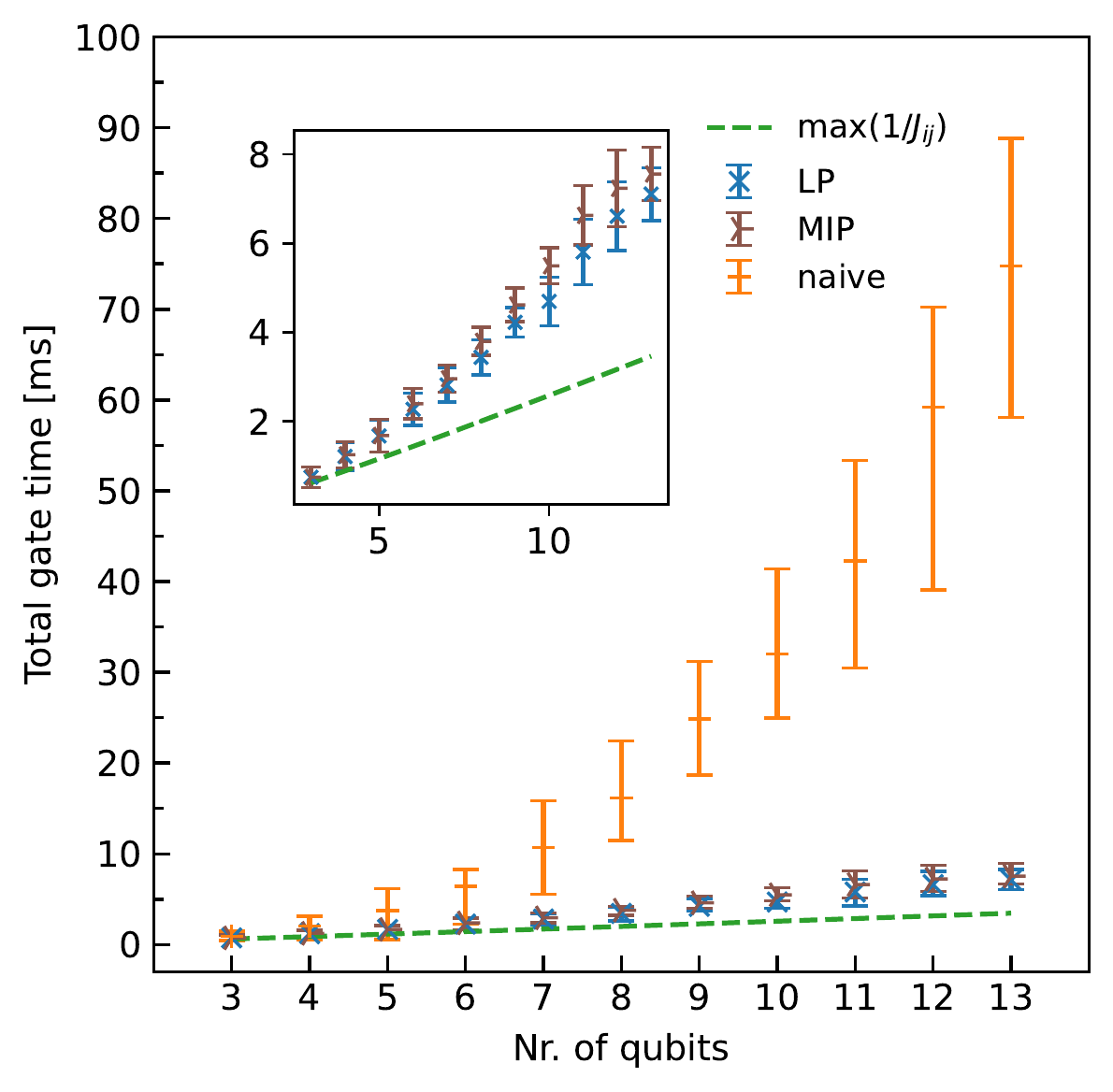}
		\end{tabular}
		\caption{\label{fig:MIP}
			Comparing the performances of the \ac{LP} and the \ac{MIP} to the naive approach for the implementation of a random $\Gate{ZZ}$ layer.
			The error bars show the minimum/maximum deviation from the mean over 20 randomly sampled binary $A$ matrices.
			\textbf{Left:} The encoding cost to implement a random $\Gate{ZZ}$ layer on $n$ qubits.
			The naive approach only costs half as much as the \ac{LP}, see \cref{sec:numerics}.
			For many qubits, the \ac{LP} approaches the upper bound (red dashed line) of \cref{prop:optSparseSol1}, whereas the \ac{MIP} exceeds this upper bound.
			\textbf{Right:} The total gate time of the naive approach scales quadratically with the number of qubits due to the quadratic scaling of the number of possible $\Gate{ZZ}$ gates.
			The total gate time of both the \ac{LP} and the \ac{MIP} scales approximately linear.
			The inset shows a zoomed in version of the plot.
			One can see that the \ac{LP} and the \ac{MIP} roughly take twice as long as the slowest two-qubit $\Gate{ZZ}$ gate (green dashed line).
		}
	\end{center}
\end{figure*}

On a logical level, we consider a symmetric binary matrix $A \in \HTL(\FF_2^n)$ defined in \cref{eq:Jdecomp}, which indicates where the $\Gate{ZZ}$ gates are located:
For all $i<j$, $A_{ij} = 1$ if there is a $\Gate{ZZ}$ gate between qubits $i$ and $j$. 
We simulate a random $\Gate{ZZ}$ gate layer by sampling the entries of the lower/upper triangular part of $A$ uniformly from $\{ 0, 1 \}$.

We compare the results of the \ac{LP} and the \ac{MIP} with a ``naive approach'', which corresponds to a sequential execution of the $\Gate{ZZ}$ gates.
As before, we neglect the gate time for single-qubit gates.
Moreover, ``total gate time'' refers to the total $\Gate{GZZ}$ time, i.e.\ $\vec{1}^T \vec{\lambda}$, for the \ac{LP} and the \ac{MIP}.
For the naive approach, the ``total gate time'' is the time needed to execute the sequence of $\Gate{ZZ}$ gates, i.e.\ $\sum_{i < j} A_{ij}/J_{ij}$, and the encoding cost is the number of $\Gate{ZZ}$ gates $\sum_{i < j} A_{ij}$.

\Cref{fig:MIP} shows the numerical results for solving the \ac{LP}~\eqref{eq:LP1} and the \ac{MIP}~\eqref{eq:MIP}.
The encoding cost for the naive approach scales roughly with $n^2/4$ since on average half of the randomly chosen $n(n-1)/2$ entries of the lower triangular part of $A$ do not vanish.
In \cref{fig:MIP} we find that the total $\Gate{GZZ}$ time scales linearly with the number of participating qubits.
In contrast, by \cref{prop:optSparseSol1} the number of encodings needed to implement a $\Gate{GZZ}$ gate increases quadratically with the number of participating qubits.
Thus, the time between the encodings becomes shorter and shorter, which results in arbitrarily small evolution times $\lambda_{\vec m}$ for many qubits, see also \cref{sec:limits}.
Hence, the more qubits we consider, the more solutions $\lambda_{\vec m}$ of the \ac{LP} are smaller than the lower bound $\varepsilon_l$, see \cref{fig:trunc_err}. 
This explains the deviation of the encoding cost of the \ac{MIP} from that of the \ac{LP}.

\Cref{fig:MIP} also shows that the \ac{MIP} can be solved in practice and, similar to the \ac{LP}, yields nearly time-optimal $\Gate{GZZ}$ gates.
Thus, even when taking practical limitations into account, $\Gate{GZZ}$ gates can be implemented in linear time and an encoding cost of approximately $n(n-1)/2$.

\paragraph*{Details of computer implementation.}
We use the Python package CVXPY \cite{diamond2016cvxpy,agrawal2018rewriting_new} with the GNU linear program kit simplex solver \cite{GLPK} to solve the \ac{LP}, and the MOSEK solver \cite{mosek} for the \ac{MIP}.
We change one parameter of MOSEK to improve the runtime at the expense of not finding the optimal solution. Concretely, we set the MOSEK parameter {\tt MSK\_DPAR\_MIO\_TOL\_REL\_GAP} to $0.6$.

For the \ac{MIP} we choose the lower bound $27 \mu \text{s} = \varepsilon_l \leq \lambda_{\vec m}$ for all $\vec m$, which is motivated by the concrete ion trap setup:
the duration of a robust $\Gate{X}$ gate is about five times longer than the duration of a $\Gate{X}$ gate (i.e.\ a $\pi$-pulse), which is roughly $\pi / \Omega = 5 \mu \text{s}$ \cite{kukita_short_2021}.
We further pick the upper bound $\varepsilon_u = 3/2 \max_{i<j} \abs{M_{ij}}$.
Note that if the interval $[\varepsilon_l , \varepsilon_u]$ is too narrow, we might not be able to find any feasible solution. If it is too wide, the runtime of the solver might increase \cite{karahanoglu_mixed_2013}. 
We observed in our numerical studies (not shown) that both the total gate time and the encoding cost are essentially constant in $\alpha$ as long as $\alpha$ is not too close to the extremal values $0$ and $1$.
We therefore, deliberately put equal weights on the two terms in the objective function and set $\alpha = 0.5$.

In practice, the simplex algorithm has a runtime which is polynomial in the problem size \cite{spielman_smoothed_2004}.
Since the \ac{LP}~\eqref{eq:LP1} has $2^{n-1}$ variables, the runtime is exponential in $n$.
For moderate $n$, this is however still manageable on modern hardware -- a Laptop with Intel Core i7 Processor (8x 1.8 GHz) and 16 GB RAM needs on average only $20$ seconds to solve the \ac{LP} for $n=13$.
This is about the size of most ion trap quantum computers nowadays.
The runtime of mixed integer programs is exponential in the worst case. 
Nevertheless, solving \ac{MIP}~\eqref{eq:MIP} for $n=13$ qubits using the MOSEK solver took on average about seven minutes on the same hardware.

An implementation of the \ac{LP} and \ac{MIP} used to generate \cref{fig:MIP} is provided on GitHub \cite{GitHub}.

\section{Compilation with \texorpdfstring{$\Gate{GZZ}$}{GZZ} gates}
\label{sec:compiling}
The $\Gate{GZZ}$ gates defined in \cref{eq:GZZ} can be used in compilation schemes to improve multi-qubit gate counts over previous results \cite{maslov_use_2018,VanDeWetering20ConstructingQuantumCircuits,parra_rodriguez_digital_analog_2020}. 
Our main result is that we can implement any Clifford circuit on $n$ qubits using only $n+1$ $\Gate{GZZ}$ gates, $n$ two-qubit gates and single-qubit gates.

First, we introduce the notation used in later sections and derive some gate equivalences that will help us later.
Then, we study the practically relevant case of compiling global Clifford unitaries. 
These ubiquitous unitaries play an important role as basic building blocks of quantum circuits, especially in fault-tolerant quantum computing, and are of major importance for cryptographic and tomographic protocols due to their statistical properties.
For the compilation of Clifford unitaries, we make use of their Bruhat decomposition, also used in Ref.~\cite{bravyi_hadamard_free_2020}, and compile the entangling layers into $\Gate{GZZ}$ gates.
Interestingly, this compilation scheme for entangling Clifford layers can be generalized beyond pure Clifford circuits. 
As an example, we show how it can be used to compile an $n$-qubit \ac{QFT}.
Finally, we propose a compiling scheme for general diagonal unitaries which may be used as a step towards decomposing general unitaries.

\subsection{Notation}
\label{sec:gate_notation}
We introduce some notation and link our $\Gate{GZZ}$ gate from \cref{eq:GZZ} to other known entangling gates.
Subscripts on gates indicate on which qubits they act. 
For $\vec x \in \FF_2^n$ and $\alpha \in [0, 2\pi)$ we thus have
\begin{equation}
\label{eq:notation}
 \begin{alignedat}{2}
  \Gate{Z} \text{ rotation: } & \qquad & \Gate{R_Z} (\alpha)_j \ket{\vec x} &= \e^{\i \alpha x_j} \ket{\vec x} \, ,\\
  \Gate{ZZ} \text{ gate: } &\qquad & \Gate{ZZ} (\alpha)_{i,j} \ket{\vec x} &= \e^{\i \alpha (x_i \oplus x_j)} \ket{\vec x} \, ,\\
  \text{Controlled }\Gate{Z} \text{ rotation: } &\qquad & \Gate{CR_Z} (\alpha)_{i,j} \ket{\vec x} &= \e^{\i \alpha x_i x_j } \ket{\vec x} \, ,\\
  \text{Controlled }\Gate{X} \text{ gate: } &\qquad & \Gate{CX}_{i,j} \ket{\vec x} &= \ket{x_1, \dots , x_{j-1}, x_{j} \oplus x_{i}, x_{j+1}, \dots, x_n} \, ,\\
  \text{Hadamard gate: } &\qquad & \Gate{H}_{j} \ket{\vec x} &= \frac{1}{\sqrt{2}}\ket{x_1, \dots, x_{j-1}}\big(\ket{0}+(-1)^{x_j}\ket{1}\big)\ket{x_{j+1}, \dots, x_n} \,,
 \end{alignedat}
\end{equation}
where $\oplus$ denotes addition modulo $2$.
In particular, we denote the phase gate and the controlled-Z gate as
\begin{equation}\label{eq:def_s_and_cz}
	\Gate{S}_j\coloneqq\Gate{R_Z}(\pi/2)_j\quad\text{and}\quad\Gate{CZ}_{i,j} \coloneqq \Gate{CR_Z}(\pi)_{i,j},
\end{equation}
respectively. 
Note that in the definition of the $\Gate{ZZ}$ gate $\alpha$ can be identified with entries of the matrix $A$ in
\cref{eq:Jdecomp}.
In terms of integer arithmetic
, we have $2 x y = x+ y-( x \oplus y)$ for $ x, y \in \{0,1 \}$ which yields
\begin{equation}\label{eq:CR_7}
  \Gate{CR_Z} (\alpha)_{i,j} \equiv \Gate{R_Z} (\alpha/2)_i \Gate{R_Z} (\alpha/2)_j \Gate{ZZ} (-\alpha/2)_{i,j} \, .
\end{equation}
Since the $\Gate{GZZ}$ gate from \cref{eq:GZZ} consists only of $\Gate{ZZ}$ gates, we can express it as
\begin{equation}
\label{eq:GZZ_to_ZZ}
 \Gate{GZZ}(A) = \e^{\i a} \prod_{i<j} \Gate{ZZ} (-2A_{ij})_{i,j} \, ,
\end{equation}
where $A \in \HTL (\RR^n)$
is a symmetric matrix with vanishing diagonal
and $a \coloneqq \sum_{i<j} A_{ij}$.
Similarly, a layer of arbitrary controlled $R_Z$ rotations is characterized by $A \in \HTL (\RR^n)$ via
\begin{equation}
\label{eq:CZlayer}
 \Gate{GCR_Z} (A)
 \coloneqq \prod_{i<j} \Gate{CR_Z} (A_{ij})_{i,j}
 = \e^{-\frac{\i}{4} a} \Gate{GZZ}(A/4) \prod_{i=1}^n \Gate{R_Z} (b_i/2)_i  \, ,
\end{equation}
where we used \cref{eq:CR_7} and abbreviated $b_i \coloneqq \sum_j A_{ij}$.
A general $\Gate{CX}$ layer is given by
\begin{equation}
\label{eq:dirGCX}
 \Gate{GCX} (B) \ket{\vec x} \coloneqq \ket{B \vec x}
 \, ,
\end{equation}
with a matrix $B \in \mathrm{GL}_n(\FF_2)$.
We call $\Gate{GCX} (B)$ a \emph{directed $\Gate{CX}$ layer} if $B$ is lower or upper triangular.

\subsection{Clifford circuits}
\label{sec:clifford_compiling}
The \emph{Clifford group} $\Cl{n}$ is a finite subgroup of the unitary group $\U(2^n)$ that is generated by the single-qubit Hadamard gate $\Gate{H}_j$ and phase gate $\Gate{S}_j$, as well as the two-qubit $\Gate{CX}_{i,j}$ gate.
Conversely, it is a natural task to decompose an arbitrary Clifford unitary $U\in\Cl{n}$ into these generators.
This task is solved by a number of algorithms, see e.g.\ Refs.~\cite{aaronson_improved_2004, maslov_shorter_2018,duncan_graph_theoretic_2020}, with the same asymptotic gate count, but a differently structured output circuit .

Here, we make use of the so-called ``Bruhat decomposition'' \cite{maslov_shorter_2018, bravyi_hadamard_free_2020}.
This decomposition has the advantage that the entangling gates are grouped either in $\Gate{CZ}$ gate layers or directed $\Gate{CX}$ layers, which both can be directly compiled into $\Gate{GZZ}$ gates.
More precisely, the algorithm in Ref.~\cite{bravyi_hadamard_free_2020} writes any Clifford unitary in the form -X-Z-CX-CZ-S-H-CX-CZ-S-, where
\begin{itemize}\itemsep=-2pt
\item -X-Z- is a layer of Pauli gates, 
\item -CX- are layers of directed $\Gate{CX}$ gates,
\item -CZ- are layers of controlled $\Gate{Z}$ gates,
\item -S- are layers of phase gates $\Gate{S}$ and
\item -H- is a layer consisting of Hadamard gates $\Gate{H}$ (and permutation operations).
\end{itemize}

In the following, we concentrate on the decomposition of the $\Gate{CZ}$ and $\Gate{CX}$ layers as the remaining layers consist of local gates with straightforward implementation.
As we show, the compilation of a directed $\Gate{CX}$ layer is a lot more expensive than the compilation of a $\Gate{CZ}$ layer.

As a corollary of $\Gate{CZ}$ layer compilation, we show how to efficiently prepare multi-qubit stabilizer states with $\Gate{GZZ}$ gates, i.e.\ states of the form $U\ket{0}$ where $U\in\Cl{n}$ is a Clifford unitary.
Stabilizer states are important e.g.\ for the construction of mutually unbiased bases, or more generally, informationally complete \acp{POVM}, and their preparation is thus of practical relevance for tomographic protocols, see e.g.\ Refs.~\cite{kueng_qubit_2015, huang_predicting_2020}.

First we transform the directed $\Gate{CX}$ layer by conjugating the $\Gate{CX}$ gate targets with Hadamard gates.
Then we use the structure of the resulting gate layer to reduce the encoding cost with the $\Gate{GZZ}$ gate.
Furthermore, we use the same method to reduce the encoding cost of the \ac{QFT}.
We underpin the advantages of our method with numerical simulations.

\subsubsection{Implementing CZ layers}
\label{sec:CZ}
Since $\Gate{CZ}$ gates commute, we can rewrite any $\Gate{CZ}$ circuit, using \cref{eq:def_s_and_cz,eq:CZlayer}, as
\begin{equation}\begin{aligned}
 \Gate{GCZ}(A)
 &\coloneqq \prod_{i<j} \Gate{CZ}_{i,j}^{A_{ij}} = \prod_{i<j} \Gate{CR}_{\Gate Z}(\pi A_{ij})_{i,j} \\
 &= \Gate{GCR_Z}(\pi A) \\
 &= \e^{-\frac{\i\pi}{4} a} \Gate{GZZ}\left( \frac{\pi}{4} A \right) \prod_{i=1}^n \Gate{R_Z}\left( \frac{\pi}{2} b_i \right)_i \\
 &= \e^{-\frac{\i\pi}{4} a} \Gate{GZZ}\left( \frac{\pi}{4} A \right) \prod_{i=1}^n \Gate{S}_i^{b_i},
\end{aligned}\end{equation}
where $A\in\HTL (\FF^n)$ is again a symmetric matrix with zero diagonal and $a$ and $b_i$ are defined as in \cref{sec:gate_notation}.
Note that the phase gate $\Gate S$ has order 4, so effectively only $b_i \text{ mod } 4$ plays a role and the single-qubit gates are from the set $\{ \Gate I, \Gate S, \Gate Z, \Gate S^\dagger \}$.

\subsubsection{Stabilizer state preparation}
Although any stabilizer state can be written as $U\ket{0}$ for some Clifford unitary $U\in\Cl{n}$, not the full Clifford group is needed to generate all stabilizer states.
In fact, it is well known that any stabilizer state can be obtained by acting with local Clifford gates on \emph{graph states} \cite{schlingemann_stabilizer_2001_new,schlingemann_quantum_2001}.
Graph states are defined as
\begin{equation}
 \ket{A} \coloneqq  \prod_{i<j} \Gate{CZ}_{i,j}^{A_{ij}} \ket{+^n},
\end{equation}
where $A\in\HTL (\FF_2^n)$ and $\ket{+^n}=H^{\otimes n}\ket{0}$.
Hence, by the above, any stabilizer state can be prepared by an initial global Hadamard layer, a $\Gate{GZZ}$ gate, and a final layer of single-qubit Clifford gates.

\subsubsection{Decomposing directed CX layers}
\label{sec:dirCX}
Let us now consider the two directed $\Gate{CX}$ layers in the Bruhat decomposition \cite{bravyi_hadamard_free_2020}.
With the notation introduced in \cref{eq:dirGCX}, a \emph{fan-out} gate with the following action on a state $\vec x \in \FF_2^n$,
\begin{equation}
 \Gate{GCX} (B_{FO}) \ket{\vec x} = \ket{x_1, x_2 \oplus x_1, \dots , x_n \oplus x_1} \, ,
\end{equation}
has the binary matrix
\begin{align}
\label{eq:fanoutParity}
 B_{FO} = \begin{pmatrix}
 1      & 0      & \cdots & 0 \\
 1      &       &  &  \\
 \vdots &  & \1_{n-1} &  \\
 1      &      &  &  \\
\end{pmatrix} \, .
\end{align}
From the structure of the matrix $B_{FO}$ it can be seen that any directed $\Gate{CX}$ layer can be realized by at most $n-1$ fan-out gates.
Each fan-out gate is equivalent to a $\Gate{GCR_Z}$ gate, conjugated with Hadamard gates on the target qubits.
Since by \cref{eq:CZlayer} each $\Gate{GCR_Z}$ can be realized with one $\Gate{GZZ}$ gate, we need at most $n-1$ $\Gate{GZZ}$ gates to realize a directed $\Gate{CX}$ layer, see \cref{eq:fullCX2CZ} for an example.
Concretely, the total encoding cost for a directed $\Gate{CX}$ layer, implementing each of the $n-1$ fan-out gates with one $\Gate{GZZ}$ gate scales as $\LandauO (1/6 n^3)$.
Below \cref{eq:fullCX2CZ} we show that $\lfloor \frac{n-1}{2} \rfloor$ $\Gate{GZZ}$ gates are enough to implement a directed $\Gate{CX}$ layer.
Therefore, one requires $n-1$ (if $n$ is odd) or $n-2$ (if $n$ is even) $\Gate{GZZ}$ gates, for the two directed $\Gate{CX}$ layers appearing in the Bruhat decomposition of Ref.~\cite{bravyi_hadamard_free_2020}.
Since we can realize the $\Gate{CZ}$ layer with exactly one $\Gate{GZZ}$ gate, each Clifford circuit requires only $n+1$ or $n$ $\Gate{GZZ}$ gates and only $n-1$ or $n$ two-qubit $\Gate{CZ}$ gates for $n$ odd or $n$ even, respectively.

\paragraph{Fully directed $\Gate{CX}$ layer.}
We call a directed $\Gate{CX}$ layer \emph{fully directed} if the corresponding gate $\Gate{GCX} (B)$ is characterized by a $n \times n$ matrix $B$ with zeros in the upper triangular matrix and ones everywhere else.
Fully directed $\Gate{CX}$ layers are related to the textbook \ac{QFT}, see \cref{sec:QFT} below.

Remember that we can represent any directed $\Gate{CX}$ layer as a concatenation of $n-1$ fan-out gates.
Commuting Hadamard gates through each target of the fan-out gate transforms the controlled $\Gate{X}$ to a controlled $\Gate{Z}$ gate by $\Gate{HXH=Z}$.
This converts a fan-out gate into a $\Gate{CZ}$-type fan-out gate, which we also call fan-out gate for short.
We can thus transform a fully directed $\Gate{CX}$ layer by applying $\1 = H^2$ from the left and commuting one of the Hadamard gates to the right until it hits a control:
\begin{equation}
\label{eq:fullCX2CZ}
\scalebox{0.80}{
\begin{quantikz}
 & \ctrl{5} & \qw      & \qw \ldots & \qw & \qw     & \qw\\[0.3em]
 & \targ{}  & \ctrl{4} & \qw \ldots & \qw & \qw     & \qw\\[0.5em]
 & \targ{}  & \targ{}  & \qw \ldots & \qw & \qw     & \qw\\[0.35em]
\wave&&&&&&\\[0.35em]
 & \targ{}  & \targ{}  & \qw \ldots & \qw & \ctrl{1}& \qw\\[0.45em]
 & \targ{}  & \targ{}  & \qw \ldots & \qw & \targ{} & \qw
\end{quantikz}}
=
\scalebox{0.80}{
\begin{quantikz}
 & \qw      & \ctrl{5}   & \qw      & \qw    & \qw      & \qw \ldots & \qw & \qw      & \qw     & \qw & \qw\\
 & \gate{\Gate H} & \gate{\Gate Z} & \gate{\Gate H} & \ctrl{4} & \qw      & \qw \ldots & \qw & \qw      & \qw     & \qw & \qw\\
 & \gate{\Gate H} & \gate{\Gate Z} & \qw      & \gate{\Gate Z} & \gate{\Gate H} & \qw \ldots & \qw & \qw      & \qw     & \qw & \qw\\
\wave&&&&&&&&&&&\\
 & \gate{\Gate H} & \gate{\Gate Z} & \qw      & \gate{\Gate Z} & \qw & \qw \ldots & \qw & \gate{\Gate H} & \ctrl{1}& \qw & \qw\\
 & \gate{\Gate H} & \gate{\Gate Z} & \qw      & \gate{\Gate Z} & \qw & \qw \ldots & \qw & \qw      & \gate{\Gate Z} & \gate{\Gate H} & \qw
\end{quantikz}}
\end{equation}
Note that such a transformation obviously works also for an arbitrary directed $\Gate{CX}$ layer.
The locations of the resulting $\Gate H$ and $\Gate{CZ}$ gates can be represented as tables
\begin{equation}
\label{eq:hadamard}
 T_{\Gate H} = \begin{bmatrix}
 0      & 0      & 0      & \cdots & 0 \\
 1      & 1      & 0      & \cdots & 0 \\
 1      & 0      & 1      &        & \vdots \\
 \vdots & \vdots &        & \ddots & 0 \\
 1      & 0      & \cdots & 0      & 1
\end{bmatrix}
 \,
\hspace{0.5em} \text{ and } \hspace{0.5em}
 T_{\Gate{CZ}} = \begin{bmatrix}
 1      & 0      & \cdots & 0 \\
 1      & 1      & \cdots & 0 \\
 \vdots & \vdots & \ddots & \vdots \\
 1      & 1      & \cdots & 1
\end{bmatrix}
\, ,
\end{equation}
where $T_{\Gate{CZ}}$ has the same form as the matrix $B$ for the fully directed $\Gate{CX}$ layer. 
The row index of these tables indicates the qubit on which the corresponding gate acts, while the column index indicates the ``time step'', i.e.\ the horizontal position in the circuit diagram. 
Thus, $T_{\Gate{CZ}}[i,j] = 1$ if qubit $i$ is either a control or a target of a $\Gate{CZ}$ at time step $j$, and $T_{\Gate{CZ}}[i,j] = 0$ if qubit $i$ is idle at time step $j$.
Similarly, $T_{\Gate H}[i,j] = 1$ if a Hadamard gate acts on qubit $i$ at time step $j$, and $T_{\Gate H}[i,j] = 0$ if qubit $i$ remains unchanged.
For a fully directed $\Gate{CX}$ layer, $T_{\Gate H}$ and $T_{\Gate{CZ}}$ always have this form.
To locate the components of the circuit on the right-hand side of \cref{eq:fullCX2CZ} one starts with reading the first column of $T_{\Gate H}$, then the first column of $T_{\Gate{CZ}}$ and so on.

We now aim at reducing the encoding cost of the fully directed $\Gate{CX}$ layer on $n$ qubits. 
To this end, we reduce the supports $m$ of the $\Gate{GZZ}$ gates implementing the fan-out gates in the $\Gate{CX}$ layer, since the encoding cost of a $\Gate{GZZ}$ gate scales as $m(m-1)/2$, see \cref{sec:J_shaping}.
Consider an odd column $i$ in $T_{\Gate{CZ}}$ which corresponds to a $\Gate{CZ}$-type fan-out gate with control on qubit $i$.
We split column $i$ into two columns as follows:
The first one representing a two-qubit $\Gate{CZ}$ gate on qubits $i$ and $i+1$, and the second column is the same as the original one except that it does not target qubit $i+1$.
This splitting increases the number of columns in $T_{\Gate{CZ}}$ by one.
For example, for $n=5$ and $i=1$ we split the column $[1,1,1,1,1]^T$ into $[1,1,0,0,0]^T$ and $[1,0,1,1,1]^T$, where we keep in mind that the first nonzero entry in a column denotes the control qubit and hence has to appear in both parts.
Furthermore, we update the Hadamard table $T_{\Gate{H}}$ by inserting a zero column after column $i+1$ to account for the new column in $T_{\Gate{CZ}}$.
The fan-out gate resulting from the split of the odd column $i$ of $T_{\Gate{CZ}}$ together with the even column $i+1$ can be implemented with one $\Gate{GZZ}$ on $i$ qubits.

Note that we can not move parts of columns of $T_{\Gate{CZ}}$ to the left since there is always a Hadamard gate on the left of the control qubit of that column that blocks it.
Therefore, we only split odd columns~$i$ and move them to the right.

This splitting of columns of $T_{\Gate{CZ}}$ corresponds to moving the Hadamard gate on the $i+1$ qubit to the left or, equivalently, moving all $\Gate{CZ}$ gates except the one acting on the $i+1$ qubit to the right, as exemplified in the following for $n=5$ qubits:
\begin{equation}
\label{eq:Hcommute}
\scalebox{0.75}{
\begin{quantikz}[column sep=tiny]
 & \qw            & \ctrl{4}       & \qw            & \qw            & \qw            & \qw            & \qw            & \qw            & \qw            & \qw \\
 & \gate{\Gate H} & \gate{\Gate Z} & \gate{\Gate H} & \ctrl{3}       & \qw            & \qw            & \qw            & \qw            & \qw            & \qw \\
 & \gate{\Gate H} & \gate{\Gate Z} & \qw            & \gate{\Gate Z} & \gate{\Gate H} & \ctrl{2}       & \qw            & \qw            & \qw            & \qw \\
 & \gate{\Gate H} & \gate{\Gate Z} & \qw            & \gate{\Gate Z} & \qw            & \gate{\Gate Z} & \gate{\Gate H} & \ctrl{1}       & \qw            & \qw \\
 & \gate{\Gate H} & \gate{\Gate Z} & \qw            & \gate{\Gate Z} & \qw            & \gate{\Gate Z} & \qw            & \gate{\Gate Z} & \gate{\Gate H} & \qw 
\end{quantikz}}=
\scalebox{0.75}{
\begin{quantikz}[column sep=tiny]
 & \qw            & \ctrl{1}       & \qw            & \ctrl{4}       & \qw            & \qw            & \qw            & \qw            & \qw            & \qw            & \qw            & \qw \\
 & \gate{\Gate H} & \gate{\Gate Z} & \gate{\Gate H} & \qw            & \ctrl{3}       & \qw            & \qw            & \qw            & \qw            & \qw            & \qw            & \qw \\
 & \gate{\Gate H} & \qw            & \qw            & \gate{\Gate Z} & \gate{\Gate Z} & \gate{\Gate H} & \ctrl{1}       & \qw            & \ctrl{2}       & \qw            & \qw            & \qw \\
 & \gate{\Gate H} & \qw            & \qw            & \gate{\Gate Z} & \gate{\Gate Z} & \qw            & \gate{\Gate Z} & \gate{\Gate H} & \qw            & \ctrl{1}       & \qw            & \qw \\
 & \gate{\Gate H} & \qw            & \qw            & \gate{\Gate Z} & \gate{\Gate Z} & \qw            & \qw            & \qw            & \gate{\Gate Z} & \gate{\Gate Z} & \gate{\Gate H} & \qw
\end{quantikz}}
\end{equation}
For general $n$, the scheme works exactly the same as in this example. 

The circuit before the splitting on the left-hand side of \cref{eq:Hcommute} can be implemented with $n-1$ $\Gate{GZZ}$ gates each acting on $n, \dots , 2$ qubits respectively, where $n$ is the number of qubits of the directed $\Gate{CX}$ layer, and thus have an encoding cost of $\sum_{i=2}^n i(i-1)/2=n^3/6+\LandauO(n^2)$.
On the right-hand side we need $\lceil \frac{n-1}{2} \rceil$ $\Gate{CZ}$ gates and $\lfloor \frac{n-1}{2} \rfloor$ $\Gate{GZZ}$ gates, resulting in an encoding cost of
\begin{equation}
\label{eq:dir_CX_implementation_cost}
 \bigg\lceil \frac{n-1}{2} \bigg\rceil + \sum_{i=0}^{\lfloor \frac{n-1}{2} \rfloor} \frac{n-2i}{2} (n-2i-1) = \frac{1}{12} n^3 + \LandauO(n^2) \, .
\end{equation}
The first term comes from the encoding cost of the $\Gate{CZ}$ gates, which is one per $\Gate{CZ}$ gate.
The second term comes from combining a $\Gate{GZZ}$ gate on $n$ qubits with a $\Gate{GZZ}$ gate on $n-1$ qubits resulting in a $\Gate{GZZ}$ gate on $n$ qubits.
The encoding cost in \cref{eq:dir_CX_implementation_cost} has the same cubic scaling with $n$ as before combining $\Gate{GZZ}$ gates, but we were able to improve the coefficient from $1/6$ to $1/12$.
Recall from \cref{sec:numerics} that the encoding cost of the naive approach scales only quadratic, so we trade higher encoding cost for faster gates.
Each $\Gate{CZ}$ gate can be implemented as described in \cref{eq:CR_7} for $\alpha=\pi$ by a single $\Gate{ZZ}$ gate and two additional single-qubit $\Gate{R_Z}(\pi/2) \equiv \sqrt{\Gate{Z}} = \Gate{S}$ gates on the control and target qubit, respectively.
Since the $\Gate{S}$ gates do commute with $\Gate{ZZ}$ and $\Gate{GZZ}$ gates but do not commute with Hadamard gates, we can combine most of the $\Gate{S}$ gates to an $\Gate{S}^k$ gate, where $k \in \{0,1,2,3\}$.
There are two Hadamard gates on $n-1$ qubits and none on the first qubit, therefore we have $2(n-1)+1$ $\Gate{S}^k$ gates.
To summarize, in addition to $\lceil \frac{n-1}{2} \rceil$ $\Gate{CZ}$ gates and $\lfloor \frac{n-1}{2} \rfloor$ $\Gate{GZZ}$ gates we need $2n-1$ $\Gate{S}^k$ gates to implement a fully directed $\Gate{CX}$ layer on $n$ qubits.

Our method optimizes both the encoding cost and the total gate time of the directed $\Gate{CX}$ layer.
We expect to lose time-optimality for a fully directed $\Gate{CX}$ layer if we split any $\Gate{GZZ}$ gate into smaller pieces and therefore reduce the encoding cost.
In the extreme case, we would split all $\Gate{GZZ}$ gates into two-qubit $\Gate{ZZ}$ gates and thereby would end up with the naive approach explained in \cref{sec:numerics}.
Our method combines parts of the $i$-th column of $T_{\Gate{CZ}}$ with the $i+1$-th column to a $\Gate{GZZ}$ gate on $n-i+1$ qubits. 
One could think that moving parts of the larger $\Gate{GZZ}$ gate farther to the right might improve the encoding cost.
But the support of the part left behind, and therefore the encoding cost increases the farther we push the other part to the right.
In the extreme case of pushing all parts as far as possible to the right we end up with a ``transposed'' table where the lower triangular part is zero, and we did not achieve any reduction of the encoding cost.
\begin{algorithm}[H]
	\caption{Moving Hadamard gates.}\label{alg:hadamard}
	\raggedright\textbf{Input:} $T_{\Gate{CZ}}$
	\begin{algorithmic}[0]
		\State Initialize $T_{\Gate H}$ as in \cref{eq:hadamard}
		\State{$h_\mathrm{max} \gets 0$} \Comment{Position of the rightmost $\Gate H$ that has already been moved left}
		\For{$i = 1, \dots, n-1$}
		\State{$T_{\Gate{H}}[i, i] \gets 0$} \Comment{$\Gate{H}$ on $i$-th qubit leaves its position}
		\State{$c \gets \max\{ j = 0, \dots, i-1 | T_{\Gate{CZ}}[i, j] = 1 \} + 1$} \Comment{Find position directly after first $\Gate{CZ}$ to the left}
		\If{c = \verbat{NaN}} \Comment{No $\Gate{CZ}$ found ($\max$ was taken on an empty set)}
		\State{$T_{\Gate{H}}[i, 0] \gets 0$} \Comment{Cancel $\Gate{H}$ in the first layer}
		\ElsIf{$c = i$} \Comment{Unable to move left, attempt to move right}
		\If{$\{ j = i, \dots, n-1 | T_{\Gate{CZ}}[i, j] = 1 \} = \emptyset$} \Comment{If no $\Gate{CZ}$ is to the right...}
		\State{$T_{\Gate{H}}[i, n-1] \gets 1$} \Comment{...move $\Gate{H}$ to the last layer, ...}
		\Else
		\State{$T_{\Gate{H}}[i, i] \gets 1$} \Comment{...otherwise remain in place.}
		\EndIf
		\Else
		\State{$h_\mathrm{max} \gets \max\{ h_\mathrm{max}, c \}$} \Comment{Find the more restrictive condition (either $\Gate{CZ}$ on current qubit or $\Gate H$ on previous)}
		\State{$T_{\Gate{H}}[i, h_\mathrm{max}] \gets 1$} \Comment{Move $\Gate{H}$ to the target layer}
		\EndIf
		\EndFor
	\end{algorithmic}
	\raggedright\textbf{Output:} $T_{\Gate H}$
\end{algorithm}

\paragraph{Arbitrary directed $\Gate{CX}$ layer.}
Until now, we considered only fully directed $\Gate{CX}$ layers which in practice is a very special case.
More common are arbitrary directed $\Gate{CX}$ layers which corresponds to a $\Gate{GCX}(B)$ with $B \in \FF_2^{n \times n}$ still being lower/upper triangular but more sparse.
This sparsity, which translates to the table $T_{\Gate{CZ}}$, can be used to further reduce the encoding cost.
One might be able to move the Hadamard gates to the left/right which changes the support of the $\Gate{GZZ}$ gates.
We explain three different scenarios which can be easily verified by the simple structure of $T_{\Gate{CZ}}$ and $T_{\Gate H}$.
Obviously, the two Hadamard gates on any qubit $i>1$ cancel if they are separated only by identities. 
Similarly, if there is no control or target of a $\Gate{CZ}$ gate to the right of a Hadamard gate $\Gate H$, one can move $\Gate H$ to the last position of the circuit (without cancellation).
Otherwise, it might still be possible to push the Hadamard gate on qubit $i$ to the left until it hits a target of some $\Gate{CZ}$ gate.
However, one needs to be careful to not disrupt any previously generated $\Gate{GZZ}$ gates with qubit $i$ in its support.
This can be accomplished by disallowing the Hadamard gate to move across other Hadamard layers.
\Cref{alg:hadamard} implements these moves of the Hadamard gates by updating the table $T_{\Gate H}$ accordingly.

\Cref{alg:CZ} pools multiple columns of $T_{\Gate{CZ}}$ together into a single $\Gate{GZZ}$ gate similar as for fully directed $\Gate{CX}$ layers.
This takes into account the sparsity of $T_{\Gate{CZ}}$ and the positions of the Hadamard gates, i.e.\ $T_{\Gate H}$ generated by \cref{alg:hadamard}.
As for fully directed $\Gate{CX}$ layers, we can split a column of $T_{\Gate{CZ}}$ into a two-qubit $\Gate{CZ}$ gate and a $\Gate{GZZ}$ gate if necessary.
The algorithm starts from the left and tries to move columns of $T_{\Gate{CZ}}$, or parts of it, to the right.
The following cases occur:
If a column of $T_{\Gate{CZ}}$ has a Hadamard gate left to the first nonzero entry, i.e.\ the control of the fan-out gate, then it can only be moved to the right.
If a column of $T_{\Gate{CZ}}$ has no Hadamard gate left to the first nonzero entry, then this column can be combined with the previous column.
If a column has one Hadamard gate to the right, then the column can be split into a two-qubit $\Gate{CZ}$ gate whose target is the qubit on which the Hadamard gate acts, and a $\Gate{GZZ}$ gate which can be moved to the right.
Due to the structure of $T_{\Gate{CZ}}$ and $T_{\Gate H}$ there is never a Hadamard gate to the right of the first nonzero entry, i.e.\ the control of the fan-out gate, of a column.

\begin{algorithm}[H]
\caption{Moving $\Gate{CZ}$ gates.}\label{alg:CZ}
\raggedright\textbf{Input:} $T_{\Gate{CZ}}$, $T_{\Gate H}$
\begin{algorithmic}[0]
\For{$i = 1, \dots, n-1$}
    \If{$T_{\Gate H}[:, i] \wedge T_{\Gate{CZ}}[:,i]$ has at least one $1$} \Comment{$T_{\Gate{CZ}}[:,i]$ has Hadamard gates to the left}
        \If{$T_{\Gate H}[:, i+1]$ has exactly one $1$} \Comment{$T_{\Gate{CZ}}[:,i]$ has one Hadamard gate to the right}
            \State Split $T_{\Gate{CZ}}[:,i]$ into:
            \State $T_{\Gate{CZ}}^1 \coloneqq T_{\Gate H}[:, i+1] \oplus e_i$ and \Comment{$e_i = [0, \dots , 0,1,0, \dots , 0 ]^T$ with $1$ at the $i$th position.}
            \State $T_{\Gate{CZ}}^2 \coloneqq T_{\Gate H}[:, i+1] \oplus T_{\Gate{CZ}}[:,i]$
            \If{$T_{\Gate{CZ}}^2$ has at least two $1$'s} \Comment{Check that $T_{\Gate{CZ}}^2$ is not trivial}
                \State Move $T_{\Gate{CZ}}^2$ to the right.
            \EndIf
        \EndIf
    \Else \Comment{$T_{\Gate{CZ}}[:,i]$ has no Hadamard gates to the left}
        \State Move $T_{\Gate{CZ}}[:,i]$ to the left.
    \EndIf
\EndFor
\end{algorithmic}
\raggedright\textbf{Output:} $T_{\Gate{CZ}}$ (modified)
\end{algorithm}

\begin{example}
Consider the following circuit implementing an arbitrary directed $\Gate{CX}$ layer on $n=5$ qubits where the first two fan-out gates do not target qubits $2$ and $3$, respectively:
\begin{equation*}
\begin{tikzpicture}[baseline=(X.base)]
\node[scale=.65, ampersand replacement=\&] (X)	{
\begin{tikzcd}
 \& \qw            \& \ctrl{4}       \& \qw            \& \qw            \& \qw            \& \qw            \& \qw            \& \qw            \& \qw            \& \qw \\
 \& \gate{\Gate H} \& \qw            \& \gate{\Gate H} \& \ctrl{3}       \& \qw            \& \qw            \& \qw            \& \qw            \& \qw            \& \qw \\
 \& \gate{\Gate H} \& \gate{\Gate Z} \& \qw            \& \qw            \& \gate{\Gate H} \& \ctrl{2}       \& \qw            \& \qw            \& \qw            \& \qw \\
 \& \gate{\Gate H} \& \gate{\Gate Z} \& \qw            \& \gate{\Gate Z} \& \qw            \& \gate{\Gate Z} \& \gate{\Gate H} \& \ctrl{1}       \& \qw            \& \qw \\
 \& \gate{\Gate H} \& \gate{\Gate Z} \& \qw            \& \gate{\Gate Z} \& \qw            \& \gate{\Gate Z} \& \qw            \& \gate{\Gate Z} \& \gate{\Gate H} \& \qw 
\end{tikzcd}
};
\end{tikzpicture}\stackrel{\text{\cref{alg:hadamard}}}{\Longrightarrow}
\begin{tikzpicture}[baseline=(X.base)]
\node[scale=.65, ampersand replacement=\&] (X)	{
\begin{tikzcd}
 \& \qw            \& \ctrl{4}       \& \qw            \& \qw            \& \qw            \& \qw            \& \qw            \& \qw            \& \qw \\
 \& \qw            \& \qw            \& \qw            \& \ctrl{3}       \& \qw            \& \qw            \& \qw            \& \qw            \& \qw \\
 \& \gate{\Gate H} \& \gate{\Gate Z} \& \gate{\Gate H} \& \qw            \& \ctrl{2}       \& \qw            \& \qw            \& \qw            \& \qw \\
 \& \gate{\Gate H} \& \gate{\Gate Z} \& \qw            \& \gate{\Gate Z} \& \gate{\Gate Z} \& \gate{\Gate H} \& \ctrl{1}       \& \qw            \& \qw \\
 \& \gate{\Gate H} \& \gate{\Gate Z} \& \qw            \& \gate{\Gate Z} \& \gate{\Gate Z} \& \qw            \& \gate{\Gate Z} \& \gate{\Gate H} \& \qw 
\end{tikzcd}
};
\end{tikzpicture}
\end{equation*}
\cref{alg:hadamard} cancels two Hadamard gates on the second qubit which is not targeted by the first fan-out gate, i.e.\ $T_{\Gate{CZ}}[2,1] = 0$. On the third qubit, \cref{alg:hadamard} moves the rightmost Hadamard gate one layer to the left since $T_{\Gate{CZ}}[3,2] = 0$ but $T_{\Gate{CZ}}[3,1] = 1$.

\cref{alg:CZ} takes the output of \cref{alg:hadamard} and splits $\Gate{GCZ}_1$ into $\Gate{CZ}_1$ and $\Gate{GCZ}^\prime_1$ since it has Hadamard gates to the left and right.
Then it pools $\Gate{GCZ}^\prime_1$, $\Gate{GCZ}_2$ and $\Gate{GCZ}_3$ together into a single gate $\Gate{GCZ}_5$:
\begin{equation*}
\begin{tikzpicture}[baseline=(X.base)]
\node[scale=.65, ampersand replacement=\&] (X)	{
\begin{tikzcd}
 \& \qw\slice{}    \& \ctrl{4}\slice{$\Gate{GCZ}_1 \qquad\quad $} \& \qw\slice{}            \& \qw\slice{$\Gate{GCZ}_2 \qquad\quad $}      \& \qw\slice{$\Gate{GCZ}_3 \qquad\quad $} \& \qw\slice{} \& \qw\slice{$\Gate{GCZ}_4 \qquad\quad $}            \& \qw            \& \qw \\
 \& \qw            \& \qw            \& \qw            \& \ctrl{3}       \& \qw            \& \qw            \& \qw            \& \qw            \& \qw \\
 \& \gate{\Gate H} \& \gate{\Gate Z} \& \gate{\Gate H} \& \qw            \& \ctrl{2}       \& \qw            \& \qw            \& \qw            \& \qw \\
 \& \gate{\Gate H} \& \gate{\Gate Z} \& \qw            \& \gate{\Gate Z} \& \gate{\Gate Z} \& \gate{\Gate H} \& \ctrl{1}       \& \qw            \& \qw \\
 \& \gate{\Gate H} \& \gate{\Gate Z} \& \qw            \& \gate{\Gate Z} \& \gate{\Gate Z} \& \qw            \& \gate{\Gate Z} \& \gate{\Gate H} \& \qw 
\end{tikzcd}
};
\end{tikzpicture}\stackrel{\text{\cref{alg:CZ}}}{\Longrightarrow}
\begin{tikzpicture}[baseline=(X.base)]
\node[scale=.65, ampersand replacement=\&] (X)	{
\begin{tikzcd}
 \& \qw\slice{}    \& \ctrl{2}\slice{$\Gate{CZ}_1 \qquad\quad $} \& \qw\slice{}            \& \ctrl{4}            \& \qw            \& \qw\slice{$\Gate{GCZ}_5 \qquad\qquad\qquad\qquad $} \& \qw\slice{}            \& \qw\slice{$\Gate{CZ}_2 \qquad\quad $}            \& \qw            \& \qw \\
 \& \qw            \& \qw            \& \qw            \& \qw            \& \ctrl{3}       \& \qw            \& \qw            \& \qw            \& \qw            \& \qw \\
 \& \gate{\Gate H} \& \gate{\Gate Z} \& \gate{\Gate H} \& \qw            \& \qw            \& \ctrl{2}       \& \qw            \& \qw            \& \qw            \& \qw \\
 \& \gate{\Gate H} \& \qw            \& \qw            \& \gate{\Gate Z} \& \gate{\Gate Z} \& \gate{\Gate Z} \& \gate{\Gate H} \& \ctrl{1}       \& \qw            \& \qw \\
 \& \gate{\Gate H} \& \qw            \& \qw            \& \gate{\Gate Z} \& \gate{\Gate Z} \& \gate{\Gate Z} \& \qw            \& \gate{\Gate Z} \& \gate{\Gate H} \& \qw 
\end{tikzcd}
};
\end{tikzpicture}
\end{equation*}
Note that any $\Gate{GCZ}$ gate is equivalent to a $\Gate{GZZ}$ up to single-qubit $\Gate{R_Z}$ rotations.
The total encoding cost of the compiled circuit is $2+n/2(n-1) = 12$ whereas the original circuit has the encoding cost $6+3+3+1 = 13$.
\end{example}

In this section we discussed how to implement the entangling operations of a Clifford unitary.
We only need one $\Gate{GZZ}$ gate to implement one fan-out gate, as one can see in \cref{eq:CZlayer}, \eqref{eq:fanoutParity} and Hadamard gate commutation, instead of two multi-qubit gates used in the compilation schemes in Refs.~\cite{maslov_use_2018, VanDeWetering20ConstructingQuantumCircuits}.
We further showed that we only require $n+1$ $\Gate{GZZ}$ and $n$ two-qubit $\Gate{CZ}$ gates to implement a Clifford unitary on $n$ qubits.
Due to the flexibility of our $\Gate{GZZ}$ gate, we further optimized the implementation of a directed $\Gate{CX}$ layer to reduce the encoding cost.

The compilation of directed CX circuits via Algorithms \ref{alg:hadamard} and \ref{alg:CZ} is available as a Python implementation on GitHub \cite{GitHub}.
\begin{figure*}[t]
\begin{center}
  \includegraphics[width=0.47\linewidth]{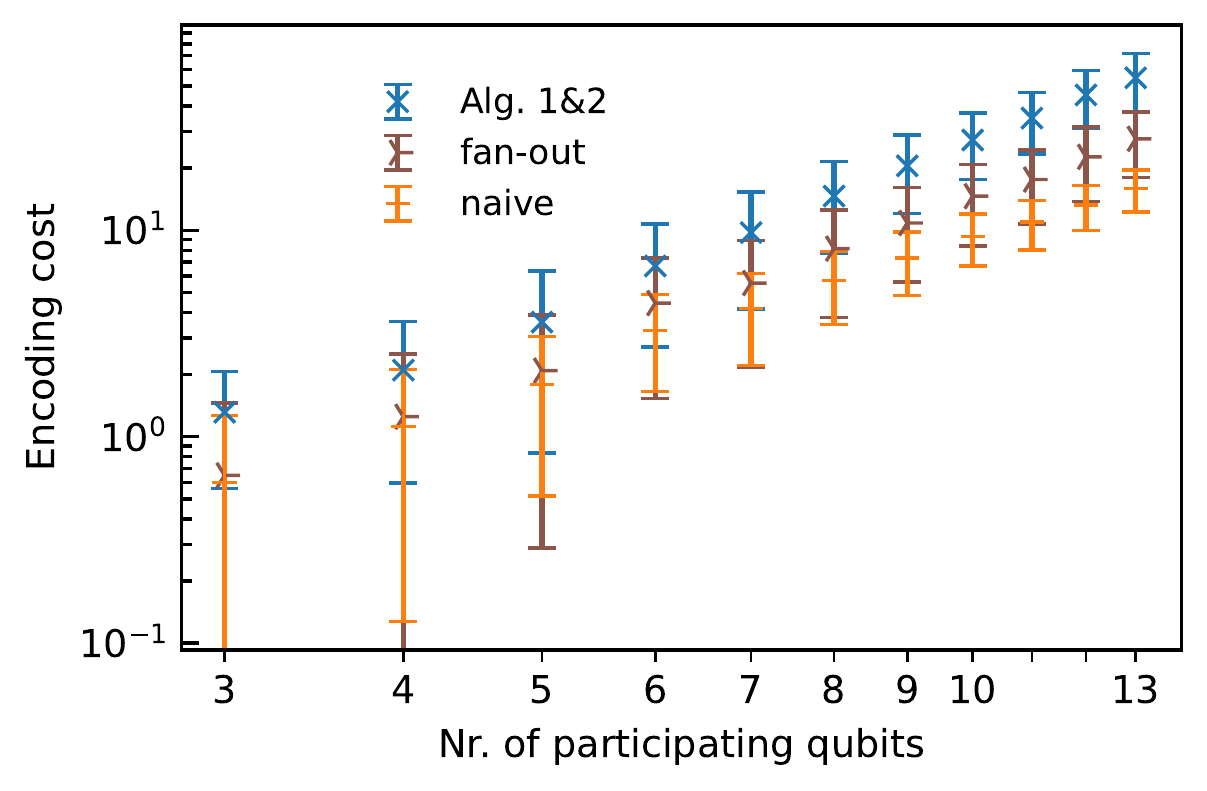}\hfill
  \includegraphics[width=0.47\linewidth]{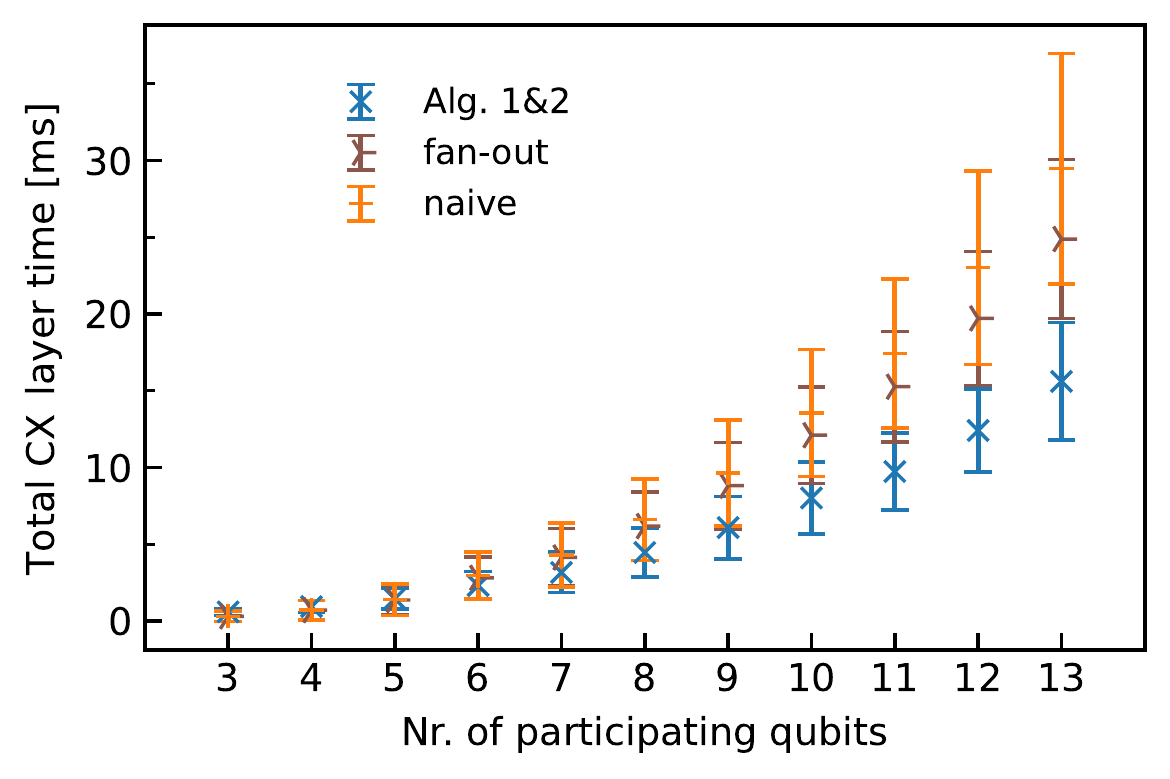}
  \\
  \includegraphics[width=0.47\linewidth]{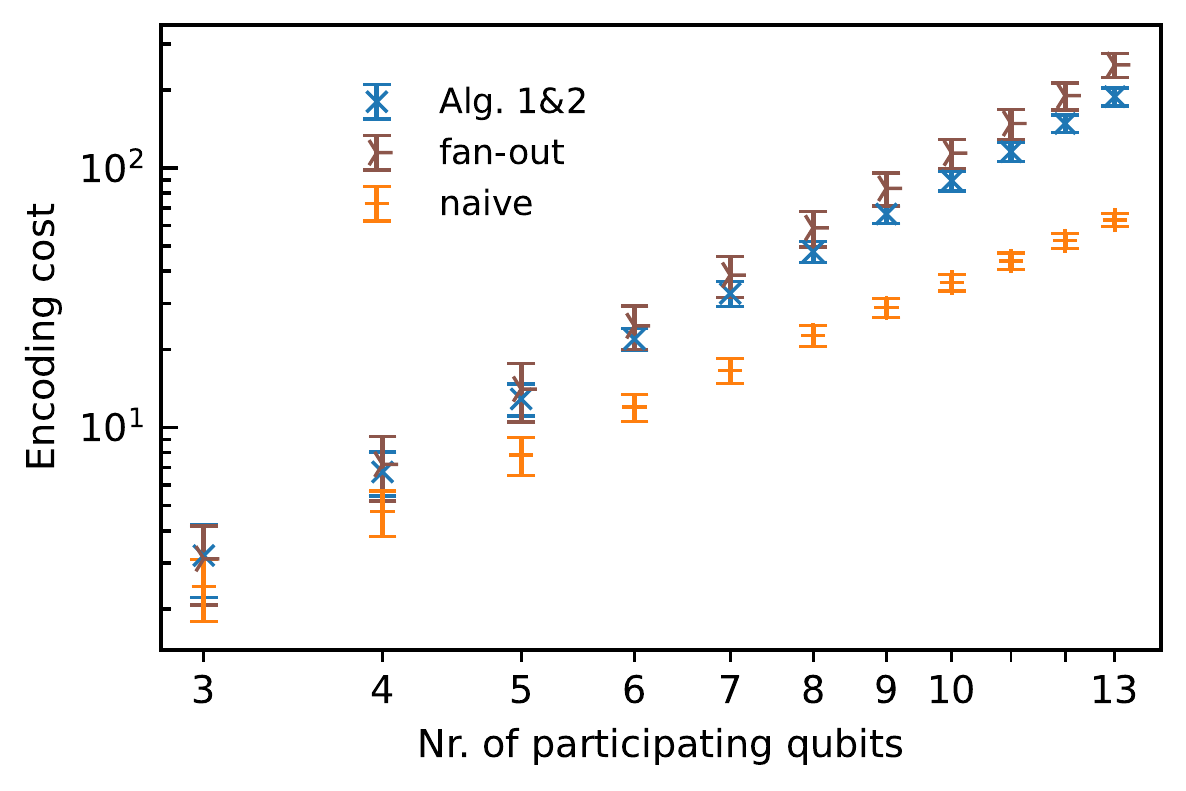}\hfill
  \includegraphics[width=0.47\linewidth]{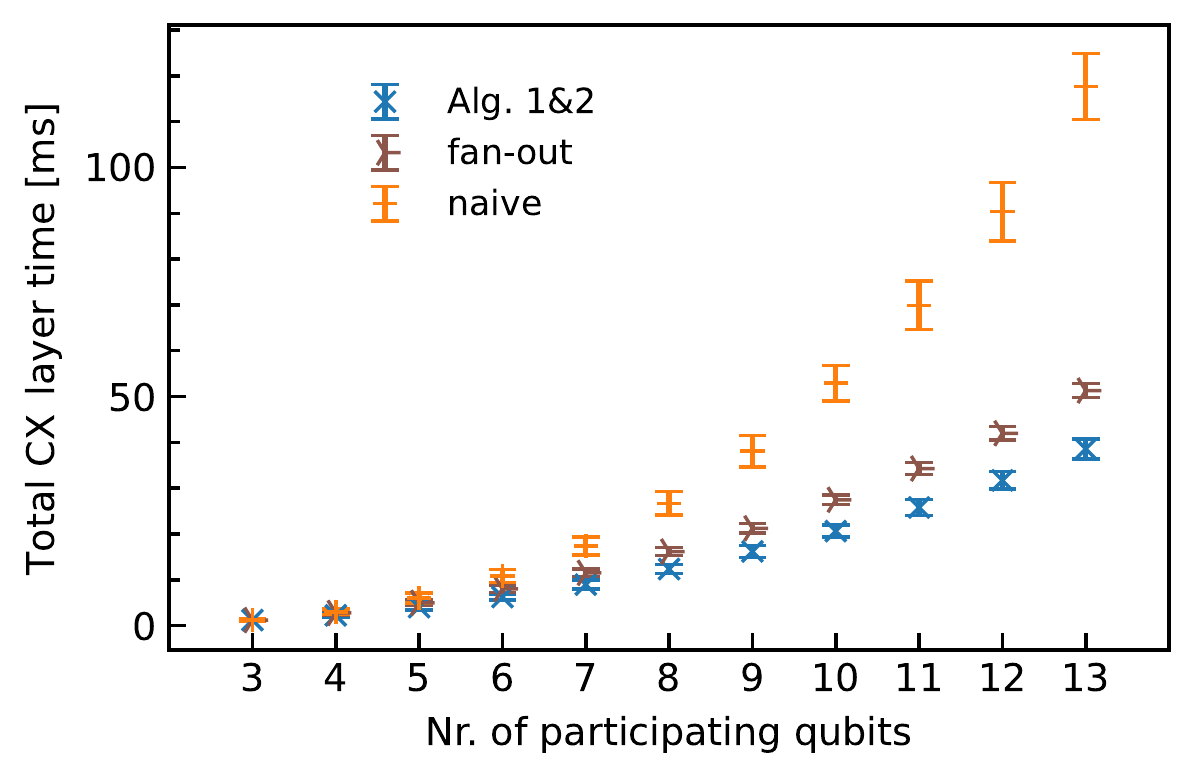}
  \caption{\label{fig:dirCX}
  Comparing the performances of \cref{alg:hadamard,alg:CZ}, fan-out and the naive approach for the implementation of a random directed $\Gate{CX}$ layer.
  The error bars show the variance over $100$ samples.
  The numerical values are obtained using the hardware specific parameters from \cref{sec:numerics}.
  \textbf{Top and bottom:} Sparse and dense $T_{\Gate{CZ}}$ with the probability $0.2$ and $0.8$ for picking a ``$1$'', respectively.
  \textbf{Left:} The encoding cost of a random directed $\Gate{CX}$ layer on $n$ qubits.
  In the top left the large variance in the encoding cost is due to the sparsity of $T_{\Gate{CZ}}$.
  In contrast, for dense $T_{\Gate{CZ}}$ (bottom left) \cref{alg:hadamard,alg:CZ} lead to a reduced encoding cost compared to the fan-out approach.
  \textbf{Right:} Neglecting local gates, for both sparse and dense $T_{\Gate{CZ}}$ the \cref{alg:hadamard,alg:CZ} yields the lowest total $\Gate{CX}$ layer time.
}
\end{center}
\end{figure*}

\subsubsection{Numerical results for the directed CX layer}
\label{sec:dir_CX_numerics}
We demonstrate the performance of \cref{alg:hadamard,alg:CZ} for compiling an arbitrary directed $\Gate{CX}$ layer.
Since the $\Gate{CX}$ layer is the most costly gate layer, in the Bruhat decomposition we only present the compilation of this layer.

Consider a directed $\Gate{CX}$ layer with randomly chosen entries of the lower triangular part of $T_{\Gate{CZ}}$, and $T_{\Gate H}$ as in \cref{eq:hadamard}.
We distinguish between the naive implementation, the implementation of the fan-out gates directly as a $\Gate{GZZ}$ gate (and local gates) and the application of \cref{alg:hadamard,alg:CZ} to $T_{\Gate{CZ}}$ and $T_{\Gate H}$.
Like in \cref{sec:numerics}, the naive implementation corresponds to the sequential execution of two-qubit $\Gate{ZZ}$ gates.
Therefore, the encoding cost and the total gate time for the naive approach is $\sum_{i < j} T_{\Gate{CZ}}[i,j]$ and $\sum_{i < j} T_{\Gate{CZ}}[i,j]/J_{ij}$, respectively.

Implementing the fan-out gate directly associates one $\Gate{GZZ}$ gate with each column of $T_{\Gate{CZ}}$.
We further take advantage of the sparsity of $T_{\Gate{CZ}}$.
If we consider the fan-out gate represented by the $i$-th column of $T_{\Gate{CZ}}$, then for all $j$ with $T_{\Gate{CZ}}[i,j] = 0$ we can exclude all the $j$-th qubits from the participation in that $\Gate{GZZ}$ gate, reducing the encoding cost.
Thus, for the fan-out gate implementation we have $n-1$ symmetric matrices $A$ with ones in the first row/column and zeros everywhere else.
The total encoding cost for the fan-out approach is therefore the sum of the encoding costs for the $n-1$ $\Gate{GZZ}(A)$ gates. The same holds for the total $\Gate{CX}$ layer time as the sum of the total $\Gate{GZZ}$ times.

\Cref{alg:hadamard,alg:CZ} take a different advantage of the sparsity of $T_{\Gate{CZ}}$ by commuting Hadamard gates and combining parts of multiple fan-out gates to one $\Gate{GZZ}$ gate.
As stated above, we need at most $\lfloor \frac{n-1}{2} \rfloor$ $\Gate{GZZ}$ gates.
If $k$ fan-out gates are combined, the resulting $\Gate{GZZ}$ gate is characterized by a symmetric matrix $A$ where the first $k$ rows (and also columns) can have non-zero values.
The total encoding cost for the \cref{alg:hadamard,alg:CZ} is the sum of the encoding costs of the $\lfloor \frac{n-1}{2} \rfloor$ $\Gate{GZZ}$ gates plus the encoding cost of the $\lceil \frac{n-1}{2} \rceil$ $\Gate{ZZ}$ gates.
The same holds for the total $\Gate{CX}$ layer time as the sum of the total $\Gate{GZZ}$ times and the times for the $\lceil \frac{n-1}{2} \rceil$ $\Gate{ZZ}$ gates.
Note that we neglect local gates in our considerations.

\Cref{fig:dirCX} shows that for a dense directed $\Gate{CX}$ layer \cref{alg:hadamard,alg:CZ} have a significant advantage in the total $\Gate{CX}$ layer time over the naive and the fan-out implementation.
Also, the encoding cost is reduced compared to the fan-out implementation.
For sparse $T_{\Gate{CZ}}$ the advantage is still visible in the total $\Gate{CX}$ layer time, but the difference between the approaches is less substantial.

\subsection{Quantum Fourier transform}
\label{sec:QFT}
The quantum Fourier transform is an essential ingredient in many quantum algorithms.
To define the corresponding unitary operator on an $n$-qubit register, we identify the elements of the computational basis $\ket{x_{1},\dots,x_{n}}$ for $x_j \in \FF_2$ with integers in binary representation, i.e.\ $x = \sum_{j=1}^{n} x_j 2^{n-j}$.
The \ac{QFT} is then given as
\begin{equation}
\label{eq:def-qft}
 \mathrm{QFT}\ket{x}
 = \frac{1}{2^{n/2}} \sum_{y=0}^{2^n-1} e^{2\pi \i xy 2^{-n}} \ket{y}
 = \frac{1}{2^{n/2}} \bigotimes_{j=1}^n\left( \ket{0} +  \e^{2\pi \i x 2^{-j}} \ket{1}\right).
\end{equation}
The latter form immediately leads to the efficient quantum circuit in \cref{fig:qft} which uses $n(n-1)/2$ controlled $\Gate{R_Z}$-rotations with angles $2\pi/2^j$ for $j=2,\dots,n$ and $n$ Hadamard gates.
For convenience, we introduce the shorthand notation $\Gate{R}_j \coloneqq \Gate{R_Z}(2\pi/2^j)$, in particular $\Gate{Z}\coloneqq\Gate{R_1}$, $\Gate{S}\coloneqq\Gate{R_2}=\sqrt{ \Gate{Z}}$, and $\Gate{T}\coloneqq\Gate{R_3} = \sqrt[4]{\Gate{Z}} $.
In the \ac{QFT} circuit in \cref{fig:qft} we can collect for each $j<n$ the subsequent controlled $\Gate{R}_{n-j}, \dots , \Gate{R}_1$ to a (multi-qubit) $\Gate{GCR_Z} (A)$ gate specified by the $(n-j+1) \times (n-j+1)$ symmetric matrix
\begin{equation}
 A = 2 \pi \begin{pmatrix}
 0        & 2^{-1} & \dots & 2^{-n+j}  \\
 2^{-1}   & 0       &  \dots     & 0  \\
 \vdots   & \vdots       & \ddots &   \\
 2^{-n+j} &  0      &       &  0 \\
\end{pmatrix} \, .
\end{equation}

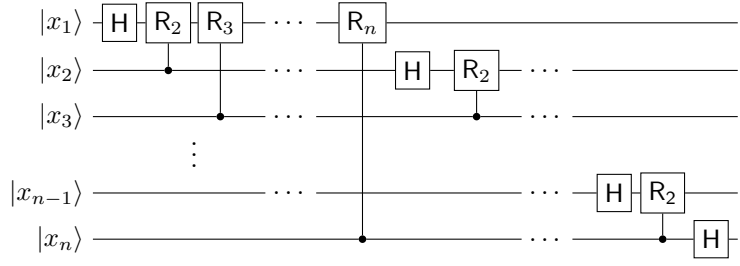
\begin{figure}
\begin{center}
\begin{tikzpicture}
 \begin{yquant}
  qubit {$\ket{x_{1}}$} q;
  qubit {$\ket{x_{2}}$} q[+1];
  qubit {$\ket{x_{3}}$} q[+1];
  [name=ypos]
  nobit n;
  qubit {$\ket{x_{n-1}}$} q[+1];
  qubit {$\ket{x_{n}}$} q[+1];

  box {$\Gate{H}$} q[0];
  [name=left]
  box {$\Gate{R}_2$} q[0] | q[1];
  [name=right]
  box {$\Gate{R}_3$} q[0] | q[2];
  hspace {2mm} -;
  text {$\dots$} q[0-3];
  hspace {2mm} -;
  box {$\Gate{R}_{n}$} q[0] | q[4];

  box {$\Gate{H}$} q[1];
  hspace {2mm} -;
  box {$\Gate{R}_2$} q[1] | q[2];
  hspace {2mm} -;
  text {$\dots$} q[1-4];
  hspace {2mm} -;
  box {$\Gate{H}$} q[3];
  box {$\Gate{R_2}$} q[3] | q[4];

  box {$\Gate{H}$} q[4];
 \end{yquant}
 \path (left |- ypos-0) -- (right |- ypos-0) node[midway,yshift=1mm] {$\vdots$};
\end{tikzpicture}
\end{center}
\caption{The quantum Fourier transform (with reversed order of the output qubits, i.e. without the swapping gates at the end).}
\label{fig:qft}
\end{figure}

\begin{figure*}[t]
	\begin{center}
		\begin{tabular}{cc}
			\includegraphics[width=0.45\linewidth]{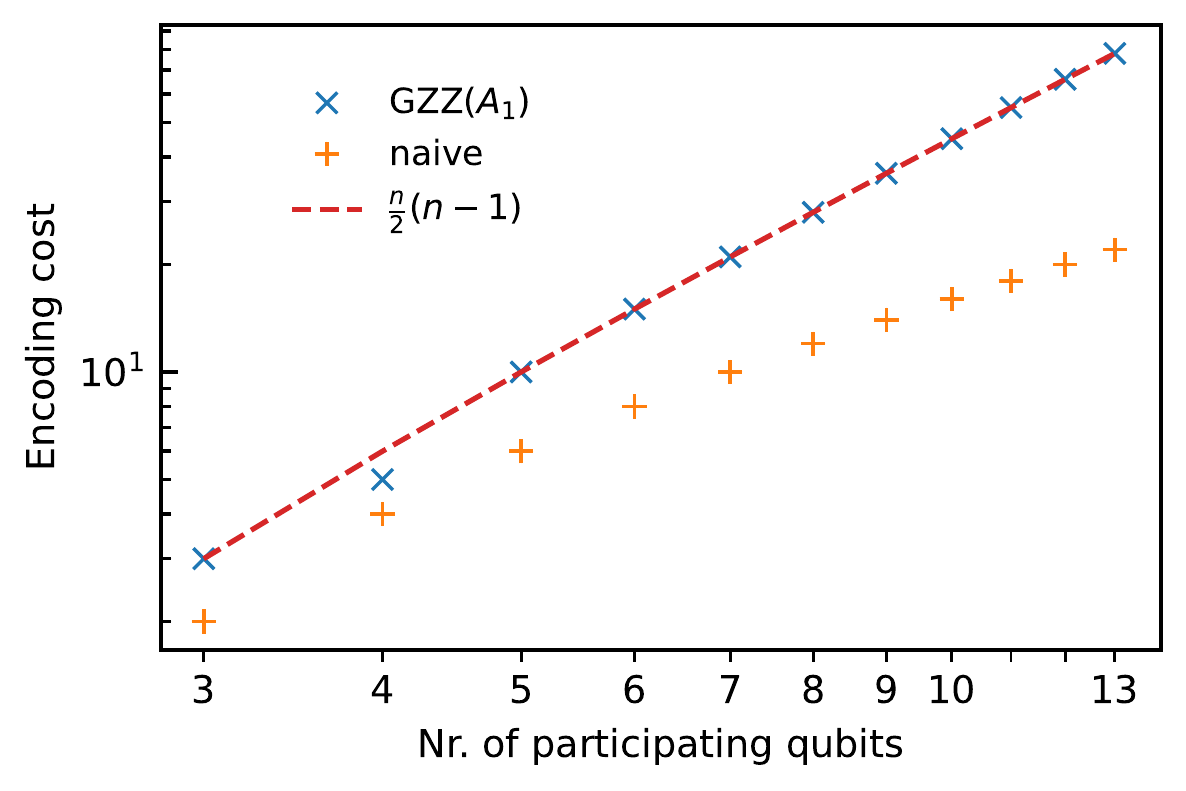}
			&
			\includegraphics[width=0.45\linewidth]{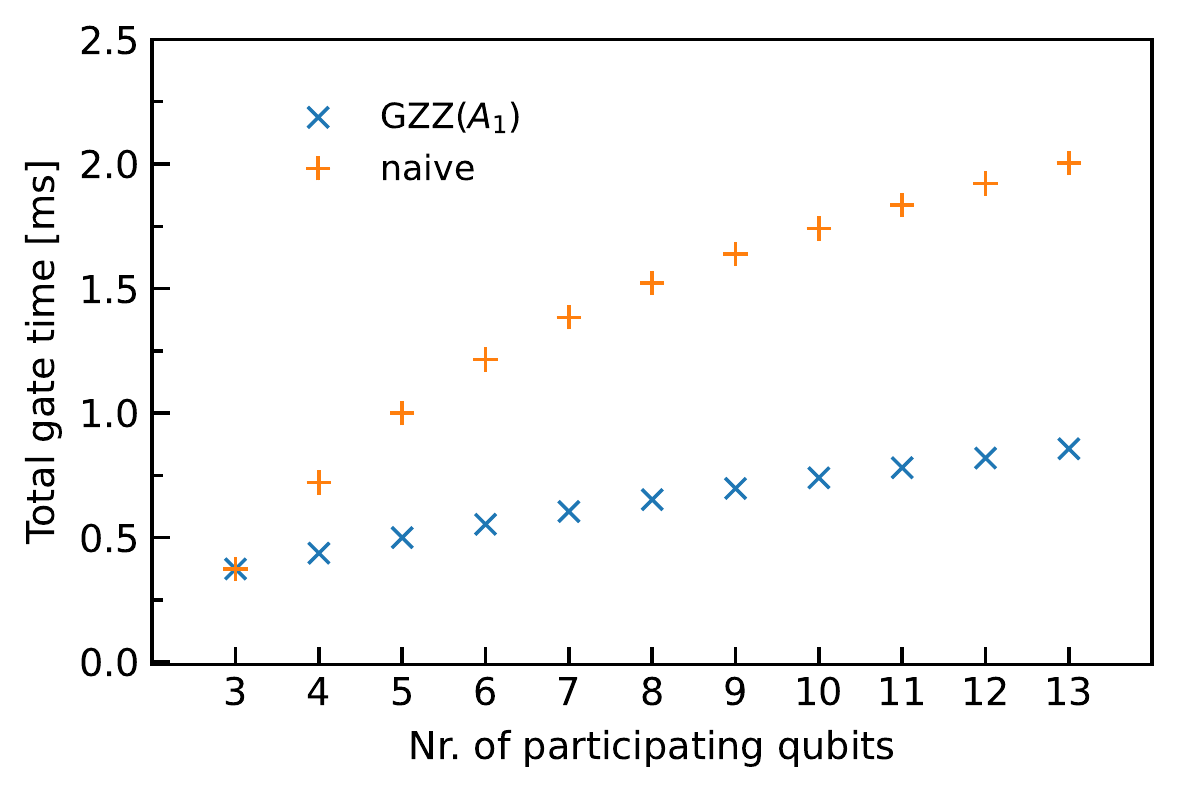}
		\end{tabular}
		\caption{\label{fig:numQFT} \textbf{Left: }The encoding cost, i.e.\ the different encodings,  we need to implement $\Gate{GZZ} (A_1)$ on $n$ qubits.
			For the naive approach the encoding cost only increases with $2(n-2)$, since that is the number of non-zero entries in $A_1$.
			\textbf{Right: }Due to the exponentially small non-zero entries in $A_1$ even the naive approach seems to have an asymptotic linear scaling of the total gate time.
			As before, the numerical values are obtained using the hardware specific parameters from \cref{sec:numerics}.
	  }
	\end{center}
\end{figure*}
Since our compilation scheme for the fully directed $\Gate{CX}$ layer in \cref{sec:dirCX} is agnostic of the rotation angles, up to local $\Gate{R_Z}$ rotations, we can apply it directly to the \ac{QFT} circuit.
This yields the compiled \ac{QFT} circuit:
\begin{equation*}
\scalebox{0.80}{
\begin{tikzpicture}
   \tikzset{
      operator/.append style={ minimum height=2em }
  }
 \begin{yquant}
  qubit {} q[4];
  [name=wave, register/minimum height=5mm]
  nobit wave;
  qubit {} q[+2];

  [name=left]
  box {$\Gate{H}$} q[0];
  box {$\Gate{S}$} q[0] | q[1];

  box {$\Gate{H}$} q[1];
  box {$\Gate{GCR_Z}(A_1)$} (-);

  box {$\Gate{H}$} q[2];
  [name=left]
  box {$\Gate{S}$} q[2] | q[3];

  box {$\Gate{H}$} q[3];
  box {$\Gate{GCR_Z}(A_3)$\\\\\\\\} (q[2]-q[5]);
  hspace {2mm} -;
  text {$\dots$} q[2-5];
  hspace {2mm} -;
  box {$\Gate{S}$} q[4] | q[5];

  box {$\Gate{H}$} q[5];
 \end{yquant}
 \node[wave, fit=(wave) (current bounding box.east |- wave), inner ysep=.7pt, inner xsep=0pt] {};
\end{tikzpicture}
}
\end{equation*}
Note that we are able to achieve this form without any Hadamard gate obstructing the required movement of $\Gate{CR_Z}$~gates.
The combination of diagonal multi-qubit gates across non-diagonal local gates has been carried out for small systems e.g. in Refs.~\cite{piltz2016versatile, ivanov2015simplified}.
However, this approach requires numerics in unitary groups beyond $\mathsf{SU}(2)$, which we avoid to ensure scalability (w.r.t.\ the Hilbert space dimension).
We also do not address the compilation of local $\Gate{R_Z}$ rotations here, see e.g.\ Ref.~\cite{kliuchnikov_shorter_2022} which covers that topic.

After our compilation scheme we have $\lceil \frac{n-1}{2} \rceil$ $\Gate{CS}$ gates (controlled-$\sqrt{ \Gate{Z}}$ gates) and $\lfloor \frac{n-1}{2} \rfloor$ $\Gate{GCR_Z}$ gates characterized by the symmetric matrix
\begin{equation}
\label{eq:A_k}
 A_j \coloneqq 2 \pi \begin{pmatrix}
 0        & 0          & 2^{-2} & \cdots & 2^{-n+j} \\
 0        & 0          & 2^{-1} & \cdots & 2^{-n+j+1} \\
 2^{-2}   & 2^{-1}     & 0      & \cdots & 0 \\
 \vdots   & \vdots     & \vdots & \ddots &  \\
 2^{-n+j} & 2^{-n+j+1} & 0      &        & 0 \\
\end{pmatrix} \, .
\end{equation}
A $\Gate{GCR_Z}(A_j)$ gate can be implemented as described in \cref{sec:gate_notation} using a $\Gate{GZZ}(A_j)$ and additional single-qubit $\Gate{R_Z}$ gates on the qubits they act on.
For the \ac{QFT}, we can push the resulting $\Gate{R_Z}$ gates on the target qubits to the end of circuit, and likewise we can push the ones on the control qubits to the beginning.

In conclusion, we can apply the fully directed $\Gate{CX}$ layer compiling scheme to implement a \ac{QFT} circuit with the same amount of $\GZZ$ gates and the same encoding cost.
The most difficult part in the \ac{QFT} are the local $\Gate{R_Z}$ rotations with exponentially small angles.
Due to the small values in the matrices $A_j$ one can also expect a practical issue in the implementation with too small $\lambda_{\vec m}$, see \cref{sec:limits}.

\subsubsection*{Numerical results for the \texorpdfstring{\ac{QFT}}{QFT}}
\label{sec:QFT_numerics}
We provide the numerical results for the performance of the $\Gate{GZZ} (A_j)$ gate with $A_j$ from \cref{eq:A_k} in \cref{fig:numQFT}.
It is a priori unclear how the exponentially small entries in $A_j$ effect the time spent in each encoding $ \lambda_{\vec m}$.
We set $j = 1$ since $\Gate{GZZ} (A_1)$ acts on $n$ qubits and is therefore the most costly $\Gate{GZZ}$ gate in the \ac{QFT}.
As in \cref{sec:numerics} we compare the result of the $\Gate{GZZ}$ gate implementation via the \ac{LP}~\eqref{eq:LP1} against the naive approach, i.e.\ the sequential implementation of two-qubit controlled $\Gate{R}_j$ gates.
As before, we neglect the finite recoding time, i.e.\ the time for executing the $\Gate{X}$ gate layers, and only consider the time needed to execute a $\Gate{GZZ} (A_1)$ gate or the naive approach.
Note that in \cref{fig:numQFT} we consider only a single $\Gate{GZZ} (A_1)$ gate, while the total \ac{QFT} circuit consists of $\lfloor \frac{n-1}{2} \rfloor$ $\Gate{GZZ} (A_j)$ gates.

Due to the exponentially small entries in $A_1$ we might in practice run into the problem of too small $\lambda_{\vec m}$ for a \ac{QFT} on a moderate amount of qubits.

\subsection{Circuit reduction for quantum chemistry applications}
\label{sec:chemistry}
Simulating molecular dynamics is probably one of the main applications for quantum computations.
Molecular dynamics are governed by the Coulomb Hamiltonian, consisting of the kinetic energy terms and of the Coulomb interactions between the electrons and the nuclei.
It is common to neglect the kinetic terms of the nuclei which is called ``Born-Oppenheimer approximation'' \cite{BO}.

The determination of approximate ground and excited eigenstates of the remaining electronic Hamiltonian is a hard task even on a quantum device \cite{Kempe04TheComplexityOf}.
Since simulating the time-dependent Schrödinger equation is more natural on a quantum computer there has been much effort to approximate eigenstates by solving the time-dependent Schrödinger equations \cite{parrish_quantum_filter_2019,klymko2021real,stair_quantum_krylov_2020}.
The corresponding time evolution can be approximated by factorizing the time dependent interactions into $m-1$ layers \cite{Cohn21QuantumFilter}:
\begin{equation}
 U_t \approx U_{Ext} \Gate{\hat G}(\varphi_{m}) \left[ \prod_{k=1}^{m-1} \Gate{GZZ}(A_k) \Gate{ \hat G}(\varphi_{k}) \right] \Gate{\hat R_Z}(\theta_{0,1}) \Gate{\hat G}(\varphi_{0}) \, ,
\end{equation}
where the hat denotes layers of the gates and $\Gate{G}(\varphi)$ is a Givens rotation on two qubits.
$U_{Ext}$ represents the constant nucleus-nucleus Coulomb interaction and corresponds to a global phase which we henceforth omit.
$A_k \in \HTL (\RR^n)$ denotes the total coupling matrix, characterizing the $\Gate{GZZ}$ gate.
The right-hand side can then be represented as the circuit
\begin{equation}
\label{eq:qchemistry}
\scalebox{0.85}{
\begin{tikzpicture}
 \begin{yquant}
  qubit {} q[2];
  [name=ypos]
  nobit n;
  qubit {} q[+2];

  box {$\Gate{G}(\varphi_0)$} (q[0],q[1]);
  box {$\Gate{G}(\varphi_0)$} (q[2],q[3]);
  [name=left]
  box {$\Gate{R_Z}(\theta_0)$} q[0];
  [name=right]
  box {$\Gate{R_Z}(\theta_1)$} q[1];
  box {$\Gate{R_Z}(\theta_0)$} q[2];
  box {$\Gate{R_Z}(\theta_1)$} q[3];

  align q;
  box {$\Gate{G}(\varphi_1)$} (q[0],q[1]);
  box {$\Gate{G}(\varphi_1)$} (q[2],q[3]);

  box {$\Gate{GZZ}(A_1)$} (-);
  [name=left2]
  box {$\Gate{G}(\varphi_2)$} (q[0],q[1]);
  [name=right2]
  box {$\Gate{G}(\varphi_2)$} (q[2],q[3]);

  box {$\Gate{GZZ}(A_2)$} (-);
  output {$\dots$} q;

 \end{yquant}
 \path (left |- ypos-0) -- (right |- ypos-0) node[midway,yshift=1mm] {$\vdots$};
 \path (left2 |- ypos-0) -- (right2 |- ypos-0) node[midway,yshift=1mm] {$\vdots$};
\end{tikzpicture}
} \, .
\end{equation}
The implementation of the $\Gate{GZZ}$ gates is straightforward, see \cref{sec:J_shaping}, and we again do not address the compilation of local $\Gate{R_Z} (\theta)$ rotations in this work.
Therefore, we focus on the decomposition of the Givens rotation
\begin{equation}
\label{eq:twoQubitGivens}
\Gate{G}(\varphi) = \begin{pmatrix}
 1 & 0            & 0           & 0 \\
 0 & \cos (\varphi)  & \sin (\varphi) & 0 \\
 0 & -\sin (\varphi) & \cos (\varphi) & 0 \\
 0 & 0            & 0           & 1 \\
\end{pmatrix}
\equiv
\begin{tikzpicture}[baseline=(X.base)]
\node[scale=.75, ampersand replacement=\&] (X)	{
\begin{tikzcd}[column sep=tiny]
 \qw \& \gate{\Gate S} \& \qw            \& \targ{}   \& \gate{\Gate{R_Y} (-\varphi)} \& \targ{}    \& \qw            \& \gate{\Gate S^\dagger} \& \qw \\
 \qw \& \gate{\Gate S} \& \gate{\Gate H} \& \ctrl{-1} \& \gate{\Gate{R_Y} (-\varphi)} \& \ctrl{-1}  \& \gate{\Gate H} \& \gate{\Gate S^\dagger} \& \qw \\
\end{tikzcd}
};
\end{tikzpicture}
\end{equation}
where $\Gate{R_Y} (\varphi) = e^{-i Y \varphi / 2}$ is a rotation around the $y$-axis.
Using the Euler decomposition, up to global phases we can express $\Gate{R_Y} (\varphi)$ as
\begin{equation}
\label{eq:Yeulerdecomp}
 \Gate{R_Y} (\varphi) \equiv \sqrt{\Gate X}^\dagger \Gate{R_Z} (\varphi) \sqrt{\Gate X} \equiv \Gate S^\dagger \Gate{R_X} (\varphi) \Gate S \,.
\end{equation}
Inserting this into \cref{eq:twoQubitGivens} we obtain
\begin{equation}
 \Gate{G}(\varphi) \equiv \scalebox{0.75}{
\begin{quantikz}[column sep=tiny]
 \qw & \gate{\Gate S} & \qw            & \gate{\sqrt{\Gate X}} & \targ{}   & \gate{\Gate{R_Z} (-\varphi)} & \targ{}   & \gate{\sqrt{\Gate X}^\dagger}  & \qw            & \gate{\Gate S^\dagger} & \qw \\
 \qw & \gate{\Gate S} & \gate{\Gate H} & \gate{\Gate S}        & \ctrl{-1} & \gate{\Gate{R_X} (-\varphi)} & \ctrl{-1} & \gate{\Gate S^\dagger}  & \gate{\Gate H} & \gate{\Gate S^\dagger} & \qw
\end{quantikz}} \, .
\end{equation}
Using the commutation rule
\begin{equation}
\label{eq:CX_Z_ZZ_CX}
\Gate{CX}_{2,1} \Gate{R_Z} (\varphi)_1 \equiv \Gate{ZZ}_{1,2} (\varphi) \Gate{CX}_{2,1} \,
\end{equation}
twice, first we commute the $\Gate{R}_Z$ gate with the right $\Gate{CX}$, then conjugate the $\Gate{R}_X$ gate and the left $\Gate{CX}$ with Hadamard gates and commute them.
This gives
\begin{equation}
\label{eq:compiled_givens}
 \Gate{G}(\varphi) \equiv
 \begin{tikzpicture}[baseline=(X.base)]
\node[scale=.75, ampersand replacement=\&] (X)	{
\begin{tikzcd}[column sep=tiny]
 \qw \& \gate{\Gate S} \& \qw            \& \gate{\sqrt{\Gate X}} \& \gate{\Gate H} \& \gate[2]{\Gate{ZZ}(-\varphi)} \& \gate{\Gate H} \& \gate[2]{\Gate{ZZ}(-\varphi)} \& \gate{\sqrt{\Gate X}^\dagger}  \& \qw            \& \gate{\Gate S^\dagger} \& \qw \\
 \qw \& \gate{\Gate S} \& \gate{\Gate H} \& \gate{\Gate S}        \& \gate{\Gate H} \&  \& \gate{\Gate H} \& \& \gate{\Gate S^\dagger}  \& \gate{\Gate H} \& \gate{\Gate S^\dagger} \& \qw \\
\end{tikzcd}
};
\end{tikzpicture}
\, .
\end{equation}
Since any layer of $\Gate{ZZ}$ gates can be implemented with only one $\Gate{GZZ}$ gate, see \cref{eq:GZZ_to_ZZ}, we can collect the parallel $\Gate{ZZ}(-\varphi)$ gates of the Givens rotation layer $\Gate{\hat G}(\varphi)$ into a single $\Gate{GZZ} (-\varphi A_{NN})$ gate.
Overall, this results in two $\Gate{GZZ}$ gates per Givens rotation layer.
Here, $A_{NN} \in \HTL(\FF_2^n)$ denotes the total coupling matrix that pairwise couples next-neighbor qubits, i.e.\ qubits 1 and 2, qubits 3 and 4, etc.
Concretely, $A_{NN}$ takes the form
\begin{equation}
A_{NN} = \begin{pmatrix}
		0&1&&&&&\\
		1&0&&&&&\\
		&&&\ddots&&&\\
		&&&&&0&1\\
		&&&&&1&0
	\end{pmatrix}. 
\end{equation}
Usually, the subdiagonal of the hardware-given coupling matrix $J$ contains the largest values and thus corresponds to the fastest $\Gate{ZZ}$ gates.
This results in very fast $\Gate{GZZ}(-\varphi A_{NN})$ gates.

A straightforward implementation of the Givens rotation layer with \cref{eq:twoQubitGivens} results in $n$ local $\Gate{R_Y}$ rotations with arbitrary angle.
On a quantum computer with a finite native gate set the rotations with arbitrary angle might be challenging to implement.
Although there are methods to decompose rotations with arbitrary angle into a native gate set, see e.g.\ Refs.~\cite{kliuchnikov_fast_2013, ross_optimal_2016, bouland_efficient_2021}, they always come with a considerable encoding cost.

To sum up, to approximate molecular dynamics as described by the circuit \eqref{eq:qchemistry} one needs $2(m+1)$ $\Gate{GZZ} (-\varphi_k A_{NN})$ gates, $(m-1)$ $\Gate{GZZ} (A_k)$ gates, $n$ local $\Gate{R_Z}(\theta_{0,1})$ gates and some local Clifford gate.

\subsection{Decomposing general diagonal unitaries}
\label{sec:diagunit}
Since the $\Gate{GZZ}$ gate is diagonal, it is natural to look for a compilation procedure to implement general diagonal unitaries using the $\Gate{GZZ}$ gate together with local $\Gate{R_Z}$ and Hadamard gates.
We assume that we are able to implement $\Gate{R_Z}$ gates with arbitrarily small angles which in practice might be difficult.
Yet, there are algorithms to approximate any $\Gate{R_Z}$-angle using only Clifford+T gates \cite{ross_optimal_2016} or Clifford+$\sqrt{T}$ \cite{kliuchnikov_shorter_2022}.
This compilation scheme leads us beyond Clifford circuits
and should be understood as a stepping stone towards more general compilation schemes. 

First we introduce the representation of diagonal unitaries by phase polynomials \cite{amy_cnot_complexity_2018,amy_polynomial_time_2014}.
Then we show how certain terms of the phase polynomial can be implemented in parallel.
We optimize the order of these parallelized layers such that the encoding cost and the total $\Gate{GZZ}$ time is reduced.

Any diagonal $n$-qubit unitary acts as
\begin{equation}
\label{eq:diag_unitary}
 U_{f} \ket{\vec x} \coloneqq \e^{2 \pi \i f( \vec x )} \ket{\vec x},
\end{equation}
with $\vec x \in \FF_2^n$ for some pseudo-Boolean function $f: \FF_2^n \rightarrow \R$.
Such a function can be uniquely expressed as a multilinear \emph{phase polynomial} \cite[Theorem 1.1]{odonnell_analysis_2021_new},
\begin{equation}
\label{eq:FE}
 f(\vec x) = \sum_{\vec y \in \FF_2^n} \alpha_{\vec y} \chi_{\vec y} (\vec x) 
\end{equation}
with coefficients or phases $\alpha_{\vec y} \in \RR$ and \emph{parity function} $x\mapsto \chi_{\vec y}(x)=y_1 x_1 \oplus \dots \oplus y_n x_n$ for each \emph{parity} $\vec y \in \FF_2^n$.
Additionally taking into account basis state transformations $\ket{\vec x}\mapsto\ket{g(\vec x)}$ the above extends to the so-called \emph{sum-over-paths representation}, which was used in Refs.~\cite{amy_cnot_complexity_2018,amy_polynomial_time_2014}. 

Terms in the expansion of $f$ of the form $\alpha x_i$ and $\alpha (x_i \oplus x_j)$ can directly be implemented as $\Gate{R_Z}(\alpha)_i$ and $\Gate{ZZ} (\alpha)_{i,j}$ gate, respectively, and are therefore easy.
We thus only consider terms with more than two variables $x_i$.

We define the unitary representing one term in $f(\vec x)$ by
\begin{equation}
 U_{i, \vec y} \ket{\vec x} \coloneqq \e^{2 \pi \i \alpha_{\vec y} \chi_{\vec y} (\vec x)} \ket{\vec x} \, ,
\end{equation}
where $i$ is a qubit index.
The role of the subscript $i$ gets apparent in the circuit representation of the unitary $U_{i, \vec y}$
\begin{equation}
\label{eq:one_fourier_term}
\begin{aligned}
U_{i, \vec y} &\equiv
\begin{tikzpicture}[baseline=(X.base)]
\node[scale=.75, ampersand replacement=\&] (X)	{
\begin{tikzcd}[column sep=tiny]
\lstick{$x_{l_1}$}     \& \ctrl{3} \& \qw \ldots  \& \qw      \& \qw                                         \& \qw      \& \qw \ldots  \& \ctrl{3} \& \qw \\
\vdots                 \&          \&             \&          \&                                             \&          \&             \&          \&\\
\lstick{$x_{l_{m-1}}$} \& \qw      \& \qw \ldots  \& \ctrl{1} \& \qw                                         \& \ctrl{1} \& \qw \ldots  \& \qw      \& \qw \\
\lstick{$x_{i} = x_{l_m}$}     \& \targ{}  \& \qw \ldots  \& \targ{}  \& \gate{\Gate{R_Z} (\alpha_{\vec y})} \& \targ{}  \& \qw \ldots  \& \targ{}  \& \qw
\end{tikzcd}
};
\end{tikzpicture}
\\
&=\begin{tikzpicture}[baseline=(X.base)]
\node[scale=.75, ampersand replacement=\&] (X)	{
\begin{tikzcd}[column sep=tiny]
\lstick{$x_{l_1}$}     \& \qw            \& \ctrl{3}   \& \qw \ldots \& \qw        \& \qw              \& \qw                                         \& \qw       \& \qw        \& \qw \ldots \& \ctrl{3}   \& \qw      \& \qw \\
\vdots                 \&                \&            \&            \&            \&                  \&                                             \&           \&            \&            \&            \&          \&     \\
\lstick{$x_{l_{m-1}}$} \& \qw            \& \qw        \& \qw \ldots \& \ctrl{1}   \& \qw              \& \qw                                         \& \qw       \& \ctrl{1}   \& \qw \ldots \& \qw        \& \qw      \& \qw \\
\lstick{$x_{i} = x_{l_m}$}     \& \gate{\Gate H} \& \control{} \& \qw \ldots \& \control{} \& \gate{\Gate H}   \& \gate{\Gate{R_Z} (\alpha_{\vec y})} \& \gate{\Gate H}  \& \control{} \& \qw \ldots \& \control{} \& \gate{\Gate H} \& \qw
\end{tikzcd}
};
\end{tikzpicture}
 \,
\end{aligned}
\end{equation}
where $l = (k \in [n]: y_k \neq 0)\in [n]^m$ is a sequence containing the $m$ non-zero components of $\vec y$ and with the conjugating Hadamard gates, the $\Gate{R_Z}$ gate and all the $\Gate{CX}$ gate targets on the $i$-th qubit.
Note that we omit $\alpha_{\vec y}$ in $U_{i, \vec y}$ since the following discussion is agnostic of the phase $\alpha_{\vec y}$.
The domain of $U_{i, \vec y}$ is given by the support of the parity $\supp (\vec y) \coloneqq \Set*{ k \given y_k \neq 0 }$.
Since the unitaries of the two $\Gate{CX}$ layers conjugating the $\Gate{R_Z}$ gate in $U_{i, \vec y}$ are permutation matrices one has the freedom in choosing any $i \in \supp (\vec y)$. 
That is, for any $i,j\in\supp(\vec y)$, $U_{i,\vec y}$ and $U_{j,\vec y}$ implement the same unitary.
Applying \cref{eq:CZlayer} to the right-hand side of \cref{eq:one_fourier_term}, we find that each $U_{i, \vec y}$ can be implemented with two $\Gate{GZZ}$ gates, one $\Gate{R_Z} (\alpha_{\vec y})$ gate, four Hadamard gates and some $\Gate{R_Z} (\pi/2)$ gates.

Clearly, two such diagonal unitaries $U_{i,\vec y}$ and $\U_{i,\vec y^\prime}$ can be implemented in parallel if their supports are disjoint, i.e.\ if $\supp(\vec y) \cap \supp(\vec y^\prime) = \emptyset$.
Since our $\Gate{GZZ}$ gates are time-optimal, this results in a time-optimal implementation of those unitaries, assuming as above a time-optimal implementation of the single-qubit $\Gate{R_Z}$ gates.
Denote by $\mc S \coloneqq \Set*{ \supp (\vec y) \given \vec y \in \FF_2^n , |\vec y|>2 }$ the set of all supports of the parities without the easy terms.
\Cref{alg:par} is a heuristic algorithm, which parallelizes $U_{i, \vec y}$ with disjoint supports such that the resulting support is as large as possible.
\begin{algorithm}[H]
\caption{Parallelizing supports.}\label{alg:par}
\textbf{Input:} $\mc S$ \Comment{Set of supports}
\begin{algorithmic}[0]
\State $\mc{L} \gets \emptyset$
\While{$\mc S \neq \emptyset$}
    \State Choose $s_k \in \mc S$
    \State $\mc S \gets \mc S \setminus \{ s_k \}$  \Comment{Remove $s_k$ from $\mc S$}
    \State $s \gets \{ s_k \}$
    \While{There exists $s_i \in \mc S$ s.t. $s \cap s_i = \emptyset$ and $\max_{s_i \in \mc S} |s \cup s_i|$}
        \State $s \gets s \cup \{ s_i \}$  \Comment{Append disjoint support with maximal union size}
        \State $\mc S \gets \mc S \setminus \{ s_i \}$
    \EndWhile
    \State $\mc{L} \gets \mc{L} \cup s$  \Comment{Append to set of parallelized layers}
\EndWhile
\end{algorithmic}
\textbf{Output:} $\mc L$ \Comment{Set of parallelized layers}
\end{algorithm}

Between any two consecutive unitaries $U_{i,\vec y}$ and $\U_{i',\vec y^\prime}$ with support overlap $o \coloneqq \supp (\vec y) \cap \supp (\vec y^\prime) = s \cap s^\prime\neq\emptyset$ two types of cancellation can happen.
First, the Hadamard gates cancel, provided that they act on the same qubit.
This is achieved by choosing the qubit $i$ on which the Hadamard gates act, concretely by demanding $i = i^\prime \in o$.
Thereafter, one can combine the two $\Gate{CZ}$ layers of $U_{i,\vec y}$ and $\U_{i',\vec y^\prime}$ and implement them with just one $\Gate{GZZ}$ gate, see \cref{eq:CZlayer}.
Doing so, a second cancellation happens automatically:
One can see from the right-hand side of \cref{eq:one_fourier_term} that between $U_{i,\vec y}$ and $\U_{i,\vec y^\prime}$, $2|o|$ $\Gate{CZ}$ gates cancel and therefore $|o|$ less qubits participate in the $\Gate{GZZ}$ gate.
This leads to a quadratic reduction $\LandauO(|o|^2)$ of the encoding cost, which for a $n$ qubit $\Gate{GZZ}$ gate is $n(n-1)/2$.

Let $\mc{L}_r \coloneqq \Set*{ u \in \mc{L} \given |u|=r }$, where $\mc{L}$ is the set of all parallelized layers returned by \cref{alg:par}, be the subset of all parallelized layers that contain $r$ disjoint supports of unitaries $U_{i, \vec y}$.
The set $\mc{L}_r$ is important for the placement of the Hadamard gates.
For example, if two parallelized layers $u,u^\prime \in \mc{L}_r$ are executed consecutively and there exist $i$ and $j$ for all $s_1 , \dots , s_r \in u$ and $s_1^\prime , \dots , s_r^\prime \in u^\prime$ such that $s_i \cap s_j^\prime \neq \emptyset$, then all interleaving Hadamard gates cancel if they are placed on the qubits contained in $s_i \cap s_j^\prime$.

We now extend this argument to all parallelized layers $u^{(k)} \in \mc{L}_r$.
If there is a qubit contained in the repeated overlap of the $r$ supports of all consecutive parallelized layers, all Hadamard gates between the parallelized layers cancel.
Concretely, for the cancellation of all Hadamard gates there must exist indices $i_k \in \{1, \dots , r\}$ for all $s_1^{(k)} , \dots , s_r^{(k)} \in u^{(k)}$ such that $\bigcap_k^{|\mc{L}_r|} s_{i_k}^{(k)} \neq \emptyset$.
We further note that if we apply $u \in \mc{L}_r$ and $v \in \mc{L}_{r^\prime}$ successively for $r \neq r^\prime$, then it is impossible that all interleaving Hadamard gates cancel.

The next lemma gives the encoding cost for the set of parallelized layers $\mc{L}_r$ and deals with the placement of the Hadamard gates by introducing few ancilla qubits.
\begin{lemma}
\label{lem:Pr}
 The set of parallelized layers $\mc{L}_r$ can be implemented using $2r$ Hadamard gates, $|\mc{L}_r|r$ non-Clifford $\Gate{R_Z}$ rotations and $|\mc{L}_r|+1$ $\Gate{GZZ}$ gates by introducing at most $r$ ancilla qubits.
\end{lemma}
\begin{proof}
 Each $u^{(k)} \in \mc{L}_r$ contains $2r$ Hadamard gates.
 The position of the $2r$ Hadamard gates can be chosen freely on the $r$ supports $s_1^{(k)}, \dots , s_r^{(k)} \in u^{(k)}$.
 So we want to find qubits $q_i$ in the overlap of all supports, i.e.\ $q_i \in \bigcap_k^{|\mc{L}_r|} s_{i_k}^{(k)}$ with for each $i = 1, \dots ,r$.
 If this is not possible for qubit $q_i$, i.e.\ $\bigcap_k^{|\mc{L}_r|} s_{i_k}^{(k)} = \emptyset$, we add an ancilla qubit to all supports containing $q_i$ and set $q_i$ to that qubit.
 One sees that in the worst case, setting all $q_1, \dots , q_r$ to the ancilla qubits and adding them to the supports $s_1, \dots , s_r$ for all $u \in \mc{L}_r$ results in the cancellation of all interleaving Hadamard gates, leaving only $2r$ Hadamard gates.
 The $\Gate{GZZ}$ gate count results from combining the two $\Gate{GZZ}$ gates of subsequent $u$'s as discussed above.
\end{proof}
Note that the $r$ ancilla qubits can be reused by different $\mc{L}_r$ since we do not encode any information on them.
For an $n$-qubit diagonal unitary one needs $r \leq \lfloor \frac{n}{3} \rfloor$ ancilla qubits to ensure that all Hadamard gates cancel.
The upper bound comes from the fact that we consider only supports with $|\vec y| > 2$ so at most $\lfloor \frac{n}{3} \rfloor$ many supports can be in parallel.

Assume now, that we have $r$ ancilla qubits or that all interleaving Hadamard gates cancel.
Then one can freely permute the order of the parallelized layers.
Such reordering does not change the $\Gate{GZZ}$ gate count, but it possibly reduces the amount of qubits participating in the $\Gate{GZZ}$ gates and therefore the encoding cost.
We thus want to find an order such that the cancellation of $\Gate{CZ}$ gates is maximized, i.e.\ we want to maximize the overlap of the supports of consecutive parallelized layers.
We define the \emph{shared support size} between two parallelized layers $u,u^\prime \in \mc{L}_r$ with $r$ elements as $S_{u,u^\prime} \coloneqq \sum_{i=0}^r |s_i \cap s_i^\prime|$, where $s_i \in u$ and $s_i^\prime \in u^\prime$.
Choose the ordering of elements in $\mc{L}_r$ that maximizes the shared support size $S_{u,u^\prime}$ over all pairs of parallelized layers $u,u^\prime \in \mc{L}_r$.
This relates to the traveling salesman problem in the formulation of Ref.~\cite{Dantzig1954} with each parallelized layer as a vertex and weights $S_{u,u^\prime}$ on the edge between $u$ and $u^\prime$.
Hence the reordering of the parallelized layers cannot be solved efficiently.
However, there are heuristic algorithms, e.g.\ the Christofides algorithm \cite{christofides_worstcase_1976} which only takes $\LandauO (|\mc{L}_r|^3)$ steps and guarantees that the solution is within a factor $3/2$ of the optimal solution.

In the worst case of the proposed compiling method for general diagonal unitaries, we consider the set of all supports $\mathcal S$ without easy gates, i.e., without single- and two-qubit gates.
This set has $|\mathcal S|=2^n-(n+n(n-1)/2) = 2^n - n(n+1)/2$ elements, and contains for every element $s\in\mc S$ also the complement $\bar s$ with $|s \cup \bar s| = n$.
Therefore, we can always parallelize two gates, the one corresponding to $s$ and that corresponding to $\bar s$.
We thus only have $\mc{L}_2$ which has size $|\mc L_2|=|\mc S|/2=2^{n-1} - n(n+1)/4$.
Using Lemma~\ref{lem:Pr}, this requires only two ancilla qubits and $|\mc S|/2$ $\Gate{GZZ}$ gates.

In this section we proposed a method to decompose arbitrary diagonal unitaries into local gates and $\Gate{GZZ}$ gates.
We showed that this leads to an optimization problem which in general is NP-hard.
Using heuristic algorithms might still achieve a significant reduction in the encoding cost.

\begin{example}
We illustrate the compilation scheme explained above on a diagonal unitary on five qubits.
Let $\vec x \in \FF_2^5$ and define a diagonal unitary as in \cref{eq:diag_unitary} with
\begin{equation}
\begin{aligned}
 f(\vec x) = &\frac{1}{2} \big( x_1 \oplus x_2 + x_3 \oplus x_4 + x_4 \oplus x_5 \\
 &+ x_1 \oplus x_2 \oplus x_3 + x_2 \oplus x_4 \oplus x_5 \\
 &+ x_2 \oplus x_3 \oplus x_4 \oplus x_5 \big) \, .
\end{aligned}
\end{equation}
The factor $1/2$ is chosen for convenience and yields a $R_Z (\pi)=\Gate{Z}$ gate in the circuit representation of \cref{eq:one_fourier_term} for each term in $f$.
Conjugating these $\Gate{Z}$ gates with Hadamard gates as in \cref{eq:one_fourier_term} we get one $\Gate{HZH} = \Gate{X}$ gate for each term.

Note that in the above we considered only $s_i = \supp (\vec y_i) > 2$, since all terms with $s_i = 2$ can be implemented with a single $\Gate{GZZ}$ gate.
To keep this example short we here also allow $s_i = 2$ since this gives more non-trivial possibilities to parallelize layers.
We illustrate the diagonal unitary with the supports $s_i$ in the circuit diagram below.
Applying \cref{alg:par} yields
\begin{equation}
\label{eq:example_diag}
\begin{tikzpicture}[baseline=(X.base)]
\node[scale=.65, ampersand replacement=\&] (X)	{
\begin{tikzcd}[column sep=tiny]
\lstick{$x_1$} \& \gate[2]{s_1}\& \qw          \& \qw          \& \qw  \& \gate[3]{s_5}\& \qw       \& \qw       \& \qw \\
\lstick{$x_2$} \& \qw          \& \qw          \& \qw          \& \gate{s_4}  \& \qw       \& \qw       \& \gate[4]{s_7}\& \qw \\
\lstick{$x_3$} \& \qw          \& \gate[2]{s_2}\& \qw          \& \qw  \& \qw       \& \gate[3]{s_6}\& \qw       \& \qw \\
\lstick{$x_4$} \& \qw          \& \qw          \& \gate[2]{s_3}\& \qw  \& \qw       \& \qw       \& \qw       \& \qw \\
\lstick{$x_5$} \& \qw          \& \qw          \& \qw          \& \gate{s_4}  \& \qw       \& \qw       \& \qw       \& \qw
\end{tikzcd}
};
\end{tikzpicture}=
\begin{tikzpicture}[baseline=(X.base)]
\node[scale=.65, ampersand replacement=\&] (X)	{
\begin{tikzcd}[column sep=tiny, slice all,slice titles=$u_\col\ $,slice style=blue,slice label style={inner sep=1pt,anchor=south east,rotate=0}]
\lstick{$x_1$} \& \gate[2]{s_1}\& \gate[3]{s_5} \& \qw          \& \qw          \& \qw \\
\lstick{$x_2$} \& \qw          \& \qw          \& \gate{s_4}   \& \gate[4]{s_7}\& \qw \\
\lstick{$x_3$} \& \gate[3]{s_6}\& \qw          \& \gate[2]{s_2}\& \qw          \& \qw \\
\lstick{$x_4$} \& \qw          \& \gate[2]{s_3}\& \qw          \& \qw          \& \qw \\
\lstick{$x_5$} \& \qw          \& \qw          \& \gate{s_4}   \& \qw          \& \qw
\end{tikzcd}
};
\end{tikzpicture} \, .
\end{equation}
We now group the parallelized layers $u_i$ together in the sets $\mc{L}_2 = \Set*{ u_1, u_2, u_3 }$ and $\mc{L}_1 = \Set*{ u_4 }$.
In this example we do not make use of ancilla qubits.
One can see that the placement of the Hadamard gates should be on $s_1 \cap s_5 \cap s_4 = \{ x_2 \}$ for the supports $s_1, s_5, s_4$ and on $s_6 \cap s_3 \cap s_2 = \{ x_4 \}$ for the supports $s_6, s_3, s_2$ such that all the interior Hadamard gates cancel.
For $s_7$ the Hadamard gate can be placed on $x_2$ or $x_4$.
Since $| s_4 \cap s_7 | = | s_2 \cap s_7 |$ both choices are equally good for the cancellation of $\Gate{CZ}$ gates.
We chose the Hadamard position for $s_7$ to be $x_4$.
The shared support size $S_{u_i,u_j}$ between two parallelized layers for the set $\mc{L}_2$ can be calculated and expressed as the matrix
\begin{equation}
 S = \begin{pmatrix}
 * & 4 & 3 \\
 4 & * & 2 \\
 3 & 2 & * \\
\end{pmatrix}
\end{equation}
With $S_{u_2,u_1} = 4$ and $S_{u_1,u_3} = 3$ we get the optimal order for $\mc{L}_2$ as $u_2, u_1, u_3$, and $2 (S_{u_2,u_1} + S_{u_1,u_3}) = 14$ $\Gate{CZ}$ gates cancel.
Using \cref{eq:one_fourier_term} and $\Gate{HZH} = \Gate{X}$, the resulting circuit is
\begin{equation}
\label{eq:example_diag_GZZ}
\begin{tikzpicture}[baseline=(X.base)]
\node[scale=.75, ampersand replacement=\&] (X)	{
\begin{tikzcd}[column sep=tiny]
\lstick{$x_1$} \& \qw\slice{}    \& \ctrl{1}  \& \qw\slice{$\Gate{GZZ}_1 \qquad $}\& \qw\slice{}    \& \qw       \& \qw\slice{$\Gate{GZZ}_2 \qquad $}\& \qw\slice{}    \& \ctrl{1}  \& \qw\slice{$\Gate{GZZ}_3 \qquad $}\& \qw\slice{}    \& \qw\slice{$\Gate{GZZ}_4 \quad $}\& \qw\slice{}    \& \qw       \& \qw\slice{$\Gate{GZZ}_5 \qquad $}\& \qw\slice{}    \& \qw       \& \qw       \& \qw\slice{$\Gate{GZZ}_6 \qquad\quad $}\& \qw     \& \qw \\
\lstick{$x_2$} \& \gate{\Gate{H}}\& \control{}\& \ctrl{1}                         \& \gate{\Gate{X}}\& \ctrl{1}  \& \qw                              \& \gate{\Gate{X}}\& \control{}\& \ctrl{3}                         \& \gate{\Gate{X}}\& \ctrl{3}                         \& \gate{\Gate{H}}\& \ctrl{2}  \& \qw                              \& \qw            \& \qw       \& \ctrl{2}  \& \qw       \& \qw     \& \qw \\
\lstick{$x_3$} \& \qw            \& \qw       \& \control{}                       \& \qw            \& \control{}\& \ctrl{1}                         \& \qw            \& \qw       \& \qw                              \& \qw            \& \qw                              \& \qw            \& \qw       \& \qw                              \& \qw            \& \qw       \& \qw       \& \ctrl{1}  \& \qw     \& \qw \\
\lstick{$x_4$} \& \gate{\Gate{H}}\& \ctrl{1}  \& \qw                              \& \gate{\Gate{X}}\& \qw       \& \control{}                       \& \gate{\Gate{X}}\& \ctrl{1}  \& \qw                              \& \gate{\Gate{X}}\& \qw                              \& \qw            \& \control{}\& \ctrl{1}                         \& \gate{\Gate{X}}\& \ctrl{1}  \& \control{}\& \control{}\& \gate{\Gate{H}}\& \qw \\
\lstick{$x_5$} \& \qw            \& \control{}\& \qw                              \& \qw            \& \qw       \& \qw                              \& \qw            \& \control{}\& \control{}                       \& \qw            \& \control{}                       \& \qw            \& \qw       \& \control{}                       \& \qw            \& \control{}\& \qw       \& \qw       \& \qw     \& \qw
\end{tikzcd}
};
\end{tikzpicture} \, .
\end{equation}
One needs six $\Gate{GZZ}$ gates to implement the diagonal unitary.
The encoding cost for each $\Gate{GZZ}_i$ in \cref{eq:example_diag_GZZ} with support $s_i$ is $|s_i|(|s_i|-1)/2$, such that the total implement cost is $10+3+6+1+3+6=29$.
For the circuit on the left-hand side of \cref{eq:example_diag} we require two $\Gate{GZZ}$ gates for each support $s_i$.
We have $7$ supports $s_i$, i.e.\ $14$ $\Gate{GZZ}$ gates which leads to the total encoding cost is $2(2+2+2+2+3+3+6)=40$.
Thus, in this example our compilation scheme reduces the encoding cost by $\approx 25\%$.
\end{example}

\section{Conclusion}
In this work, we showed how to synthesize a time-optimal multi-qubit gate, the $\Gate{GZZ}$ gate, on an abstract quantum computing platform.
The only requirements on the platforms are that \ref{item:parallel_1qubit} single-qubit rotations can be executed in parallel, \ref{item:Ising} it offers global Ising-type interactions with all-to-all connectivity, and \ref{item:exclude} that there is a way to exclude certain qubits from the participation in the multi-qubit gate.
We showed that and how to realize these requirements in a microwave controlled ion trap using \acf{MAGIC}.

Arbitrary couplings between the qubits can be generated via $\Gate{X}$ gate layers interleaved by the Ising-type evolution. 
The required $\Gate{X}$ gate layers as well as the durations of the evolution times are determined by solving an \acs{LP}. 
In numerical experiments we showed that the gate time of the resulting $\Gate{GZZ}$ gates scales approximately linear with the number of participating qubits under reasonable assumptions on an implementing physical platform.
Based on these time-optimal $\Gate{GZZ}$ gates, we presented an improved compiling strategy for Clifford circuits, and applied this strategy to compile the \ac{QFT}. 
Moreover, we applied the $\Gate{GZZ}$ gates to the simulation of molecular dynamics, and presented a compiling strategy for general diagonal unitaries.
This can be thought of as a step towards compilation strategies for arbitrary unitaries.

In the future, it will be interesting to investigate how the $\Gate{GZZ}$ gates perform on a real-world ion trap, and how robust they can be made against errors similar to the error mitigation scheme for the \ac{DAQC} setting \cite{garcia_molina_noise_2021}.
Moreover, we hope that our time-optimal gate synthesis method will be applied to small ion trap registers which are embedded in a \ac{QPU} module.

\section*{Acknowledgements}
We are grateful to the MIQRO collaboration, in particular to Robert Jördens for helpful discussions and valuable feedback on our manuscript and Christof Wunderlich for discussions on the platform and the \ac{QFT}. 
Moreover, we would like to thank Christian Gogolin for making us aware of the connection to molecular dynamics and Lennart Bittel for fruitful discussions on convex optimization. 

This work has been funded by the German Federal Ministry of Education and Research (BMBF) within the funding program ``quantum technologies -- from basic research to market'' via the joint project MIQRO (grant numbers 13N15522 and 13N15521).

\section*{Appendices} 
\appendix
\section{Ion trap quantum computing with microwave control}
\label{apdx:MAGIC}
In this section we explain in more detail how qubits are encoded into cold ions stored in a trap, how \ac{MAGIC} can be used to tailor interactions between these qubits, and how microwaves perform single-qubit gates.
A full in-depth treatment can be found in Refs.~\cite{wunderlich_conditional_2001, johanning_quantum_2009}.

In short, pairs of hyperfine states of ions are interpreted as the computational basis states of qubits.
These ions are cooled down to form a ``Coulomb crystal'', a stable configuration confined by the trap potential.
An inhomogeneous magnetic field superimposed with the trap makes the crystal equilibrium depend on the internal state (the qubits).
Changes in the internal state hence lead to excitations (phonons), which can be interpreted as interactions between the qubits.
The position-dependent Zeeman effect makes individual qubit transitions addressable.
Microwaves drive Rabi oscillations on these transitions to perform single-qubit gates.

\subsection{Ytterbium ions and qubits}
We consider ions with nuclear spin $I = \frac{1}{2}$ and total electron angular momentum $J = \frac{1}{2}$, for example single-ionized Ytterbium-171 ($^{171}\mathrm{Yb}^+$).
These values imply that the ``ground state'' (the lowest main quantum number) is spanned by four ``hyperfine'' sublevels.
Since nuclear spin and electron angular momentum couple to each other through an interaction term $\propto \vec{I} \cdot \vec{J}$, the full Hamiltonian is not diagonal in the naive product basis $\ket{m_I, m_J}$.
One instead introduces the total angular momentum $\vec{F} \coloneqq \vec{I} + \vec{J}$.
The corresponding quantum numbers $F$ and $m_F = -F, \dots, F$ can be used to label the hyperfine eigenstates, which group into a singlet $\ket{F=0,m_F=0}$ and a triplet $\ket{F=1, m_F=-1,0,+1}$.
The singlet forms the overall ground state and is separated by an energy gap from the triplet states, which all have the same energy.
The energy difference is called the ``hyperfine structure constant'' of the ion.
This quantity is usually determined experimentally and typically lies in the microwave regime.
For $^{171}\mathrm{Yb}^+$, it corresponds to a frequency of $f_0 = 12.642 812 118 466(2) \text{GHz}$ \cite{fisk1997accurate}.

Placing the ion in an external magnetic field $\vec{B}$ yields new dipole terms $\propto \vec{B} \cdot \vec{I}$ and $\propto \vec{B} \cdot \vec{J}$ in the Hamiltonian.
Since these new terms do not preserve $F$, the $\ket{F, m_F}$ basis is no longer an eigenbasis.
However, they still preserve $m_F$, so the new energy eigenstates are superpositions of states with equal $m_F$.
The ``outer'' triplet states $\ket{F=1,m_F=\pm1}$ remain basis states due to their unique $m_F$, but $\ketb{F=0,m_F=0}$ and $\ket{F=1,m_F=0}$ form two orthogonal linear combinations with amplitudes depending on the field strength $B \coloneqq \abs{\vec{B}}$.
Also the eigenvalues of the Hamiltonian become $B$-dependent, and the previous triplet degeneracy is lifted.
For the present case of four states, the problem can be solved analytically leading to the Breit-Rabi formula \cite{breit_rabi_1931}.
We are mainly interested in the weak field limit $B \rightarrow 0$, commonly known as the Zeeman effect.
The energy shift of the states $\ket{F=1,m_F=\pm1}$ is exactly proportional to $B$ (``linear Zeeman effect''), while the other two eigenvalues are field-independent to first order (``magnetic insensitive'').
However, they do have a higher-order dependence starting from $\propto B^2$, which is known as the ``quadratic Zeeman effect'', see also \cref{fig:Yb_levels}.
Not only are these eigenvalues close to the zero-field values, but also the linear combinations have a dominant component of either $\ket{F=0,m_F=0}$ or $\ket{F=1,m_F=0}$.
So although this is technically wrong, one casually keeps the $\ket{F, m_F}$ nomenclature for the basis states even in the presence of a small magnetic field.

We now wish to ``encode'' a qubit into this four-dimensional state space, i.e.\ to pick two of the energy eigenstates as computational basis states $\ket{0}$ and $\ket{1}$.
A priori, there are six possible choices.
The Zeeman splitting between the $F=1$ states is typically on the order of MHz, so orders of magnitude smaller than the hyperfine splitting.
Also, the transitions between $\ket{F=1,m_F=0}$ and the other two triplet states are degenerate to first order in $B$.
This is why one usually chooses $\ket{0} \coloneqq \ket{F=0, m_F=0}$ and only considers the three remaining possibilities which we refer to as
\begin{center}\begin{minipage}{.9\linewidth}\begin{tabular}{>{\bfseries\boldmath}rl}
	$\sigma^+$ qubit & with $\ket{1} \coloneqq \ket{F=1,m_F=+1}$, \\
	\phantom{$\sigma^\pm$} \llap{$\pi\,$} qubit & with $\ket{1} \coloneqq \ket{F=1,m_F=0}$, \\
	$\sigma^-$ qubit & with $\ket{1} \coloneqq \ket{F=1,m_F=-1}$. \\
\end{tabular}\end{minipage}\end{center}
One may also label the qubits by the $m_F$ of their $\ket{1}$ state.
The $B$-dependent transitions frequencies between $\ket{0}$ and $\ket{1}$ can then be written
\begin{equation}\label{eq:omega_gen}
\omega(m_F,B) = 2\pi f_0 + m_F \frac{\mu B}{\hbar} + \LandauO(B^2),
\end{equation}
where $\mu$ has units of a magnetic moment ($\mu \approx \mu_\mathrm{B}$, the Bohr magneton, for $^{171}\mathrm{Yb}^+$).
This reflects both the quadratic ($m_F=0$) and the linear ($m_F=\pm1$) Zeeman effect.

\subsection{Trap Hamiltonian and \texorpdfstring{\ac{MAGIC}}{MAGIC}}
So far, we only considered a single ion.
In an ion trap quantum computer, $N$ such ions are confined in an (effective) potential generated by DC and RF electrodes.
We assume that the potential geometry and $N$ have been selected such that the ions form a ``linear Coulomb crystal'', i.e.\ a one-dimensional chain stabilized by the external trap potential and the mutual Coulomb repulsion of the ions, see \cref{fig:trap_sketch}.
If the crystal has been cooled sufficiently, we can disregard the radial dynamics and focus on the direction of the ion chain which we choose as the $z$ axis.
The system is described by a Hamiltonian with three contributions: the kinetic energy $T(\vec{p}) \propto \vec{p}^2$ of the ions with momenta $\vec{p} = (p_1, \dots, p_N)^T$, and the external and Coulomb potentials
\begin{equation}
	V_\mathrm{tot}(\vec{z}) \coloneqq \sum_{i=1}^N V_\mathrm{ext}(z_i) + K \sum_{i<j}^N \frac{1}{|z_i - z_j|},
\end{equation}
with ions positions $\vec{z} \coloneqq (z_1, \dots, z_N)^T$ and Coulomb constant $K$.
Again, after sufficient cooling it is reasonable to assume that the system is close to an equilibrium configuration $\vec{\bar{z}}$.
The potential can then be expanded in terms of the elongation $\vec{q} \coloneqq \vec{z} - \vec{\bar{z}}$.
The first order term vanishes due to equilibrium, such that
\begin{equation}
	V_\mathrm{tot}(\vec{\bar{z}} + \vec{q}) = V_\mathrm{tot}(\vec{\bar{z}}) + \frac{1}{2} \vec{q}^T \operatorname{H}_V(\vec{\bar{z}}) \vec{q} + \LandauO(|\vec{q}|^3),
\end{equation}
where $\operatorname{H}_V$ is the Hessian matrix of $V_\mathrm{tot}$.
We drop the constant term and neglect the higher-order terms in the following (harmonic approximation).
This results in a quadratic many-body Hamiltonian with a well-known spectrum of phonon excitations.

We add two more ingredients to the system:
An inhomogeneous magnetic field with constant gradient $B_1$ in $z$ direction,
\begin{equation}\label{eq:magnetic_field}
	\vec{B}(z) = B(z) \hat{e}_z, \qquad B(z) = B_0 + B_1 z,
\end{equation}
and the internal dynamics of ionic qubits with chosen bases specified by $\vec{m}_F \coloneqq (m_{F,1}, \dots, m_{F,N})^T$.
Plugging \cref{eq:magnetic_field} into \cref{eq:omega_gen} yields a qubit-specific transition frequency that depends on both $z_i$ and $m_{F,i}$.
In the Hamiltonian, the splitting can be expressed using the Pauli $Z$ operator of the qubit.
Expanded to first order in $\vec{q}$, the corresponding term reads
\begin{equation}
\begin{aligned}
	-\frac{\hbar}{2} \sum_{i=1}^N \omega(m_{F,i}, B(z_i)) Z_i
	&= -\frac{\hbar}{2} \sum_{i=1}^N \left( \omega^{(0)}_i + m_{F,i} \frac{\mu B_1}{\hbar} q_i \right) Z_i \\
	&= -\frac{\hbar}{2} \sum_{i=1}^N \omega^{(0)}_i Z_i - \frac{\mu B_1}{2} \vec{q}^T (\vec{m}_F \circ \vec{Z})\, ,
\end{aligned}
\end{equation}
with $\vec{Z} \coloneqq (Z_1, \dots, Z_N)^T$ the local $Z$ operators and frequencies $\omega^{(0)}_i$ that depend on $m_{F,i}$, $\bar{z}_i$ and various constants.
The first term of the result is a typical Zeeman term involving only the qubits, while the second term couples external ($\vec{q}$) and internal ($\vec{Z}$) degrees of freedom.
This last term is the central ingredient of \acf{MAGIC}.

Collecting all contributions gives the Hamiltonian
\begin{equation}
	H = T(\vec{p}) + \frac{1}{2} \vec{q}^T \operatorname{H}_V(\vec{\bar{z}}) \vec{q} - \frac{\mu B_1}{2} \vec{q}^T (\vec{m}_F \circ \vec{Z}) - \frac{\hbar}{2} \sum_{i=1}^N \omega^{(0)}_i Z_i \, .
\end{equation}
It is still second-order in the external degrees of freedom, but no longer purely quadratic.
We fix this by completing the square in $\vec{q}$,
\begin{equation}\label{eq:H_complete_square}
\begin{aligned}
		H ={} &T(\vec p) + \frac{1}{2} \left( \vec q - \frac{\mu B_1}{2} \operatorname{H}_V^{-1}(\bar{\vec{z}}) (\vec m_F \circ \vec Z) \right)^T \operatorname{H}_V(\bar{\vec{z}}) \left( \vec q - \frac{\mu B_1}{2} \operatorname{H}_V^{-1}(\bar{\vec{z}}) (\vec m_F \circ \vec Z) \right) \\
		&- \frac{\hbar}{2} \sum_{i=1}^N \omega^{(0)}_i Z_i - \frac{1}{2} \left(\frac{\mu B_1}{2}\right)^2 (\vec m_F \circ \vec Z)^T \operatorname{H}_V^{-1}(\bar{\vec{z}}) (\vec m_F \circ \vec Z) \,.
\end{aligned}
\end{equation}
The existence of the inverse Hessian $\operatorname{H}_V^{-1}(\bar{\vec{z}})$ is guaranteed by the assumption of a stable equilibrium, which implies that $\operatorname{H}_V$ is positive definite.
Note also that $\operatorname{H}_V$ and its inverse are symmetric matrices.
In Eq.~\eqref{eq:H_complete_square}, it is evident that the first line is quadratic in the external degrees of freedom ``up to a state dependent translation'', while the second line only involves the qubits.
Formally, this means that we can conjugate the Hamiltonian with the unitary
\begin{equation}\label{eq:separating_transform}
	U \coloneqq \exp\left( \i\, \frac{\mu B_1}{2\hbar} \vec p^T \operatorname{H}_V^{-1}(\bar{\vec{z}}) (\vec m_F \circ \vec Z) \right)
\end{equation}
to separate external and internal dynamics.
Noting that $U$ commutes with all terms of $H$ except for the one involving $\vec{q}$, this results in
\begin{equation}
	\tilde{H} \coloneqq U H U^\dagger = H_\mathrm{phonons}(\vec{p}, \vec{q}) + H_\mathrm{qubits}(\vec{Z}),
\end{equation}
with $H_\mathrm{phonons}(\vec{p}, \vec{q})$ being a standard many-body phonon Hamiltonian that is easily solvable in normal coordinates.
The internal part is simply the second line of \cref{eq:H_complete_square}, which we rewrite as
\begin{equation}\label{eq:Hqubits}
	H_\mathrm{qubits}(\vec{Z}) = -\frac{\hbar}{2} \sum_{i=1}^N \omega^{(0)}_i Z_i - \frac{\hbar}{2} \vec{Z}^T J(\vec{m}_F) \vec{Z}.
\end{equation}
Here, we used the Hadamard product identity (for vectors $\vec{a}, \vec{b}$ and a matrix $C$)
\begin{equation}\label{eq:Had_identity}
	(\vec{a} \circ \vec{b})^T C (\vec{a} \circ \vec{b}) = \vec{a}^T (C \circ \vec{b} \vec{b}^T) \vec{a}
\end{equation}
and defined
\begin{equation}\label{eq:def_J}\begin{aligned}
	J(\vec{m}_F) {}&\coloneqq J \circ \vec{m}_F \vec{m}_F^T, \\
	J {}&\coloneqq \frac{1}{\hbar} \left(\frac{\mu B_1}{2}\right)^2 \operatorname{H}_V^{-1}(\bar{\vec{z}}).
\end{aligned}\end{equation}
The two terms of \cref{eq:Hqubits} commute, so we can treat them separately, or remove the first one entirely by another unitary transformation.
In any case, the local operators $Z_i$ do not entangle the qubits, so we focus our analysis on the second term.
The diagonal entries of $J(\vec{m}_F)$ simply generate a constant term $-\frac{\hbar}{2} \Tr[J(\vec{m}_F)]\, \1$, which amounts to a global and hence unobservable phase.
Without loss of generality, we can thus set the diagonal terms of $J$ to zero.
Finally, for the discussion on the logical layer, we choose units in which $\hbar=1$.
After these simplifying steps, \cref{eq:Hqubits} becomes the Hamiltonian~\eqref{eq:Hising} on which our analysis is based.

We provide Python code for the computation of coupling matrices based on the trap potential and physical parameters on GitHub \cite{GitHub}.

\subsection{Microwaves and single-qubit gates}
\label{apdx:single_qubit}
A comprehensive discussion of the physics of microwave-controlled single-qubit gates, in particular within the \ac{MAGIC} scheme, goes even beyond the scope of this appendix.
It involves common concepts from quantum optics like the rotating wave approximation and Lamb-Dicke parameter, extended to the situation with a permanent inhomogeneous magnetic field.
This extension is crucial, because it is the magnetic gradient that enables reasonable gate times even for microwave radiation which are otherwise much slower than optical gates.
We recommend Refs.~\cite{wunderlich_conditional_2001, johanning_quantum_2009} for more details.

Important for our analysis is that resonant microwave excitation of a hyperfine transition induces a Rabi Hamiltonian on the corresponding qubit,
\begin{equation}\label{eq:H_Rabi}
	H_\mathrm{Rabi}(\Omega, \phi) = \frac{\hbar\Omega}{2} \big( \cos(\phi) X + \sin(\phi) Y \big).
\end{equation}
Here, $X$ and $Y$ are Pauli operators, $\Omega$ is the Rabi frequency which is proportional to the microwave amplitude, and $\phi$ is a parameter controlled by the relative phase between the microwave and essentially the qubit's Larmor precession.
Resonant excitation means that the microwave frequency matches the transition frequency $\omega(m_F, B)$ of the targeted qubit.
While the microwave is switched on, one (or several, if the experimental setup allows for multitone signals) Hamiltonian(s) of the form \eqref{eq:H_Rabi} act simultaneously with~\eqref{eq:Hqubits}, potentially with different Rabi frequencies $\Omega_i$.
As these Hamiltonians do not commute, the time evolution of the system usually eludes exact analytical treatment.
However, in applications one aims to make the Rabi frequencies orders of magnitude larger than the entries of $J$, which govern the entangling dynamics.
This allows for the approximation that the gates induced by the Hamiltonian~\eqref{eq:H_Rabi} are instantaneous, or equivalently that the Hamiltonian~\eqref{eq:Hqubits} is negligible while the microwave is on.
The validity of this approximation and the errors introduced by it are discussed in \cref{sec:limits} and Ref.~\cite{parra_rodriguez_digital_analog_2020}.
Time evolution under the Rabi Hamiltonian~\eqref{eq:H_Rabi} for a time $\theta/\Omega$ implements a family of Bloch rotation gates,
\begin{equation}
	\Gate{R}(\theta, \phi) \coloneqq \exp\!\left( -\frac{\i\theta}{\hbar\Omega} H_\mathrm{Rabi}(\Omega, \phi) \right).
\end{equation}
On the Bloch sphere, they rotate by an angle $\theta$ around an axis in the $xy$-plane specified by azimuth $\phi$.
The $\Gate{X}$ and $\Gate{H}$ gates necessary for our analysis are generated by
\begin{equation}
	\Gate{X} = \i \Gate{R}(\pi, 0), \qquad
	\Gate{H} = \i \Gate{R}(\pi, 0) \Gate{R}\left( \frac{\pi}{2}, \frac{\pi}{2} \right).
\end{equation}

\subsection{Physical and virtual recoding}
As mentioned before, the $\pi$~qubit with $m_F=0$ is called ``magnetic insensitive'' because it depends on $B$ only in second order.
This can be used to exclude (up to quadratic corrections) qubits from the interaction and perform ``subset operations''.
Corresponding rows and columns in the coupling matrix $J(\vec m_F)$ in \cref{eq:def_J} manifestly vanish.
However, insensitivity also means that all $\pi$~qubits have roughly the same transition frequency, see \cref{eq:omega_gen}.
Hence, they are not separately addressable, and we cannot perform independent single-qubit operations on them.
The advantages and disadvantages of $\pi$~qubits naturally pose the question, whether one can change the qubit basis during computation.

This can be answered in the affirmative, but one intermediately has to consider (at least) three-dimensional subspaces instead of qubits.
Let $\Gate{X}_\pm$ denote the $\Gate{X}$~gate on the $\sigma^\pm$~qubit of the same ion.
These gates ``swap'' the two states on which they act, and one can easily confirm that the local sequences $\Gate{X}_+ \Gate{X}_- \Gate{X}_+$ and $\Gate{X}_- \Gate{X}_+ \Gate{X}_-$ both effectively swap the states $\ket{F=1,m_F=+1}$ and $\ket{F=1,m_F=-1}$.
Thus, they ``recode'' between the $\sigma^+$ and $\sigma^-$ qubit.

Coding into and out of the $\pi$ qubit is slightly more complicated due to the mentioned degenerate transitions.
Yet we can use the global swap-like operation $\Gate{X}_0^{\otimes N}$, which applies an $\Gate{X}$~gate to all $\pi$ qubits at once (assuming that shifts due to the quadratic Zeeman effect do not become prohibitive).
Sequences of the form
\begin{equation}
	\Gate{X}_0^{\otimes N} \big( \ldots \otimes \Gate{I} \otimes \ldots \otimes \Gate{X}_+ \otimes \ldots \otimes \Gate{X}_- \otimes \ldots \big) \Gate{X}_0^{\otimes N},
\end{equation}
i.e.\ with an arbitrary combination of local $\Gate{X}_\pm$ and idle gates in between two global $\Gate{X}_0^{\otimes N}$~gates, achieve a recoding between $\pi$ and $\sigma^\pm$ qubits (for an $\Gate{X}_\pm$~gate) or just swap back and forth, effectively doing nothing (for an idle gate).

The described sequences of three rounds of microwave gates are able to physically manipulate the way quantum information is encoded in the ions.
However, our main concern is not how information is stored, but how it can be processed.
To that end one usually does not need to recode between the physical $\sigma^+$ and $\sigma^-$~qubits, but can instead use single layers of $\Gate{X}$~gates (on the physical qubit in use) to emulate the opposite encoding.
This can significantly reduce the number of single-qubit gates in a circuit, and we make heavy use of this possibility.
Using the notation of \cref{eq:virtual_encoding_basis}, a detailed derivation of this technique reads
\begin{equation}\begin{aligned}
	\Gate{X}^{\vec s} \exp\left( \frac{\i t}{2} \vec{Z}^T J(\vec{m}_F) \vec{Z} \right) \Gate{X}^{\vec s}
	&= \exp\left( \frac{\i t}{2} X^{\vec s} \vec{Z}^T X^{\vec s} J(\vec{m}_F) X^{\vec s} \vec{Z} X^{\vec s} \right) \\
	&= \exp\left( \frac{\i t}{2} \left( (-1)^{\vec s} \circ \vec{Z} \right)^T J(\vec{m}_F) \left( (-1)^{\vec s} \circ \vec{Z} \right) \right) \\
	&= \exp\left( \frac{\i t}{2} \vec{Z}^T \left( J(\vec{m}_F) \circ \left( (-1)^{\vec s} \left((-1)^{\vec s}\right)^T \right) \right) \vec{Z} \right) \\
	&= \exp\left( \frac{\i t}{2} \vec{Z}^T J\big( (-1)^{\vec s} \circ \vec{m}_F \big) \vec{Z} \right) \, ,
\end{aligned}\end{equation}
where we used that Pauli strings $X^{\vec s}$ are self-inverse and applied the Hadamard product identity \cref{eq:Had_identity}.
Negative signs in $(-1)^{\vec s}$ exchange, from the perspective of the $J$~coupling, the $\sigma$~qubits, but leave $\pi$~qubits invariant.
If this technique is used consecutively, subsequent $\Gate{X}$ gate layers can be combined into a single one,
\begin{equation}\label{eq:encoding_transition}
	\Gate{X}^{\vec s} \Gate{X}^{\vec s'} = \Gate{X}^{\vec s \oplus \vec s'},
\end{equation}
saving even more single-qubit gates, as detailed in the main text.

\section{Frame theory}
\label{apdx:frame}
In this section we connect our approach to frame theory, which is concerned with spanning sets for Hilbert spaces, i.e.\ with generalizations of the notion of a basis, see e.g.\ Ref.~\cite{waldron_introduction_2018}.

We consider the space of symmetric matrices with vanishing diagonal that is defined by
\begin{equation}\label{def:symtraceless_1}
	\HTL(\RR^n) \coloneqq \Set*{M \in \mathrm{Sym}(\RR^n)\given M_{ii}= 0 \ \forall i\in [n]} \, .
\end{equation}
Moreover, we denote the set of outer products of all possible encodings $\vec m= (-1)^{\vec b}$, where
$\vec b \in \FF_2^n$, by
\begin{equation}\label{def:outprod_1}
	\mc{V} \coloneqq \Set*{ \vec m \vec m^T \given \vec m \in \{-1,+1\}^n , m_n = +1} \, .
\end{equation}
We further define what it means for a frame to be harmonic:
\begin{definition}
\cite[Definition 11.1]{waldron_introduction_2018}
 Let $G$ be a finite abelian group, and let $\hat G$ be the set of irreducible characters of $G$.
 A tight frame $\{ \phi_i \}_{i \in J}$ for $\RR^k$ is called a \emph{harmonic} if it is unitarily equivalent to
 \begin{equation}
  (\xi|_J)_{\xi \in \hat G} \subset \RR^{|J|} \equiv \RR^k,
 \end{equation}
where $J \subset G$ and $|J| = k$.
\end{definition}

The following theorem extends \cref{thrm:frame1} and shows that $\mc{V}$ is a frame with certain properties.
\begin{theorem}
\label{thrm:frame}
$\mc{V}$
is a balanced equal-norm harmonic tight frame for $\HTL(\RR^n)$.
\end{theorem}
\begin{proof}
Let $k= n(n-1)/2 = \dim (\HTL(\RR^n))$ and $l = 2^{n-1} = |\mc{V}|$.
We denote the synthesis operator of $\mc{V}$ by $V: \RR_{\geq 0}^l \rightarrow \HTL(\RR^n) : \vec \lambda \mapsto \sum_{\vec m} \lambda_{\vec m} \vec m \vec m^T$ which can be represented by a matrix $V  \in \{-1,+1 \}^{l \times k}$.
We further denote the column $\vec v(\vec m)$ of $V$, containing the lower triangular elements of $\vec m \vec m^T$.
First we want to show that each row of $V$ corresponds to a row of a Hadamard matrix.
Since the elements of any Hadamard matrix $H_{\vec y,\vec x}$ (with normalization factor) are given by Walsh functions  $H_{\vec y,\vec x} = 2^{-n/2} W_{\vec y} (\vec x) = 2^{-n/2} (-1)^{\vec y \cdot \vec x}$ for $\vec y,\vec x \in \FF_2^n$, this amounts to showing that the rows of $V$ correspond to Walsh functions $W_{\vec y} (\vec x) \coloneqq (-1)^{\vec y \cdot \vec x} $, where $\vec y \cdot \vec x = \sum_i^n y_i x_i$.

Writing $\vec m=(-1)^{\vec b}$, the components of the matrix $\vec m \vec m^T$ are $(\vec m \vec m^T)_{gh} = (-1)^{b_g \oplus b_h}$.
Then the entries of the columns $\vec v (\vec m)$ are $\vec v (\vec m)_{gh} = (-1)^{b_g \oplus b_h}= (-1)^{\vec b \cdot \left(\vec e^g \oplus \vec e^h \right) }$, where the tuple $g$, $h$ is the index of the row of $V$ and  $\vec e^i \in \FF_2^{n}$ denotes the $i$-th standard unit vector.

This shows that $\vec v (\vec m)_{gh} = W_{\vec b} (\vec e^g \oplus \vec e^h)$.
So the $\vec e^g \oplus \vec e^h$ encodes the row indices and the $\vec b$ encodes the column indices of the Hadamard matrix.
Since we consider all $\vec b \in \FF_2^n$ we have all columns of the Hadamard matrix.
This shows that each row of $V$ corresponds to a row of a Hadamard matrix.

Since the rows/columns of a Hadamard matrix form an orthonormal basis one can use the row construction of tight frames from orthogonal projections \cite[Theorem 2.3]{waldron_introduction_2018} to get an equal-norm tight frame.
Furthermore, the Hadamard matrix is the character table of the cyclic group $C_2^{n} = C_2 \times \dots \times C_2$ which is abelian.
This can be seen by its recursive definition $H_n = H_1 \otimes H_{n-1}$, where $H_1$ is the character table of $C_2$.
Thus, Ref.~\cite[Theorem 11.1]{waldron_introduction_2018} implies that this equal-norm tight frame is also a harmonic.

It is balanced if $\sum_{\vec b}  (-1)^{b_g \oplus b_h} = 0$ for all $g, h = 1, \dots , n$ with $g \neq h$.
Since we only consider the lower/upper triangular matrix without the diagonal $g=h$, our frame is indeed balanced.
\end{proof}

\begin{corollary}
 The normalized frame $\{ \vec m \vec m^T /\sqrt{k} \}_{\vec m \in \{-1,+1\}^n}$ for $\HTL(\RR^n)$ forms a spherical $2$-design.
\end{corollary}
\begin{proof}
The elements of the normalized frame $\{ \vec m \vec m^T /\sqrt{k} \}_{\vec m \in \{-1,+1\}^n}$ constitute a set of unit vectors in $\RR^k$ which form a normalized balanced tight frame by \cref{thrm:frame}. 
Then, by Refs.~\cite[Proposition 1.2]{holmes_optimal_2004} and \cite[Proposition 6.1]{waldron_introduction_2018}, $\{ \vec m \vec m^T /\sqrt{k} \}_{\vec m \in \{-1,+1\}^n}$ is also a spherical $2$-design.
\end{proof}

\begin{theorem}
\label{thrm:gram}
The entries of the Gram matrix $P=V^T V$ are given by
 \begin{equation}
  P_{\vec m , \vec m^\prime}
  = \langle \vec v (\vec m) , \vec v (\vec m^\prime) \rangle
  = \frac{n}{2} (n-1)-2 \Delta_{\vec b, \vec b^\prime} (n- \Delta_{\vec b, \vec b^\prime} ) \, ,
 \end{equation}
where $\vec v (\vec m)$ are the columns of $V$, containing the lower triangular elements of $\vec m \vec m^T$, and $\Delta_{\vec b, \vec b^\prime}\coloneqq |\vec b \oplus \vec b^\prime|$ is the Hamming distance between $\vec b$ and $\vec b^\prime$.
\end{theorem}
\begin{proof}
 As in the proof above, write the $(g,h)$-th entry of the column $\vec v (\vec m)$ as
 \begin{equation}
  \vec v (\vec m)_{gh} = (-1)^{b_g \oplus b_h} \, .
 \end{equation}
This yields
\begin{equation}
\begin{aligned}
 \langle \vec v (\vec m) , \vec v (\vec m^\prime) \rangle &= \sum_{g < h} (-1)^{b_g \oplus b_h} (-1)^{ b_g^\prime \oplus b_h^\prime} \\
 &= \sum_{g < h} (-1)^{b_g \oplus b_g^\prime \oplus  b_h \oplus b_h^\prime} \\
 &= \sum_{g < h} (-1)^{c_g \oplus  c_h} \, ,
\end{aligned}
\end{equation}
where we set $\vec c \coloneqq \vec b \oplus \vec b^\prime$.
The summand after the last equation is the outer product $\vec{\tilde m} \vec{\tilde m}^T$ for $\vec{\tilde m} \in \{-1, +1\}^n$ and $\vec{\tilde m} = (-1)^{\vec c}$.
If the Hamming distance $\Delta_{\vec b, \vec b^\prime}$ vanishes, 
then $\vec{\tilde m} = (+1, \dots , +1)$ and $\langle \vec v (\vec m) , \vec v (\vec m^\prime) \rangle = n(n-1)/2$ which is the maximal entry of the Gram matrix located at the diagonal.
If $\Delta_{\vec b, \vec b^\prime} \neq 0$, $\vec{\tilde m}$ contains $\Delta_{\vec b, \vec b^\prime}$-many summands $-1$, such that the lower triangular part of $\vec{\tilde m} \vec{\tilde m}^T$ contains $\Delta_{\vec b, \vec b^\prime}(n-\Delta_{\vec b, \vec b^\prime})$-many summands $-1$.
Therefore,
\begin{equation}
 \langle \vec v (\vec m) , \vec v (\vec m^\prime) \rangle = \frac{n}{2} (n-1)-2 \Delta_{\vec b, \vec b^\prime} (n- \Delta_{\vec b, \vec b^\prime} ) \, .
\end{equation}
\end{proof}

\section{Convex optimization arguments for the \texorpdfstring{\acf{LP}}{LP}}
\label{apdx:convOpt}
We investigate the sparsity and geometric properties of the optimal solutions of the \ac{LP}~\eqref{eq:LP1}.
First, we prove \cref{prop:optSparseSol1}, which is a standard result in linear programming:

\begin{proposition}[Sparse optimal solutions]
\label{lem:optSparseSol}
 There is an optimal solution to the \ac{LP}~\eqref{eq:LP1} with sparsity $\leq n(n-1)/2$ for every $M \in \HTL(\RR^n)$.
 The simplex algorithm is guaranteed to return such an optimal solution.
\end{proposition}
\begin{proof}
 The \ac{LP}~\eqref{eq:LP1} has $m=n(n-1)/2$ equality constraints.
 The feasible polytope is the convex polytope obtained by intersecting the $(n-m)$-dimensional subspace defined by those with the positive cone $x \geq 0$.
 To define a vertex of the feasible polytope, a point $x$ has to saturate at least $n-m$ many inequalities $x_i\geq 0$, hence, it has at most $n-(n-m) = m$ many non-zero entries.
 The optimal solutions are obtained by minimizing over the feasible polytope and thus correspond, in general, to a face of this polytope.
 Any vertex of this face defines an optimal solution which has to be at least $m=n(n-1)/2$-sparse.
 The last statement follows since the simplex algorithm only returns vertices of the feasible polytope.
\end{proof}

Alternative algorithms like interior point methods should be avoided. 
Since those are not guaranteed to return vertices of the feasible polytope if there is an entire face of optimal solutions, their solution will generally be a dense vector.

We can say a bit more about the solutions of the \ac{LP}~\eqref{eq:LP1} by geometrical observations.
These observations hold more generally for \ac{LP}s of the following form.

\begin{definition}
\label{def:LP_conical}
 Suppose that vectors $v_1,\dots,v_N \in \R^d$ are the vertices of a full-dimensional polytope $P$ and the origin is contained in the convex hull of $P$, i.e.\ $0\in\mathrm{conv}(P)$.
 Given a vector $u\in\R^d$, we define the following linear program:
 \begin{mini}
	{}{  \langle \mathbf 1,x\rangle = \sum_{i=1}^N x_i}{}{}
	\label{eq:LP_conical}
	\addConstraint{u} {=Vx}{}{}
	\addConstraint{x}{ \geq 0 }{}{}
 \end{mini}
 where $V = \sum_{i=1}^N v_i e_i^\top \in\R^{d\times N}$.
\end{definition}

\begin{lemma}
 In the setting of \cref{def:LP_conical}, the following holds:
 \begin{enumerate}[label= (\roman*)]\itemsep=0pt
  \item The point $u\in\R^d$ lies in a cone generated by a face $F$ of $P$. In particular, there is a feasible solution $x$ of the \ac{LP}~\eqref{eq:LP_conical} with support only on the vertices of $F$.
  \item Every feasible solution with support on the vertices of $F$ has the same objective value, namely $\langle f, u\rangle$ where $f$ is a normal vector of $F$, i.e.~$F \subset \{ y\in\R^d \, | \, \langle f, y\rangle = 1 \}$. This value does not depend on the choice of normal vector.
  \item There is a feasible solution with sparsity $\leq \dim F + 1 \leq d$ and objective value $\langle f, u\rangle$.
  \item The optimal solutions of the \ac{LP}~\eqref{eq:LP_conical} correspond exactly to the possible conical combinations of $u$ in $\cone(F^*) \coloneqq \Set*{ \sum_i x_i v_i \given v_i \in F^*, x_i \in \RR_{\geq 0} }$ where $F^*$ is the lowest-dimensional face of $P$ such that $u\in\cone(F^*)$.
  In particular, the minimum is given by $\langle f, u\rangle$ where $f$ is some normal vector of $F^*$.
 \end{enumerate}
\end{lemma}

\begin{proof}
 Statement (i) follows from the observation that given a point $u\in\R^d \setminus 0$, there is a unique $\vec \lambda > 0$ such that $\vec \lambda u$ lies on a face $F$ of $P$.
 Thus, $\cone(F)$ contains $u$.
 More precisely, the cones spanned by the facets of $P$ form a partition of $\R^d$ where the intersections between any two cones is either $\{0\}$ or a cone spanned by a lower-dimensional face of $P$.

 For statement (ii), let w.l.o.g.~$v_1,\dots,v_s$ be the vertices that lie in $F$.
 By assumption, $u\in\cone(v_1,\dots,v_s)$ and hence $u=\sum_{i=1}^s x_i v_i$.
 We find
 \begin{equation}
  \langle f, u \rangle = \sum_{i=1}^s x_i \langle f, v_i \rangle = \sum_{i=1}^s x_i = \langle \mathbf 1, x \rangle.
 \end{equation}
 Note that the objective value $\langle f,u \rangle$ is necessarily the same, for any choice of normal vector.

 Statement (iii) follows by triangulating the face $F$ with simplices.
 These simplices have $\dim F+1$ vertices and $u$ has to lie in the cone spanned by one of those.
 Thus, there is a feasible solution with sparsity at most $\dim F+1 \leq d$ and objective value $\langle f, u \rangle$.

 Finally, let $x^*$ be an optimal solution.
 Suppose that there is a $i\in\supp(x^*)$ such that the vertex $v_i$ is not in $F$.
 Then, we can find a supporting hyperplane of $F$ with normal vector $f$ such that $\langle f, v_i\rangle < 1$.
 We thus have
 \begin{equation}
  \langle f, u \rangle = \sum_{j\neq i} x^*_j \langle f, v_j \rangle + x^*_i \langle f, v_i \rangle < \sum_{j} x_j^* = \langle \mathbf 1, x^* \rangle.
 \end{equation}
 According to (ii), $\langle f, u \rangle$ is the objective value of a feasible solution with support on $F$.
 Hence, $x^*$ could not have been optimal, and all optimal solutions have to have support on $F$.
 Clearly, all previous arguments still hold if we find a face $F'\subset F$ of $P$ such that $u \in \cone(F')$ and hence the statement follows.
\end{proof}

\section{Acronyms}

\begin{acronym}[MAGIC]\itemsep.5\baselineskip
\acro{AGF}{average gate fidelity}
\acro{AQFT}{approximate quantum fourier transform}

\acro{BOG}{binned outcome generation}

\acro{CP}{completely positive}
\acro{CPT}{completely positive and trace preserving}
\acro{CS}{compressed sensing} 

\acro{DAQC}{digital-analog quantum computing}
\acro{DFE}{direct fidelity estimation} 
\acro{DM}{dark matter}
\acro{DD}{dynamical decoupling}

\acro{EASE}{efficient, arbitrary, simultaneously entangling}

\acro{GST}{gate set tomography}
\acro{GUE}{Gaussian unitary ensemble}

\acro{HOG}{heavy outcome generation}

\acro{LDP}{Lamb-Dicke parameter}
\acro{LP}{linear program}

\acro{MAGIC}{magnetic gradient-induced coupling}
\acro{MBL}{many-body localization}
\acro{MIP}{mixed integer program}
\acro{ML}{machine learning}
\acro{MLE}{maximum likelihood estimation}
\acro{MPO}{matrix product operator}
\acro{MPS}{matrix product state}
\acro{MS}{M{\o}lmer-S{\o}rensen}
\acro{MUBs}{mutually unbiased bases} 
\acro{mw}{micro wave}

\acro{NISQ}{noisy and intermediate-scale quantum}

\acro{POVM}{positive operator valued measure}
\acro{PVM}{projector-valued measure}

\acro{QAOA}{quantum approximate optimization algorithm}
\acro{QFT}{quantum Fourier transform}
\acro{QML}{quantum machine learning}
\acro{QMT}{measurement tomography}
\acro{QPT}{quantum process tomography}
\acro{QPU}{quantum processing unit}

\acro{RDM}{reduced density matrix}

\acro{SFE}{shadow fidelity estimation}
\acro{SIC}{symmetric, informationally complete}
\acro{SPAM}{state preparation and measurement}

\acro{RB}{randomized benchmarking}
\acro{rf}{radio frequency}
\acro{RIC}{restricted isometry constant}
\acro{RIP}{restricted isometry property}

\acro{TT}{tensor train}
\acro{TV}{total variation}

\acro{VQA}{variational quantum algorithm}

\acro{VQE}{variational quantum eigensolver}

\acro{XEB}{cross-entropy benchmarking}

\end{acronym}

\bibliographystyle{./myapsrev4-2}
\bibliography{mk_no_PC,new_no_PC}
\end{document}